\documentclass[english,a4paper, cleveref, autoref, thm-restate, nolineno]{socg-lipics-v2021}

\pdfoutput=1 
\hideLIPIcs  

\title{Linear-Time \texorpdfstring{$(1+\varepsilon)$-}{}Approximation Algorithms for Two-Line-Center Problems} 

\author{Chaeyoon Chung}{Department of Computer Science and Engineering, Pohang University of Science and Technology, Pohang, South Korea.}{chaeyoon17@postech.ac.kr}{https://orcid.org/0009-0008-3363-2406}{}

\author{Anil Maheshwari}{School of Computer Science, Carleton University, Ottawa, Canada.}{anil@scs.carleton.ca}{https://orcid.org/0000-0002-1274-4598}{}

\author{Michiel Smid}{School of Computer Science, Carleton University, Ottawa, Canada.}{anil@scs.carleton.ca}{https://orcid.org/0000-0003-3683-9054}{}

\authorrunning{C. Chung, A. Maheshwari, and M. Smid} 

\Copyright{Chaeyoon Chung, Anil Maheshwari, and Michiel Smid} 



\category{} 

\funding{Chaeyoon Chung was supported by Basic Science Research Program through the
National Research Foundation of Korea (NRF) funded by the Ministry of Education(RS-2025-25436016).
Anil Maheshwari and Michiel Smid were supported by NSERC.}

\acknowledgements{This work was done while Chaeyoon Chung was visiting Carleton University.}

\nolinenumbers 

\EventEditors{Hee-Kap Ahn, Michael Hoffmann, and Amir Nayyeri}
\EventNoEds{3}
\EventLongTitle{42nd International Symposium on Computational Geometry
(SoCG 2026)}
\EventShortTitle{SoCG 2026}
\EventAcronym{SoCG}
\EventYear{2026}
\EventDate{June 2--5, 2026}
\EventLocation{New Brunswick, NJ, USA}
\EventLogo{socg-logo.pdf}
\SeriesVolume{367}
\ArticleNo{XX}     

\usepackage{amsmath}
\usepackage{amssymb}
\usepackage{graphicx}
\usepackage{multirow}
\usepackage[table]{xcolor}
\usepackage{booktabs}
\usepackage{mathtools}
\usepackage{hyperref}
\usepackage{algpseudocode}
\usepackage{algorithm}
\usepackage{xspace}
\usepackage{thmtools, thm-restate}

\newcommand{\restatemarker}{$^{\boldsymbol{*}}$}
\AddToHook{cmd/appendix/after}{\renewcommand{\restatemarker}{}}

\newtheoremstyle{star}
    {\topsep}
    {\topsep}
    {\itshape}
    {}
    {\textcolor{lipicsGray}{$\blacktriangleright$}\nobreakspace\sffamily\bfseries}
    {}
    {.5em}
    {\thmname{#1}\thmnumber{ #2\restatemarker}{\textmd{\textsf{\thmnote{(#3)}}}}.}

\newtheoremstyle{claimstar}
    {\topsep}
    {\topsep}
    {}
    {0pt}
    {\sffamily}
    {. }{5pt plus 1pt minus 1pt}%
    {$\vartriangleright$ \thmname{#1}\thmnumber{ #2\restatemarker}\thmnote{ (#3)}}

\theoremstyle{star}
\newtheorem{thm*}[theorem]{Theorem}
\newtheorem{lem*}[lemma]{Lemma}
\newtheorem{cor*}[corollary]{Corollary}
\newtheorem{obs*}[observation]{Observation}

\theoremstyle{claimstar}
\newtheorem{clm*}[claim]{Claim}

\graphicspath{ {./figs/} }

\newcommand{\eps}{\ensuremath{\varepsilon}}

\newcommand{\true}{\textsf{True}\xspace}
\newcommand{\false}{\textsf{False}\xspace}
\newcommand{\mytxt}[1]{\textup{\textbf{\textsf{\textcolor{lipicsGray}{#1}}}}}
\newcommand{\diam}{\ensuremath{\mathrm{diam}}}
\newcommand{\width}{\ensuremath{\mathrm{width}}}
\newcommand{\conv}{\ensuremath{\mathrm{conv}}}

\begin{document}

\maketitle

\begin{abstract}
Given a set $S$ of $n$ points in the plane, we study the two-line-center problem: finding two lines that minimize the maximum distance from each point in $S$ to its closest line.
We present a $(1+\eps)$-approximation algorithm for the two-line-center problem that runs in $O((n/\eps) \log (1/\eps))$ time, which improves the  previously best $O(n\log n + ({n}/{\eps^2}) \log ({1}/{\eps}) + (1/\eps^3)\log ({1}/{\eps}))$-time algorithm. 
We also consider three variants of this problem, in which the orientations of the two lines are restricted: (1)~the orientation of one of the two lines is fixed,
(2)~the orientations of both lines are fixed, and
(3)~the two lines are required to be parallel.
For each of these three variants, we give the first $(1+\eps)$-approximation algorithm that runs in linear time.
In particular, for the variant where the orientation of one of the two lines is fixed, we also give an improved exact algorithm that runs in $O(n \log n)$ time and show that it is optimal.
\end{abstract}

\section{Introduction}

Given a set $S$ of $n$ points in the plane, we study the two-line-center problem: finding two lines that minimize the maximum distance from each point in $S$ to its closest line.  
This is a special case of the $k$-line-center problem, which asks, for a given set $S$ of $n$ points in the plane, to find $k$ lines minimizing the maximum distance from each point in $S$ to its closest line.  
A more general framework is projective clustering: Given $P\subset \mathbb R^d$, the goal is to find $k$ flats of dimension $j$ that best fit $P$ under a given distance measure. 
It has wide applications in data mining~\cite{datamining1999, datamining2000}, unsupervised learning~\cite{Procopiuc2010unsupervised}, database management~\cite{Chakrabarti2000db}, and computer vision~\cite{Procopiuc2002cv}.  

The $k$-line-center problem is known to be NP-hard~\cite{Megiddo1982pc_npcomplete} when $k$ is part of the input (even in the plane), 
and various approaches including approximation algorithms and randomized methods~\cite{  pc_apx_dklogk, agarwal2005approximation, pc_shapefitting_cylinder,pc_shapefitting_box, chan2006faster,  pc_greedy_number_apx, pc_fixedkq,HASSIN199129_pc_hitting, pc_montecarlo} have been studied.  


When $k=1$, the problem becomes computing the width of a point set.  
This can be done exactly in $O(n\log n)$ time~\cite{Shamos:1978:CG:width}, and an $\Omega(n\log n)$ lower bound is known~\cite{lee1986geometric}.  
A $(1+\eps)$-approximation algorithm running in $O(n + 1/\eps)$ time was given in~\cite{chan2006faster}.

When $k=2$, i.e., the two-line-center problem, an exact algorithm with running time $O(n^2\log^2 n)$ is known~\cite{jaromczyk1995}.  
The authors in~\cite{agarwalTwo} presented a $(1+\eps)$-approximation algorithm that runs in  
$O(n\log n + (n/\eps^2)\log(1/\eps) + \eps^{-7/2}\log(1/\eps))$ time.  
Later, as the $(1+\eps)$-approximation algorithm for the width problem was improved from  
$O(n + \eps^{-3/2})$ to $O(n + 1/\eps)$~\cite{chan2006faster},  
their algorithm is consequently improved to $O(n\log n + ({n}/{\eps^2}) \log ({1}/{\eps}) + (1/\eps^3)\log ({1}/{\eps}))$, even though it is not explicitly mentioned. 


Observe that the two-line-center problem is equivalent to finding a pair of slabs whose union covers the input point set while minimizing their maximum width. 

We study the two-line-center problem and three natural variants.
Let $\mu = \min\{n,\, 1/\eps^2\}$, and let $w^*$ denote the maximum width of an optimal pair of slabs in each problem. 
Each variant is illustrated in Figure~\ref{fig:problems}.
Table~\ref{table:result} summarizes the previous and new results.

\begin{figure}[t!]
\centering
\includegraphics[width=1 \textwidth]{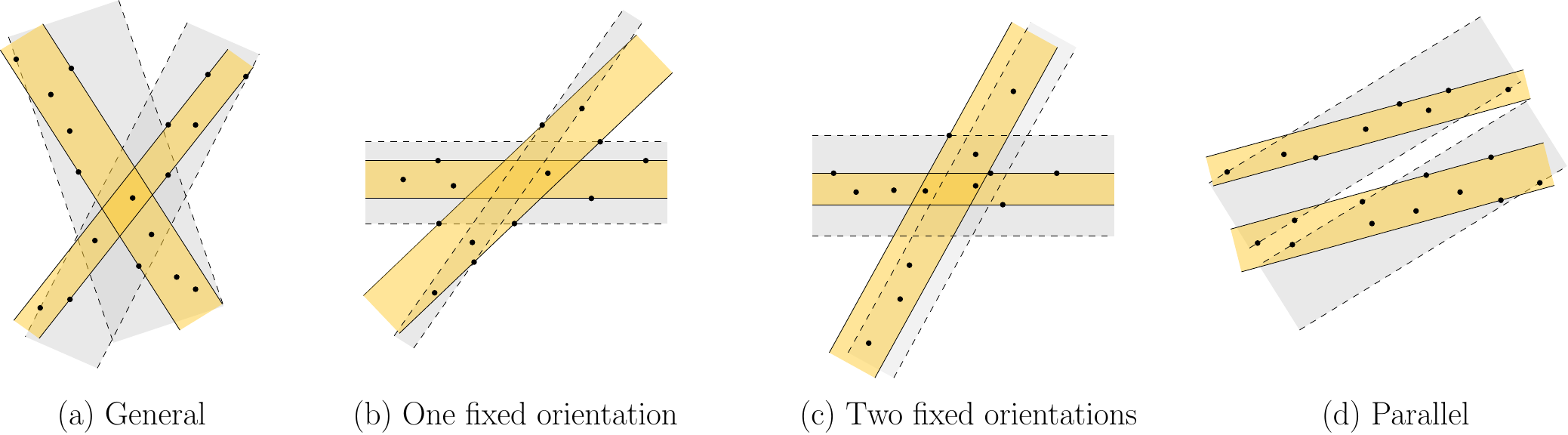}
\caption{In each case, both yellow and gray pairs satisfy the restriction, but the yellow pair has the smaller maximum width.}
\label{fig:problems}
\end{figure}

\begin{enumerate}
    \item \textbf{\textsf{General:}}
    Given a set $S$ of $n$ points in the plane, we study the problem of finding a pair of slabs whose union covers $S$ while minimizing the maximum of their widths. 
    We present an $O(n + (\mu/\eps)\log \mu )$-time algorithm that finds a pair of slabs whose union covers $S$ and whose maximum width is at most $(1+\eps)w^*$. The running time can be expressed as both $O((n/\eps)\log(1/\eps))$ and $O(n + (1/\eps^3)\log(1/\eps))$. 

    \item \textbf{\textsf{One Fixed Orientation:}} 
    Given a set $S$ of $n$ points in the plane and an orientation $\theta$, we study the problem of finding a pair of slabs such that one of them has orientation $\theta$ and their union covers $S$ while minimizing the maximum of their widths. 
    We present an exact algorithm that runs in $O(n\log n)$ time, improving the previously best $O(n\log^3 n)$-time algorithm~\cite{ahn2024constrained}, and show that this is optimal. 
    
    ~$\quad$ We also present an $O(n+ \mu \log \mu)$-time algorithm that finds a pair of slabs such that one of them has orientation $\theta$, their union covers $S$, and their maximum width is at most $(1+\eps)w^*$.
    The running time can be expressed as $O(n\log (1/\eps))$ and $O(n + 1/\eps^2 \log (1/\eps))$.

    \item \textbf{\textsf{Two Fixed Orientations:}}
    Given a set $S$ of $n$ points in the plane and two orientations, $\theta_1$ and $\theta_2$, we study the problem of finding two slabs, one has orientation $\theta_1$ and the other has orientation $\theta_2$, whose union covers $S$ while minimizing the maximum of their widths. 
    We present an $O(n+ 1/\eps)$-time algorithm that finds a pair of slabs such that 
    one has orientation $\theta_1$ and the other has orientation $\theta_2$, their union covers $S$, and their maximum width is at most $(1+\eps)w^*$.

    \item \textbf{\textsf{Parallel:}}
    Given a set $S$ of $n$ points in the plane, we study the problem of finding a pair of parallel slabs whose union covers $S$ while minimizing the maximum of their widths. 
    We present an $O(n + \mu / \eps)$-time algorithm that finds a pair of parallel slabs such that their union covers $S$ and their maximum width is at most $(1+\eps)w^*$. 
\end{enumerate}









\begin{table}[t]
\begin{tabular} {p{.25\textwidth} m{.22\textwidth} m{.43\textwidth}} 
\toprule
~ & \textbf{Exact} & \textbf{Approximation}  \\ 
\midrule

\multirow{2}{*}{General} & \multirow{2}{*}{$O(n^2\log^2n) $~\cite{jaromczyk1995}} & $O(n\log n + \frac{n}{\eps^2} \log \frac{1}{\eps} + \frac{1}{\eps^3}\log \frac{1}{\eps})$~\cite{agarwalTwo, chan2006faster}\\
\cmidrule(lr){3-3} 
& & 
 $\boldsymbol{O(\frac{n}{\eps} \log \frac{1}{\eps}) } $$,\:$ $\boldsymbol{O(n + \frac{1}{\eps^3} \log \frac{1}{\eps}) }$ (Sec.~\ref{sec:general}) 
\\
\cmidrule(r){1-1} \cmidrule(l){2-3}

\multirow{2}{*}{One fixed orientation} &  $O(n  \log^3 n)$~\cite{ahn2024constrained}  & - \\
\cmidrule(lr){2-2}  \cmidrule(lr){3-3} 
& $\boldsymbol{O(n \log n)}$ (Sec.~\ref{sec:one_fixed})
& $\boldsymbol{O(n \log \frac{1}{\eps})}$$,\:$ $\boldsymbol{O(n + \frac{1}{\eps^2} \log \frac{1}{\eps}) }$  (Sec.~\ref{sec:one_fixed})   \\
\cmidrule(r){1-1} \cmidrule(l){2-3} 

\multirow{2}{*}{Two fixed orientations} & \multirow{2}{*}{$O(n \log n)$~\cite{ahn2024constrained} } & - \\
\cmidrule(lr){3-3} 
& & 
$\boldsymbol{O(n + \frac{1}{\eps})}$  (Sec.~\ref{sec:two_fixed})  \\
\cmidrule(r){1-1} \cmidrule(l){2-3} 

\multirow{2}{*}{Parallel} & \multirow{2}{*}{$O(n^2)$~\cite{bae2020minimum}} & - \\
\cmidrule(lr){3-3} 
& & 
 $\boldsymbol{O(\frac{n}{\eps})}$$,\:$  $\boldsymbol{O(n + \frac{1}{\eps^3})}$ (Sec.~\ref{sec:parallel})  \\
\bottomrule
\end{tabular}
\caption{A summary of our results and previous work.}
\label{table:result}
\end{table}

\newpage
\subsection{Outline}
In Section~\ref{sec:preliminaries}, we introduce the notation that will be used throughout the paper.

In Section~\ref{sec:polyeps}, we introduce three key ingredients used throughout the paper.
In Section~\ref{sec:anchor_pair}, we modify the concept of an \emph{anchor pair} (as introduced in~\cite{agarwalTwo}) and show how to compute it.
In Section~\ref{sec:general_const}, we use an anchor pair to compute a $10$-approximation for the general two-line-center problem in $O(n)$ time.
In Section~\ref{sec:eps_certificate}, we recall the concept of an \emph{$\eps$-certificate} (as defined in~\cite{agarwal2005approximation}), a subset of the input point set that suffices to obtain a $(1+\eps)$-approximate solution, and show how to compute it in $O(n)$ time using an anchor pair from Section~\ref{sec:anchor_pair} and a $10$-approximate solution from Section~\ref{sec:general_const}.

In Section~\ref{sec:one_fixed}, we address the one-fixed-orientation two-line-center problem.
In Section~\ref{sec:one_fixed_decision}, we give an $O(n)$-time decision algorithm, assuming the points are sorted.
In Section~\ref{sec:one_fixed_optimization}, we present an $O(n\log n)$-time exact algorithm and prove its optimality in Appendix~\ref{sec:lowerbound}.
In Section~\ref{sec:one_fixed_approximation}, we give a $(1+\eps)$-approximation algorithm.

In Section~\ref{sec:general}, we address the general two-line-center problem.
Using an anchor pair~(Section~\ref{sec:anchor_pair}) and a $10$-approximate solution (Section~\ref{sec:general_const}), we compute a candidate set of~orientations for one of the slabs of a $(1+\eps)$-approximate solution.
Combining this with~the~algorithm from Section~\ref{sec:one_fixed} and an $\eps$-certificate (Section~\ref{sec:eps_certificate}), we obtain a $(1+\eps)$-approximation algorithm.

In Section~\ref{sec:two_fixed}, we address the two-fixed-orientations two-line-center problem.
We first give a $2$-approximation algorithm in Section~\ref{sec:2_apx_two_fixed}. 
Based on this, we present a $(1+\eps)$-approximation algorithm in Section~\ref{sec:eps_apx_two_fixed}.

In Section~\ref{sec:parallel}, we address the parallel two-line-center problem.
We use the \emph{gap-ratio} of a pair of slabs (as defined in~\cite{chung2025parallel}) and handle two cases separately: when the gap-ratio of an optimal pair is small (Section~\ref{sec:small_gap_ratio}) and when it is large (Section~\ref{sec:large_gap_ratio}).
Combining these with an $\eps$-certificate (Section~\ref{sec:eps_certificate}), we present a $(1+\eps)$-approximation algorithm.

\section{Preliminaries}
\label{sec:preliminaries}
For a set $P$ of points in the plane, we denote its diameter by $\diam (P)$, its width by $\width(P)$, and its convex hull by $\conv(P)$.
If $P = \emptyset$, then $\width(P) = 0$.
For a slab $\sigma$, we denote its width by $w(\sigma)$, which is the distance between the two bounding lines of $\sigma$ (see Figure~\ref{fig:preliminaries}(a)).

For a line $\ell$, we say that $\ell$ has orientation $\theta \in [0,\pi)$ if $\theta$ is the counterclockwise angle from the $x$-axis to $\ell$.
A slab $\sigma$ is said to have orientation $\theta \in [0,\pi)$ if its two bounding lines have orientation $\theta$ (see Figure~\ref{fig:preliminaries}(a)).
For a finite set $P$ of points in the plane, let $\sigma_\theta(P)$ denote the minimum-width slab of orientation $\theta$ that encloses $P$, and define $\width_\theta(P) = w(\sigma_\theta(P))$ (see Figure~\ref{fig:preliminaries}(b)).

\begin{figure}[ht]
\centering
\includegraphics[height=3.7cm,keepaspectratio]{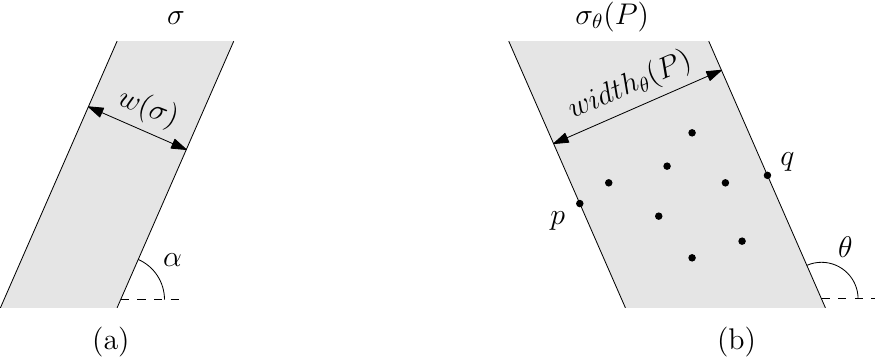}
\caption{(a) The orientation of $\sigma$ is $\alpha$. (b) For a set $P$ of points and an orientation $\theta$, $\sigma_{\theta}(P)$ and $\width_{\theta}(P)$ are shown. Observe that $(p, q)$ is an antipodal pair of $P$ with respect to $\theta$. }
\label{fig:preliminaries}
\end{figure}

For a set $P$ of points in the plane and an orientation $\theta \in [0,\pi)$, we say that a pair $(p,q)$ with $p,q \in P$ is an \emph{antipodal pair} of $P$ with respect to $\theta$ if $p$ lies on one bounding line of $\sigma_\theta(P)$ and $q$ lies on the other (see Figure~\ref{fig:preliminaries}(b)).

For two finite sets $A$ and $B$ in the plane, let $d(A,B)$ denote the minimum distance between a point of $A$ and a point of $B$.
For two points $p$ and $q$ in the plane, let $\ell_{pq}$ be the line passing through $p$ and $q$.

For each of the four problems, we assume that $w^*>0$, where $w^*$ denotes the maximum width of an optimal pair of slabs for that problem. The case where $w^*=0$ can be detected and handled in linear time. See Appendix~\ref{sec:w_0}.

Throughout the paper, let $\eps > 0$ be a sufficiently small real number and let $S$ denote an input set of $n$ points in the plane.

\section{Ingredients}
\label{sec:polyeps}

In this section, we describe three key ingredients used throughout the paper: (1) an anchor pair, (2) a linear-time $10$-approximation algorithm, and (3) an $\eps$-certificate.



Throughout this section, we denote by $\Sigma^* = (\sigma_1^*, \sigma_2^*)$ an optimal pair of slabs for the general two-line-center problem, and let $w^* = \max \{w(\sigma_1^*), w(\sigma_2^*) \}$.

\subsection{An Anchor Pair and Its Candidate Set}
\label{sec:anchor_pair}
We first define an anchor pair and show that a candidate set for an anchor pair of $\Sigma^*$ of constant size can be computed in $O(n)$ time.
Recall that $S$ denotes the input set of $n$ points in the plane.

\begin{definition}[Anchor pair]
    For a pair of slabs $\Sigma = (\sigma_1, \sigma_2)$,
    we call a pair of points $p, q \in \mathbb R^2$ an \textbf{anchor pair of $\Sigma$} if 
    \begin{enumerate}[(1)]
        \item $p, q\in \sigma_1$ and $d(p, q) \ge \diam((\sigma_1 \setminus \sigma_2) \cap S)/4$, or 
        \item $p, q\in \sigma_2$ and $d(p, q) \ge \diam((\sigma_2 \setminus \sigma_1) \cap S)/4$.
    \end{enumerate}
     Note that $p$ or $q$ may lie in $\sigma_1\cap \sigma_2$. (See Figure~\ref{fig:anchor_ex}.)
\end{definition}

\begin{figure}[ht]
\centering
\includegraphics[height=4.1cm,keepaspectratio]{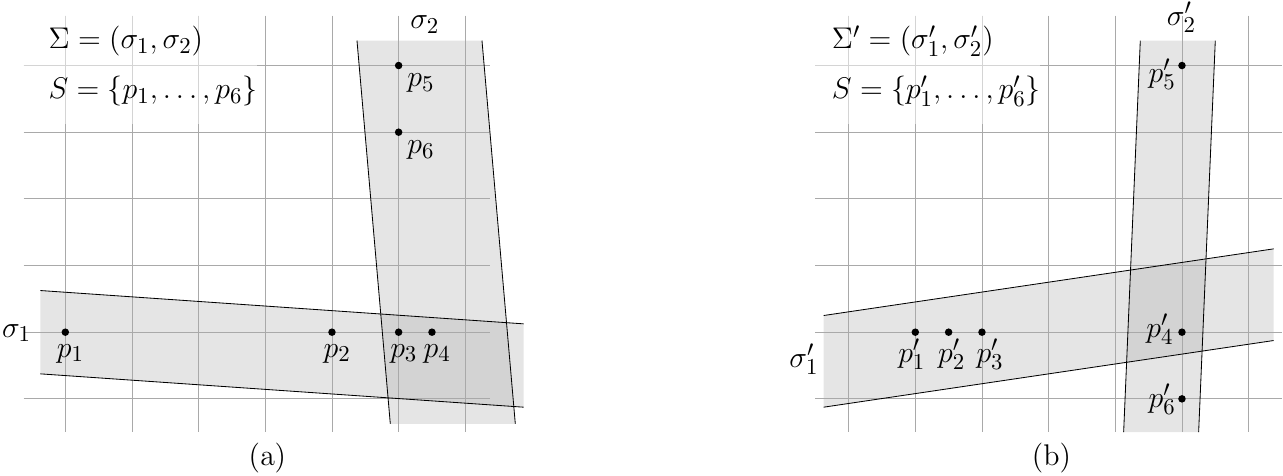}
\caption{
(a) Both $(p_2, p_3)$ and $(p_3, p_4)$ are anchor pairs of $\Sigma$, satisfying \mytxt{(1)} and \mytxt{(2)}, respectively.
(b) $(p_1', p_2')$ is an anchor pair of $\Sigma'$, satisfying \mytxt{(1)}, but $(p_4', p_6')$ is not.
}
\label{fig:anchor_ex}
\end{figure}

We remark that the authors in~\cite{agarwalTwo} used the concept of an anchor pair, defined slightly differently.  
They showed that a candidate set of constant size can be computed in $O(n \log n)$ time. In the proof of the following lemma, we point out a gap in their proof.  
By adopting~our~definition of an anchor pair, we both resolve this gap and obtain its candidate set in $O(n)$~time.

\begin{lemma}
    \label{lem:anchor_pair}
    We can compute, in $O(n)$ time, a set of at most $11$ pairs of points such that for any pair $\Sigma$ of slabs whose union covers $S$, at least one of them is an anchor pair of $\Sigma$.
\end{lemma}

\begin{proof}
    Let $\Sigma = (\sigma_1, \sigma_2)$ be a pair of slabs whose union covers $S$.
    Let $\Delta$ denote the diameter of $S$. 
    Let $p$ and $q$ be a pair of points in $S$ satisfying $d(p, q) \ge \Delta/2$, which can be found in $O(n)$ time.
    Let $R = d(p, q)/2$. 
    We consider two open disks, $D_p$ and $D_q$, of radius $R$ centered at $p$ and $q$, respectively. 
    That is $D_p \cap D_q = \emptyset$.

    If there exists $r\in S\setminus (D_p\cup D_q)$, 
    we output $F = \{(p,q ), (p, r), (q, r) \}$. See Figure~\ref{fig:anchor_pair}(a). 
    Observe that $p$, $q$, and $r$ are mutually separated by a distance of at least $R \geq \Delta/4$, and both $\diam((\sigma_1 \setminus \sigma_2) \cap S)$ and $\diam((\sigma_2 \setminus \sigma_1) \cap S)$ are at most $\Delta$.
    Since at least one pair in $F$ must lie within the same slab, $F$ necessarily contains an anchor pair of $\Sigma$.

    Now we consider the remaining case where $S\subset (D_p \cup D_q)$. 
    Let $P = S\cap D_p$ and $Q = S\cap D_q$. 
    Let $\conv(P)$ and $\conv(Q)$ be the convex hulls of $P$ and $Q$, respectively. Note that these hulls do not intersect. 
    Let $\ell_1$ and $\ell_2$ be the two lines that are inner common tangents to $\conv(P)$ and $\conv(Q)$. 
    Let $p_1 \in P$ (resp., $p_2 \in P$) and $q_1 \in Q$ (resp., $q_2 \in Q$) be the points lying on $\ell_1$ (resp., $\ell_2$). 
    See Figure~\ref{fig:anchor_pair}(b).
    We can compute $p_1, p_2, q_1$, and $q_2$ without computing the convex hulls explicitly by linear programming in $O(n)$ time~\cite{Dyer1984LinearTime_LP, Megiddo1984LinearTime_LP}. 
    Let $(p_3, p_4)$ be a pair of points in $P$ such that $d(p_3, p_4) \ge \diam(P)/2$. 
    Let $(q_3, q_4)$ be a pair of points in $Q$ such that $d(q_3, q_4) \ge \diam(Q)/2$. We can compute $p_3, p_4, q_3$, and $q_4$ in $O(n)$ time. 
    Then we output $F = \{(p, q), (p_3, p_4), (q_3, q_4) \} \cup \bigcup_{i=1}^4 (p, q_i)  \cup \bigcup_{i=1}^4 (q, p_i)$. Now we show that $F$ contains at least one anchor pair of $\Sigma$.

    Assume to the contrary that no pair of $F$ is an anchor pair of $\Sigma$. 
    Let $S_{12} = S\cap (\sigma_1 \setminus \sigma_2)$ and $S_{21} = S\cap (\sigma_2 \setminus \sigma_1)$.
    The assumption implies that $S_{12}$ (resp., $S_{21}$) contains either $p$ or $q$ but not both. 
    Without loss of generality, let $p\in S_{12}$ and $q \in S_{21}$. 
    For every $1\le i \le 4$, since both $d(p, q_i)$ and $d(q, p_i)$ are at least $R \ge \Delta/4$, the assumption also implies~$p_i \in S_{12}$~and~$q_i \in S_{21}$.


    We first claim that $P \cap S_{21} \neq \emptyset$ and $Q \cap S_{12} \neq \emptyset$ under the assumption.
    Suppose that $P \cap S_{21} = \emptyset$. Then $S_{21} \subseteq Q$ and thus $\diam(S_{21}) \le \diam(Q)$.  
    Since $\diam(S_{21})/2 \le \diam(Q)/2 \le d(q_3, q_4)$ and $q_3, q_4 \in S_{21}$, $(q_3, q_4)$ is an anchor pair of $\Sigma$, contradicting the assumption\footnote{
    In the proof of~\cite{agarwalTwo}, the authors claimed that  $P \cap S_{21} \neq \emptyset$ and $Q \cap S_{12} \neq \emptyset$ in a similar way. 
    However, under their definition of an anchor pair, it is not clear that the assumption $P \cap S_{21} = \emptyset$ (resp., $Q \cap S_{12} = \emptyset$) leads to a contradiction,
    because $\diam(Q)$ (resp., $\diam(P)$) can be much smaller than $\diam(S \cap \sigma_2)$ (resp., $\diam(S \cap \sigma_1)$).}.
    A similar contradiction occurs if we assume $Q \cap S_{12} = \emptyset$.

    Therefore there exist two points $p'$ and $q'$ such that $p' \in P \cap S_{21}$ and $q' \in Q \cap S_{12}$.
    Let $s$ be the intersection of $\ell_1$ and $\ell_2$. 
    Observe that $\triangle sq_1 q_2 \subset \triangle p'q_1 q_2 \subset \sigma_2$ as $p' \in P$ and $p', q_1, q_2 \in \sigma_2$. 
    As $q'$ should not be inside $\sigma_2$, $q'$ should lie outside of $\triangle s q_1 q_2$. 
    Therefore $\triangle p_1 p_2 q'$ and $\overline{q_1 q_2}$ intersect. See Figure~\ref{fig:anchor_pair}(c).
    As $p_1, p_2, q' \in S_{12}$ and $q_1, q_2 \in S_{21}$, it is implied that  $\sigma_1$ separates $q_1$ and $q_2$.
    In a symmetric way, we can show that $\sigma_2$ separates $p_1$ and $p_2$. 
    This implies that the segments $\overline{p_1 p_2}$ and $\overline{q_1 q_2}$ intersect (see Figure~\ref{fig:anchor_pair}(d)). But it contradicts that $D_p \cap D_q = \emptyset$. 
\end{proof}

\begin{figure}[ht]
\centering
\includegraphics[width=0.68 \textwidth]{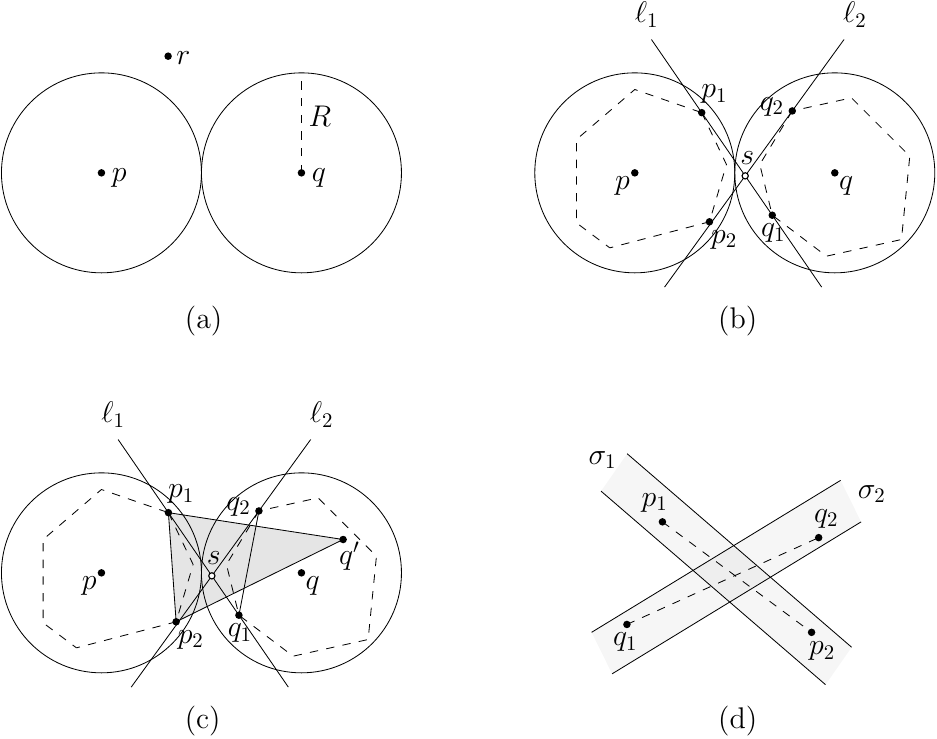}
\caption{(a) There exists a point $r \in S \setminus (D_p \cup D_q)$. 
(b) The inner common tangents $\ell_1$ and $\ell_2$ of ${conv}(P)$ and ${conv}(Q)$, and the points $p_1, p_2, q_1, q_2$, and $s$ are shown. 
(c) The triangle $\triangle p_1 p_2 q'$ and the segment $\overline{q_1 q_2}$ intersect. 
(d) The segments $\overline{p_1 p_2}$ and $\overline{q_1 q_2}$ intersect.
}
\label{fig:anchor_pair}
\end{figure}

\subsection{A Linear-Time \texorpdfstring{$10$}{10}-Approximation Algorithm}
 
In this section, we give a $10$-approximation algorithm for the general two-line-center problem that runs in $O(n)$ time. 
We remark that the authors in~\cite{agarwalTwo} gave a $6$-approximation algorithm that runs in $O(n\log n)$ time.

\label{sec:general_const}
By Lemma~\ref{lem:anchor_pair}, we can find a set of at most $11$ pairs of points, at least one of which is an anchor pair of $\Sigma^*$ in $O(n)$ time. 
For the moment, assume that we have an anchor pair $(p, q)$ of $\Sigma^*$. 
Without loss of generality, assume that $p, q \in \sigma_1^*$ and $d(p, q) \ge \diam((\sigma_1^*\setminus \sigma_2^* ) \cap S)/4$. 

Let ${p_1, p_2, \ldots, p_{n}}$ be the points of $S$, indexed in increasing order of their distance from $\ell_{pq}$. (We do not actually sort them. The indices are used only for reference.)
For three points $\alpha, \beta, \gamma \in \mathbb R^2$, 
let $\sigma(\alpha, \beta ;\gamma)$ denote the slab whose center line is $\ell_{\alpha \beta}$ and one bounding line passes through $\gamma$ (see Figure~\ref{fig:sigma_abc}).

\begin{figure}[ht]
\centering
\includegraphics[width=0.34 \textwidth]{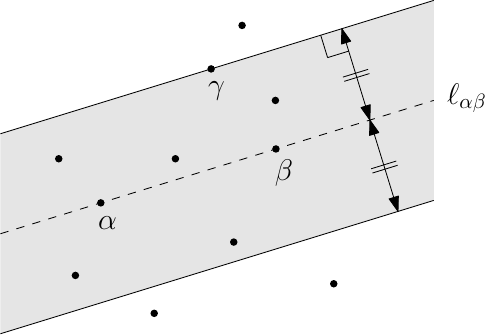}
\caption{For three points $\alpha, \beta, \gamma \in \mathbb R^2$, the gray region represents $\sigma(\alpha, \beta ;\gamma)$. 
}
\label{fig:sigma_abc}
\end{figure}

The following lemma implies the existence of a pair of slabs whose maximum width is at most $10w^*$, and one of which is parallel to $\ell_{pq}$.

\begin{lemma}
    There exists a point $r\in S$ such that $\sigma(p, q;r)$ covers all points of $(\sigma_1^*\setminus \sigma_2^* ) \cap S $ and $w(\sigma(p, q;r)) \le 10 w^*$. 
    For such $r$, $\width(S \setminus \sigma(p, q; r))\le w^*$.
    \label{lem:10apx_existence}
\end{lemma}

\begin{proof}
    We prove that for every $s\in (\sigma_1^* \setminus \sigma_2^*)\cap S$, it holds that $d(\ell_{pq}, s) \le 5w^*$. It directly implies that there exists a slab $\tau$ of width at most $10w^*$ that has $\ell_{pq}$ as its center line and covers $(\sigma_1^*\setminus \sigma_2^* ) \cap S$. 
    Observe that we can choose $r$ as the first point in $S$ touched by one of the bounding lines of $\tau$ as we move them toward $\ell_{pq}$ at the same speed.

    Without loss of generality, we assume that $p$ and $q$ lie on a horizontal line. 
    Let $D = \diam((\sigma_1^* \setminus \sigma_2^*) \cap S)$. 
    Therefore, $d(p, q) \ge D/4$. 
    Let $\alpha$ denote the orientation of $\sigma_1^*$. 
    For simplicity, we assume that $0\le \alpha < \pi/2$. The other case can be shown in a symmetric way. 
    Since $p, q \in \sigma_1^*$, we have $\sin \alpha \le w^*/ d(p, q) \le  4w^*/D$. See the right triangle having $\overline{pq}$ as one of its sides in Figure~\ref{fig:5apx}.
    
    Let $B$ be the minimum bounding box of $(\sigma_1^* \setminus \sigma_2^*) \cap S$ with orientation $\alpha$.
    Observe that every side of $B$ has length at most $D$. 
    Let $o$ and $t$ denote the lowest and highest vertices of $B$, respectively, which lie at diagonally opposite corners.
    Consider a horizontal slab $\sigma$ whose two bounding lines pass through $o$ and $t$, and let $w = w(\sigma)$.  
    Clearly, $(\sigma_1^* \setminus \sigma_2^*) \cap S \subset \sigma$. 
    We now show that $w \le 5w^*$, which completes the proof.

    Let $b$ be a vertex of $B$ other than $o$ and $t$. 
    Consider the vertical line $\ell$ that passes through $b$. Then $\ell$ intersects the bounding lines of $\sigma$. 
    Observe that one intersection point lies inside $\sigma_1^*$, while the other lies outside. 
    Let $a$ and $c$ denote these two intersection points, with $a$ being the one inside $\sigma_1^*$. 
    Then $w = \overline{ab} + \overline{bc} \le \overline{tb} + \overline{bc} \le w^* + \overline{bc}$.
     We have $\overline{bc} = \overline{ob}\cdot \sin \alpha 
    \le D \cdot \frac{4w^*}{D} \le 4w^*$.
    Therefore, $w \le 5w^*$. 
\end{proof}

\begin{figure}[ht]
\centering
\includegraphics[width=0.55 \textwidth]{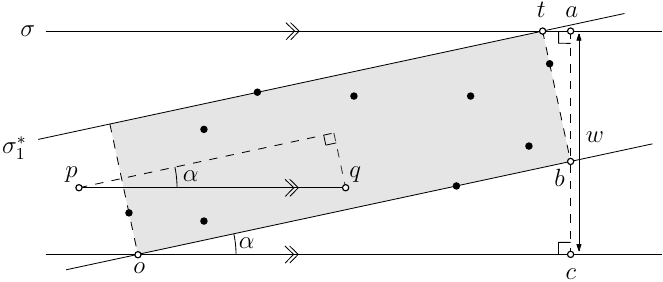}
\caption{Black points indicate the set $(\sigma_1^* \setminus \sigma_2^*) \cap S$.
The orientation of $\sigma_1^*$ is denoted by $\alpha$.
The gray region represents $B$, the minimum bounding box of $(\sigma_1^* \setminus \sigma_2^*) \cap S$ with orientation $\alpha$.
The anchor pair $(p, q)$ and the points $o, t, a, b, c$ are shown. Observe that $w(\sigma) = \overline{ab} + \overline{bc}$.
}
\label{fig:5apx}
\end{figure}

\paragraph*{Naive Approach and Idea}\label{par:naive}

By Lemma~\ref{lem:10apx_existence}, there exists a point $r \in S$ such that both $w(\sigma(p, q; r))$ and $\width(S \setminus \sigma(p, q; r))$ are at most $10w^*$.
We now give a brief idea of using a binary search to find such a point $r$.

Consider the slab $\sigma(p, q; p_i)$ and the set $S \setminus \sigma(p, q; p_i)$ for $1 \le i \le n$. 
Observe that as $i$ increases, the width $w(\sigma(p, q; p_i))$ increases, while $\width(S \setminus \sigma(p, q; p_i))$ decreases\footnote{Throughout this section, we refer to a non-decreasing (resp., non-increasing) function simply as an increasing (resp., decreasing) function.}. 
Using this monotonicity, we can find an index $i$ that minimizes the unimodal function $\max\{ w(\sigma(p, q; p_i)), \width(S \setminus \sigma(p, q; p_i)) \}$.

More specifically, at each step of the binary search over a range of indices $[s, e] \subseteq [1, n]$, 
we compute the median index $m$ and evaluate $w(\sigma(p, q; p_m))$ and $\width(S \setminus \sigma(p, q; p_m))$. 
If $w(\sigma(p, q; p_m)) \ge \width(S \setminus \sigma(p, q; p_m))$, we set the range to $[s, m]$. 
Otherwise, we set it to $[m, e]$. 
Let $j$ be the index obtained by the binary search. 
Observe that $j$ minimizes $\max\{ w(\sigma(p, q; p_i)), \width(S \setminus \sigma(p, q; p_i)) \}$ over all $1\le i \le n$.
By Lemma~\ref{lem:10apx_existence}, we have $\max\{ w(\sigma(p, q; p_j)), \width(S \setminus \sigma(p, q; p_j)) \} \le 10w^*$.

Note that the time required in each step is dominated by computing $\width(S \setminus \sigma(p, q; p_m))$, which takes $O(n \log n)$ time~\cite{Shamos:1978:CG:width}.
Therefore, the overall running time is $O(n \log^2 n)$.

To reduce the time complexity, we avoid computing ${width}(S \setminus \sigma(p, q; p_m))$ from scratch at every step by sacrificing exactness.  
We remark that the authors in~\cite{agarwalTwo} presented an $O(n \log n)$-time 
binary search algorithm by reducing the time for each step to $O(n)$. 
We present an $O(n)$-time binary search algorithm in which each step 
takes time linear in the size of the current range of indices, even when the points are not sorted.
Our algorithm uses the following incremental algorithm that computes a $6$-approximation of the width of a point set.

\begin{restatable}[Theorem A.3 of \cite{chan2006faster}]{theorem}{sixapx}
    \label{thm:6apx}
    Given an online sequence of points in the plane, we can maintain a $6$-approximation of the width with $O(1)$ space and update time.
\end{restatable}

We show the following theorem. See Appendix~\ref{sec:linear_10apx} for the proof and full details.

\begin{restatable}{theorem}{tenapx}
    \label{thm:10apx}
    In $O(n)$ time, we can compute a pair of slabs that forms a $10$-approximation for the general two-line-center problem.
\end{restatable}

\subsection{An \texorpdfstring{$\eps$}{epsilon}-certificate}
\label{sec:eps_certificate}

Before defining an $\eps$-certificate, we first define the notion of an \emph{$\eps$-expansion}.
We remark that the definitions in this section are from~\cite{agarwal2005approximation}.

\begin{definition}[$\eps$-expansion]
\begin{enumerate}[(i)]
    \item For an interval $[a, b]$, 
    we define its \emph{$\eps$-expansion} to be the interval $[a- \eps(b-a)/2, b + \eps(b-a)/2]$. See Figure~\ref{fig:expansion}(a).
    \item For a slab centered on a line $\ell$ with width $w$,
    we define its \emph{$\eps$-expansion} to be the slab centered on the line $\ell$ having width  $(1+ \eps) w$.
    See Figure~\ref{fig:expansion}(b).
\end{enumerate}
\end{definition}

\begin{figure}[!ht]
\centering
\includegraphics[width=0.6 \textwidth]{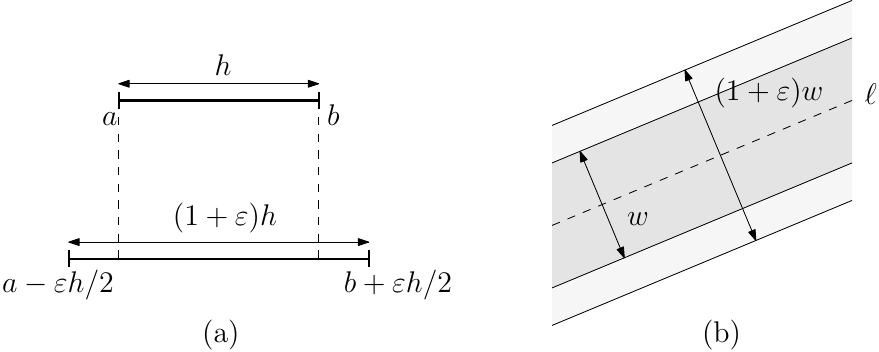}
\caption{(a) The $\eps$-expansion of an interval. 
(b) The $\eps$-expansion of a slab. }
\label{fig:expansion}
\end{figure}

Now we define an $\eps$-certificate in both $\mathbb R^1$ and $\mathbb R^2$. 

\begin{definition}[$\eps$-certificate]
\begin{enumerate}[(i)]
    \item Let $P$ be a set of $n$ points on the real line and $\eps > 0$. 
    We say that a subset $Q\subseteq P$ is an \emph{$\eps$-certificate} of $P$ 
    if for any pair of intervals whose union covers $Q$, 
    the union of their $\eps$-expansions covers $P$. 
    \item Let $P$ be a set of $n$ points in the plane  and $\eps > 0$. 
    We say that a subset $Q\subseteq P$ is an \emph{$\eps$-certificate} of $P$ 
    if for any pair of slabs with the same width whose union covers $Q$, 
    the union of their $\eps$-expansions covers $P$.
    We emphasize that the two slabs must have the same width. 
\end{enumerate}
\end{definition}

We note that the existence of an $\eps$-certificate of size $1/\eps^{O(1)}$ was shown in~\cite{agarwal2005approximation}. 
By proving the following theorem, we refine the hidden constant and show how to construct such an $\eps$-certificate in linear time. 

\begin{theorem}
    \label{thm:computing_eps_certificate}
    For a set $S$ of $n$ points in the plane and $\eps > 0$, 
    an $\eps$-certificate of $S$ of size $O(1/\eps^2)$ can be computed in $O(n)$ time. 
\end{theorem}

\begin{proof}
If $n \leq 1/\eps^2$, then we simply take $S$ itself as an $\eps$-certificate of size $O(1/\eps^2)$.
Therefore, we only consider the case where $n > 1/\eps^2$. 

We first show how to compute an~$\eps$-certificate in the one-dimensional space, and then use it to handle the two-dimensional~case.

\paragraph*{One-dimensional case}

Let $P$ be a set of $m$ points on the real line. 
We show that an $\eps$-certificate of $P$ of size $O(1/\eps)$ can be computed in $O(m)$ time. 
If $m \le 1/\eps$, then we simply take $P$ itself as an $\eps$-certificate of size $O(1/\eps)$. 
Therefore, we only consider the case where $m > 1/\eps$.
Let $I$ be the interval that has the leftmost and rightmost points of $P$ as endpoints. 
Let $D$ denote the length of $I$. 
We divide $I$ into $\lceil 4/\eps \rceil$ intervals of length at most $ \eps D / 4 $. 
Let $a_0 < a_1 < \ldots < a_{\lceil 4/\eps \rceil}$ denote the endpoints of these intervals (thus, $I = [a_0, a_{\lceil 4/\eps \rceil}])$. 
For every $i$, $0 \le i < \lceil 4/\eps \rceil$, 
we find the leftmost and rightmost points in $P \cap [a_i, a_{i+1}]$ and let $R$ be the union of them (see Figure~\ref{fig:certificate_1d_ex}).
Observe that $R$ always contains the leftmost and rightmost points of $P$, and $|R| = O(1/\eps)$.
By using the floor function, we can determine the interval to which each input point belongs in constant time. 
Therefore, $R$ can be computed in $O(m + 1/\eps) = O(m)$ time.

\begin{figure}[ht]
\centering
\includegraphics[width=0.7 \textwidth]{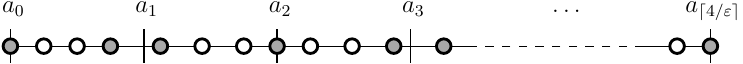}
\caption{ For a set $P$ of points on the real line, the points in $R$ are shown in gray. }
\label{fig:certificate_1d_ex}
\end{figure}

\begin{claim*}
    $R$ is an $\eps$-certificate of $P$.
\end{claim*}
\begin{claimproof}
Let $(I_1, I_2 )$ be a pair of intervals whose union covers $R$. Without loss of generality, assume that the two intervals are interior-disjoint and that $I_1$ lies to the left of $I_2$.
We first consider the case where both $I_1$ and $I_2$ have lengths less than $D/2$.
Then there exists $a_i$ that is contained in neither $I_1$ nor $I_2$ (see Figure~\ref{fig:certificate_1d}(a)). In other words, $I_1$ is strictly to the left of $a_i$, and $I_2$ is strictly to the right of $a_i$.
Let $R_1 = R \cap [a_0, a_i]$ and $R_2 = R \cap [a_i, a_{\lceil 4/\eps \rceil}]$.
As $I_1\cup I_2$ covers $R$, we have $R_1 \subset I_1$ and $R_2 \subset I_2$. 
Since the leftmost and rightmost points of $P \cap [a_0, a_i]$ are in $R_1$ and
the leftmost and rightmost points of $P \cap [a_i, a_{\lceil 4/\eps \rceil}]$ are in $R_2$,
$I_1 \cup I_2$ covers $P$. (In this case, no expansion is needed.)

Now, consider the case where $I_1$ or $I_2$ has length at least $D/2$. Without loss of generality, assume that $|I_1| \ge D/2$. 
Let $j$ be the smallest index such that $a_j \notin I_1$ (see Figure~\ref{fig:certificate_1d}(b)). Note that $I_1$ is strictly left to $a_j$. 
Since the leftmost and rightmost points of $P \cap [a_j, a_{\lceil 4/\eps \rceil}]$ are in $R$, 
$I_2$ contains all the points in $P \cap [a_j, a_{\lceil 4/\eps \rceil}]$. 
Now it suffices to show that the $\eps$-expansion of $I_1$ includes $[a_0, a_j)$.
Observe that the right endpoint of $I_1$ lies in $[a_{j-1}, a_j)$. 
By the $\eps$-expansion, the right endpoint shifts by $\eps |I_1|/2$ to the right. By our assumption, we have $\eps |I_1|/2 \ge \eps D / 4$.
By definition, $a_{i} - a_{i-1} \le \eps D/4$ for every $i = 1, \ldots, \lceil 4/\eps \rceil$. 
Therefore the $\eps$-expansion of $I_1$ contains $[a_0, a_j]$. 
\end{claimproof}

\begin{figure}[ht]
\centering
\includegraphics[width=0.93 \textwidth]{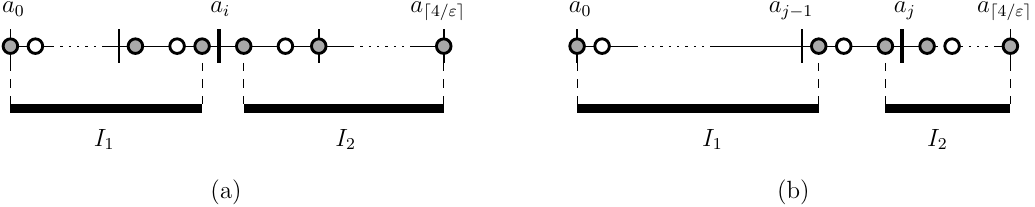}
\caption{The points in $R$ are shown in gray. 
(a) Both $|I_1|$ and $|I_2|$ are less than $D/2$. (b) $|I_1| \ge D/2$.}
\label{fig:certificate_1d}
\end{figure}

\paragraph*{Two-dimensional case}
We now show how to compute an $\eps$-certificate of $S$ of size $O(1/\eps^2)$ in $O(n)$ time.
Recall that $w^* = \max \{w(\sigma_1^*), w(\sigma_2^*) \}$ where $(\sigma_1^*, \sigma_2^*)$ is an optimal pair of slabs for the general two-line-center problem for $S$. 

We first compute a pair $(\lambda_1, \lambda_2)$ of slabs with the same width such that their union covers $S$ and their widths are at most $10w^*$.
This can be done in $O(n)$ time by Theorem~\ref{thm:10apx}.
Let $\widetilde w = w(\lambda_1) = w(\lambda_2)$, that is $ \widetilde w\le 10w^*$. 
For a small constant $c>0$ that will be specified later, let $\delta = c\eps$.
Let $L_1$ be a maximal set of equally spaced parallel lines contained in $\lambda_1$ such that $d(\ell, \ell') = \delta \widetilde  w$ for every two consecutive lines $\ell$ and $\ell'$ in $L_1$. See Figure~\ref{fig:L_construction}.
We define $L_2$ symmetrically for $\lambda_2$.
Let $L = L_1 \cup L_2$. 
Observe that $|L| = O(1/\delta) = O(1/\eps)$.

\begin{figure}[ht]
\centering
\includegraphics[width=0.5\textwidth]{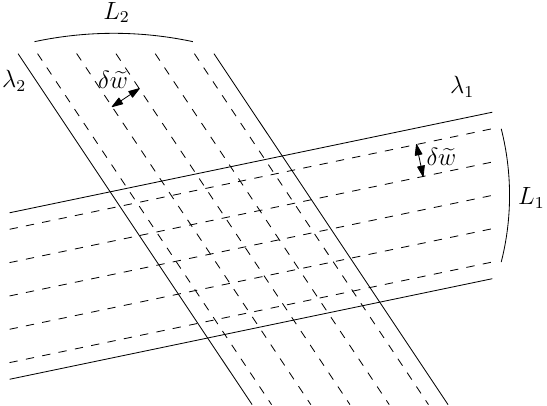}
\caption{The sets $L_1$ and $L_2$ of lines are shown as described above.}
\label{fig:L_construction}
\end{figure}

Let $S'$ be the point set obtained by projecting each point in $S$ onto its closest line in $L$.
For each $\ell \in L$, let $S'(\ell) \subseteq S'$ denote the set of projected points on $\ell$.
See Figure~\ref{fig:certificate_2d_ex}(a).
Let $Q'(\ell) \subseteq S'(\ell)$ be an $\eps$-certificate of $S'(\ell)$, obtained by applying the one-dimensional construction described above (treating $\ell$ as a one-dimensional space).
See Figure~\ref{fig:certificate_2d_ex}(b).
Finally, let $Q' = \bigcup_{\ell \in L} Q'(\ell)$, and let $Q$ be the set of original points in $S$ corresponding to $Q'$.
See Figure~\ref{fig:certificate_2d_ex}(c).
Observe that $Q\subseteq S$ and $|Q| =O(1/\eps^2)$, since $|L| = O(1/\eps)$ and $|Q'(\ell)| = O(1/\eps)$ for each $\ell\in L$.

\begin{figure}[ht]
\centering
\includegraphics[width=0.99 \textwidth]{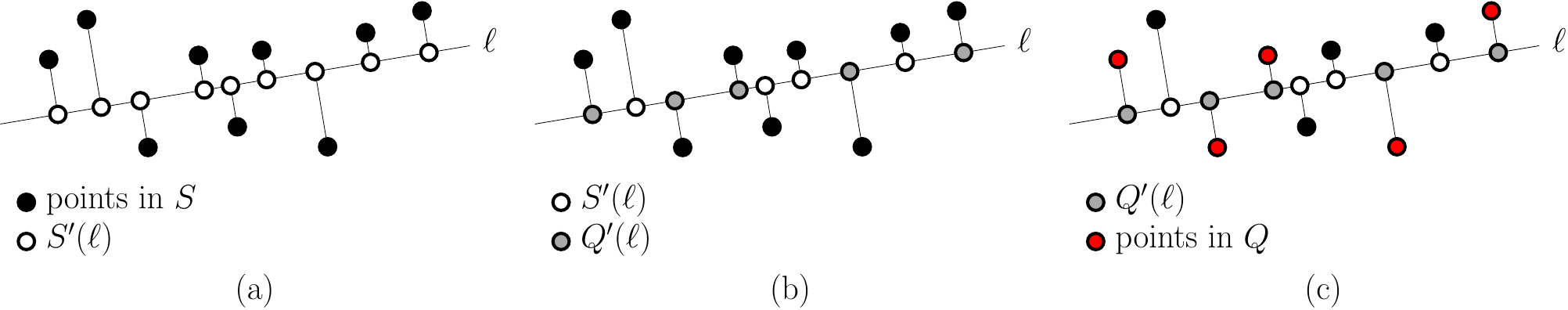}
\caption{A rough illustration of the construction of $Q$ from $S$ with respect to $\ell \in L$.}
\label{fig:certificate_2d_ex}
\end{figure}

We can compute $S'$ in $O(n)$ time using the floor function. 
It takes $O(\sum_{\ell\in L}|S'(\ell)|) = O(n)$ time to compute $Q'(\ell)$ for all $\ell\in L$.
Therefore, we can compute $Q$ in $O(n + 1/\eps)$ time, which is $O(n)$ by the assumption that $n>1/\eps^2$. 
Now we show that $Q$ is an $\eps$-certificate of $S$.

\begin{claim*}
    $Q$ is an $\eps$-certificate of $S$. 
\end{claim*}
\begin{claimproof}
Our proof follows along the same lines as the proof in~\cite{agarwal2005approximation}.
Let $(\sigma_1, \sigma_2)$ be a pair of slabs with the same width whose union covers $Q$. 
Let $(\sigma_1',\sigma_2')$ be the pair of slabs obtained by applying a $2\delta \widetilde  w$ additive expansion to $\sigma_1$ and $\sigma_2$, respectively. 
That is, $w(\sigma_1') = w(\sigma_1) + 2\delta \widetilde w$, and $\sigma_1'$ and $\sigma_1$ have the same center line. The same holds for $\sigma_2'$.
For a point $q$ in $Q$, let $q'$ denote its corresponding point in $Q'$. 
Observe that $d(q, q') \le \delta \widetilde  w$ for every $q\in Q$.
Therefore, $Q' \subset \sigma_1' \cup \sigma_2'$ (see Figure~\ref{fig:certificate_2d}(a)).

Let $(\tau_1',\tau_2')$ be the pair of slabs obtained by taking the $\delta$-expansions of $\sigma_1'$ and $\sigma_2'$, respectively. 
For each line $\ell \in L$, 
observe that the segment $\ell \cap \tau_1'$ is
the $\delta$-expansion of $\ell \cap \sigma_1'$. 
Similarly, $\ell \cap \tau_2'$ is
the $\delta$-expansion of $\ell \cap \sigma_2'$. 
As the union of $\ell \cap \sigma_1'$ and $\ell \cap \sigma_2'$ covers $Q'(\ell)$, 
the union of $\ell \cap \tau_1'$ and $\ell \cap \tau_2'$ covers $S'(\ell)$. 
This leads to $S' \subset \tau_1' \cup \tau_2'$ (see Figure~\ref{fig:certificate_2d}(b)). 

Let $(\tau_1,\tau_2)$ be the pair of slabs  obtained by applying a $2\delta \widetilde w$ additive expansion to $\tau_1'$ and $\tau_2'$, respectively.
For a point $p$ in $S$, let $p'$ denote its corresponding point in $S'$. 
Observe that $d(p, p') \le \delta \widetilde  w$ for every $p\in S$.
Therefore, $S \subset \tau_1 \cup \tau_2$ (see Figure~\ref{fig:certificate_2d}(c)). 

It remains to prove that $w(\tau_1) \le (1+\eps) \cdot w(\sigma_1)$ and $w(\tau_2) \le (1+\eps)\cdot w(\sigma_2)$. We only prove this for $\tau_1$.
Let $r = w(\sigma_1)$. 
Then, $ w(\tau_1) = (r + 2\delta \widetilde  w)(1 + \delta) + 2\delta\widetilde  w$.

As $\tau_1 \cup \tau_2$ covers $S$, it holds that $w^* \le (r + 2\delta \widetilde  w)(1 + \delta) + 2\delta\widetilde  w $. 
Since $\widetilde w\le 10 w^*$, it follows that $w^* \le  (r + 20\delta   w^*)(1 + \delta) + 20\delta  w^*$. 
Therefore, we have $w^*(1-40\delta -20\delta^2) \le r(1+\delta)$. 
For a sufficiently small $\delta$, it holds that $w^*/2 \le r(1+\delta)$. 
Therefore, $w^* \le 4r$. 

Now, we rearrange $(r + 2\delta \widetilde  w)(1 + \delta) + 2\delta\widetilde  w$ as $r(1+\delta) + 2\delta \widetilde w(2+\delta)$. 
Since $\widetilde w \le 10 w^* \le 40 r$, we have 
$r(1+\delta) + 2\delta \widetilde w(2+\delta) \le r(1+\delta) + 80\delta  (2+\delta) r = (1 + 161 \delta + 80 \delta^2)r$. 
This is at most $(1+\eps)r$ for sufficiently small $c>0$ (e.g., $c=1/200$).
\end{claimproof}
This completes the proof of Theorem~\ref{thm:computing_eps_certificate}.
\end{proof}

\begin{figure}[ht]
\centering
\includegraphics[width=0.94 \textwidth]{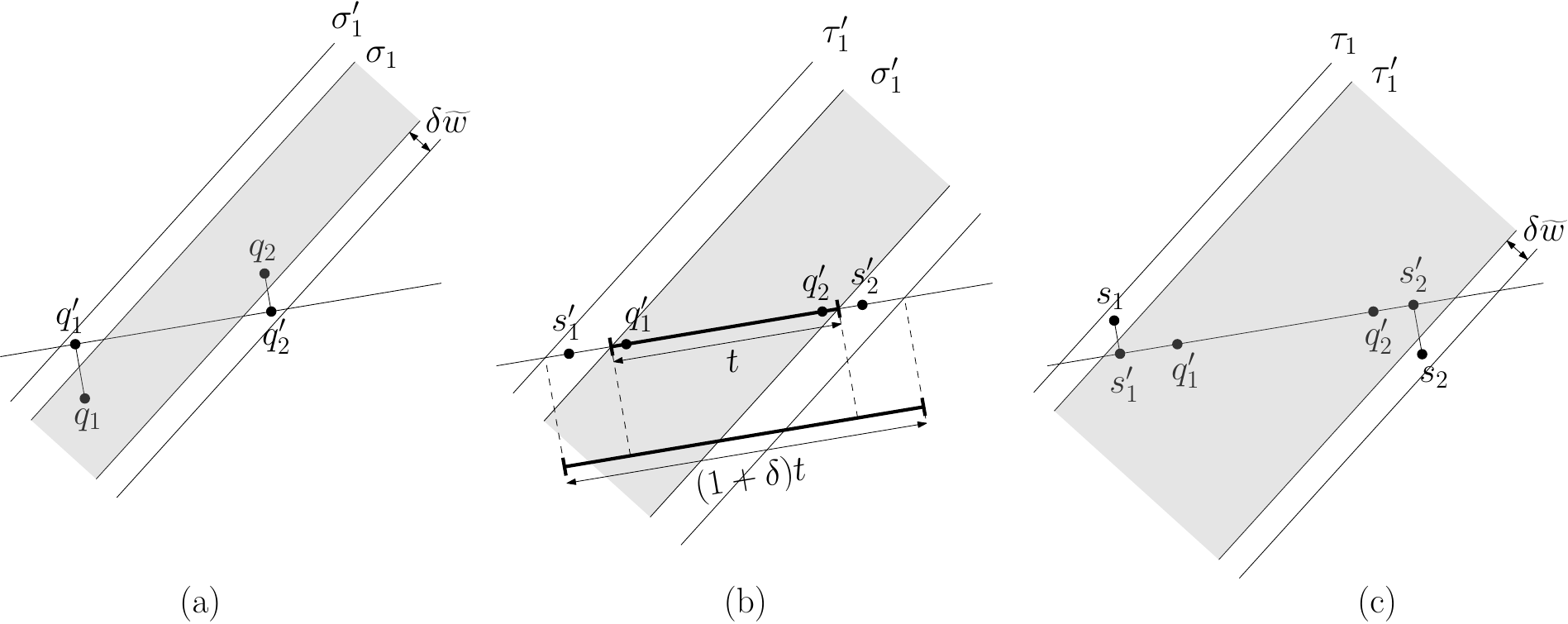}
\caption{A rough illustration of the construction of $\tau_1$ from $\sigma_1$.}
\label{fig:certificate_2d}
\end{figure}


We show that Theorem~\ref{thm:computing_eps_certificate} directly gives a $(1+\eps)$-approximation algorithm that runs in $O(n + \mathrm{poly}(1/\eps))$ time for each of the four problems.

\begin{theorem}
    \label{thm:poly_eps}
    For each of the four problems whose optimal solution can be found in $T(n)$ time, 
    we can compute a $(1+\eps)$-approximate solution in $O(n + T(1/\eps^2))$ time.
\end{theorem}

\begin{proof}
    
We first compute an $\eps$-certificate $Q$ of $S$. 
By Theorem~\ref{thm:computing_eps_certificate}, this can be done in $O(n)$ time and $|Q| = O(1/\eps^2)$.

Consider any of the four problems.
Let $w^*$ denote the maximum width of an optimal pair for $S$.  
Let $T(n)$ denote the running time of an exact algorithm for finding an optimal pair.
Such a polynomial-time exact algorithm exists for each of the four problems (see Table~\ref{table:result}).  
In $T(|Q|) = T(1/\eps^2)$ time, we find an optimal pair $(\sigma_1, \sigma_2)$ of slabs for $Q$.  
Without loss of generality, we assume that $\sigma_1$ and $\sigma_2$ have the same width (otherwise, we expand the smaller one).  
Since $Q \subseteq S$, it follows that $w(\sigma_1) = w(\sigma_2) \le w^*$.

Let $(\sigma_1', \sigma_2')$ be the pair of slabs obtained by taking the $\eps$-expansions of $\sigma_1$ and $\sigma_2$, respectively.  
By the definition of an $\eps$-expansion, we have $w(\sigma_1') = w(\sigma_2') \le (1+\eps)w^*$, and the pair $(\sigma_1', \sigma_2')$ satisfies the orientation restriction of the problem.  
Furthermore, by the definition of an $\eps$-certificate, $\sigma_1' \cup \sigma_2'$ covers $S$.  
Therefore, $(\sigma_1', \sigma_2')$ forms a $(1+\eps)$-approximate solution.
\end{proof}

For the general two-line-center problem, 
$T(n) = O(n^2 \log^2 n)$~\cite{jaromczyk1995}.
By Theorem~\ref{thm:poly_eps}, we obtain a $(1+\eps)$-approximation algorithm running in 
$O(n + (1/\eps^4) \log^2 (1/\eps))$~time. 
In Section~\ref{sec:general}, however, we present another algorithm with an improved running time. 
For the other variants as well, we present $(1+\eps)$-approximation algorithms with better running times than those obtained by applying Theorem~\ref{thm:poly_eps} with a known exact algorithm as $T(n)$~(see Table~\ref{table:result}).

\section{One Fixed Orientation}
\label{sec:one_fixed}



By rotating the coordinate system, we assume without loss of generality that the given orientation $\theta$ is $0$. 
Throughout this section, we denote by $(\sigma_1^*, \sigma_2^*)$ an optimal pair of slabs for the one-fixed-orientation problem, where $\sigma_1^*$ is horizontal. Let $w^* = \max \{w(\sigma_1^*), w(\sigma_2^*) \}$.

\subsection{Orientation-constrained Width and Dominance}
\label{sec:one_fixed_properties}
In this section, we recall the concept of an orientation-constrained width and present geometric properties as introduced in~\cite{ahn2024constrained}.
Let $\theta_1\le \theta_2$ be two orientations. 
For a set $P$ of points in the plane, we define the $[\theta_1, \theta_2]$-constrained width of $P$ as $\width_{[\theta_1, \theta_2]}(P) \coloneq \min_{\theta \in [\theta_1, \theta_2]} \width_\theta (P).$ 
We also define $\sigma_{[\theta_1, \theta_2]}(P) \coloneq \bigcap_{\theta \in [\theta_1, \theta_2]} \sigma_\theta (P)$ (see Figure~\ref{fig:domination}(a)).

Consider two non-empty sets $P$ and $Q$ of points in the plane that can be separated by a horizontal line, thus $\conv(P) \cap \conv(Q) = \emptyset$. 
Then, there are exactly two outer common tangents. 
Let $\theta_{\min} \le  \theta_{\max} $ be the orientations of these lines. 
We say that $P$ \emph{dominates} $Q$ if $\sigma_{[\theta_{\min},\theta_{\max}]}(Q) \subseteq \sigma_{[\theta_{\min},\theta_{\max}]}(P)$.
See Figure~\ref{fig:domination}.
 By construction, note that either $P$ or $Q$ dominates the other.

\begin{figure}[ht]
    \centering
    \includegraphics[scale=0.65]{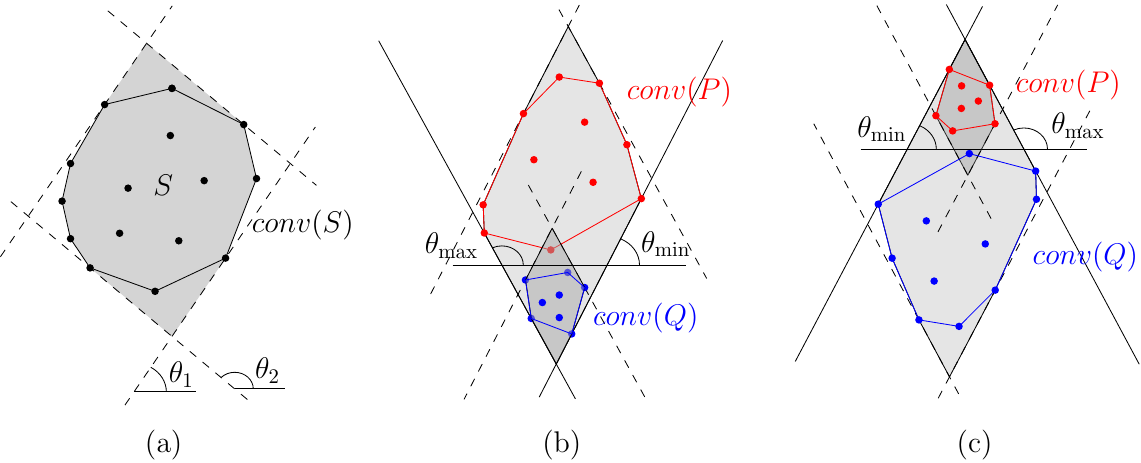}
    \caption{(a) For a point set $S$ and orientations $\theta_1$ and $\theta_2$, the gray region represents $\sigma_{[\theta_1, \theta_2]}(S)$.
    (b, c) For two point sets $P$ (red) and $Q$ (blue), the two lines represent their outer common tangents. \\
    (b) $P$ dominates $Q$.
    (c) $Q$ dominates $P$.
    }
    \label{fig:domination}
\end{figure}

\begin{lemma}[\cite{ahn2024constrained}]
    \label{lem:dom_properties}
    With the above notation, suppose that $P$ dominates $Q$. 
    Then
    \begin{enumerate}[(i)]
        \item $\width_{[\theta_{\min}, \theta_{\max}]}(P \cup Q) = \width_{[\theta_{\min}, \theta_{\max}]}(P)$, and 
        \item If $\width(P \cup Q) < d(\conv(P), \conv(Q)) $, then $\width(P\cup Q) = \width_{[\theta_{\min}, \theta_{\max}]}(P)$.
    \end{enumerate}
\end{lemma}

\subsection{Decision Algorithm}
\label{sec:one_fixed_decision}
In this section, we present a decision algorithm for a given parameter $\omega > 0$.
The algorithm determines whether there exists a pair of slabs of width $\omega$ such that one of them is horizontal and their union covers $S$.
If such a pair exists, we report YES. Otherwise, we report NO.

The points $s_1, s_2, \ldots, s_n$ of $S$ are assumed to be sorted in decreasing order of their $y$-coordinates. 
For simplicity, we assume that no two input points have the same $y$-coordinate. 
Our approach follows that of the previously known algorithm from~\cite{ahn2024constrained} but removes the dependency on the data structures they used.
As a result, we obtain an $O(n)$-time decision algorithm (after sorting), improving upon the previous $O(n \log^2 n)$ bound.

Consider a horizontal slab $\sigma_\omega$ of width $\omega$, and imagine sweeping $\sigma_\omega$ downward over the plane. 
For a fixed position of $\sigma_\omega$, let $S_a \subseteq S$ be the set of points strictly above $\sigma_\omega$, and let $S_b \subseteq S$ be the set of points strictly below $\sigma_\omega$. 
Referring to Figure~\ref{fig:decision_notations}, observe that $S_b$ dominates $S_a$ from the beginning until a certain moment, after which $S_a$ dominates $S_b$ until the end, assuming that both $S_a$ and $S_b$ are non-empty. 
Observe that the answer to the decision problem is YES if and only if there exists a placement of $\sigma_\omega$ such that $\width(S \setminus \sigma_\omega) \le \omega$~(equivalently, $\width(S_a \cup S_b) \le \omega$).

\textbf{Assumption} For the case where the answer is YES, we assume that there exists a placement of $\sigma_\omega$ such that $S_a$ dominates $S_b$. 
The other case can be handled symmetrically by sweeping $\sigma_\omega$ upward from the bottom.

\textbf{Notation}
If both $S_a$ and $S_b$ are non-empty, there exist two outer common tangents of them. 
Observe that these tangent lines do not intersect each other inside $\sigma_\omega$.
Let $\ell_1$ denote the right tangent and $\ell_2$ denote the left tangent.
We denote their orientations by $\theta_1$ and $\theta_2$, respectively. 
See Figure~\ref{fig:decision_notations}(a).
Let $a_1 \in S_a$ and $b_1 \in S_b$ be the points of tangency defining $\ell_1$, and similarly, let $a_2 \in S_a$ and $b_2 \in S_b$ be the points of tangency defining $\ell_2$.\footnote{There may be cases where more than one point of $S_a$ (resp., $S_b$) lies on the tangent line. 
Our algorithm works even if one of them is chosen arbitrarily. 
Thus, for simplicity, we assume throughout this section that there is exactly one point of $S_a$ (resp., $S_b$) on the tangent line.} 
If either $S_a$ or $S_b$ is empty, then $\ell_1$ and $\ell_2$ (and consequently $\{a_i, b_i, \theta_i \mid i = 1, 2\}$) are undefined.

\begin{figure}[ht]
    \centering
    \includegraphics[width=0.98\textwidth]{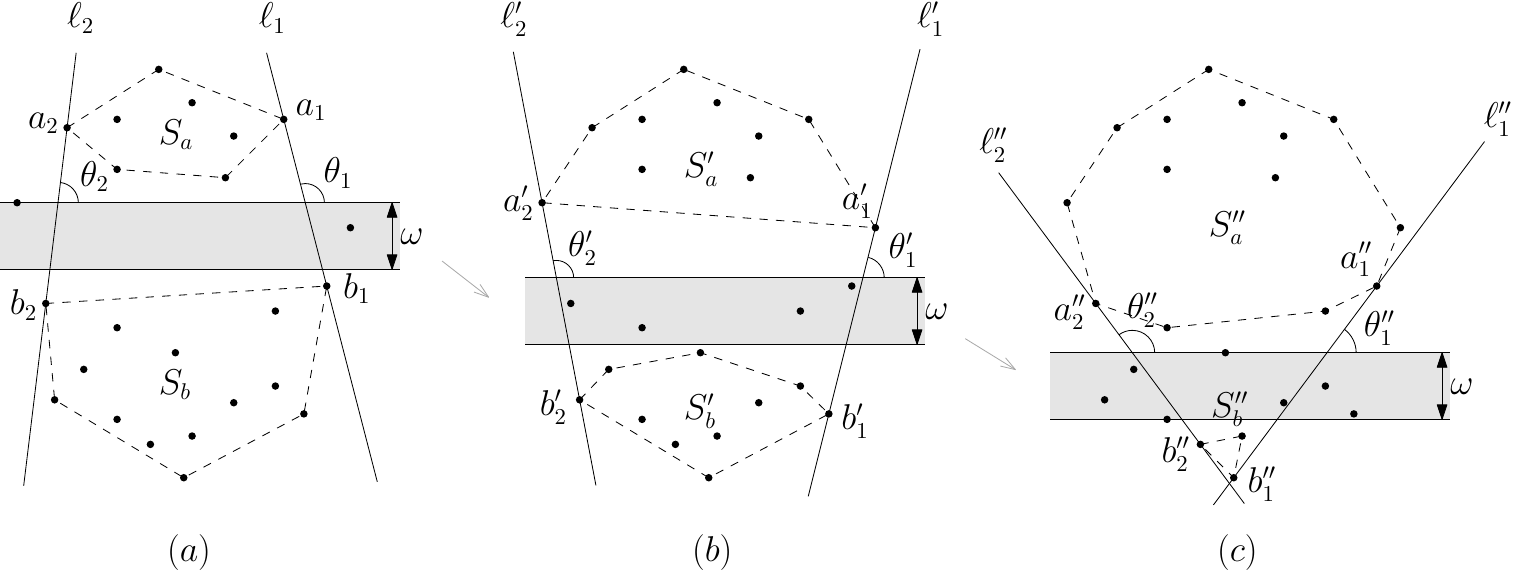}
    \caption{Three snapshots of the sweeping process for a point set.
    (a) $S_a$, $S_b$, and $\{\ell_i, a_i, b_i, \theta_i \mid i = 1, 2\}$ are shown according to the notation defined above. Observe that $S_b$ dominates $S_a$.
    (b, c) We use the notations $S_a', S_b'$ and $S_a'', S_b''$ (and the associated primed and double-primed variables) analogously for the two horizontal slabs positioned below.
    Observe that $S_a'$ (resp., $S_a''$) dominates $S_b'$ (resp., $S_b''$). 
    (a--c)~Observe that $\theta_1\ge \theta_1' \ge \theta_1''$ and $\theta_2 \le \theta_2'\le \theta_2''$. 
    }
    \label{fig:decision_notations}
\end{figure}

\begin{observation}
    \label{obs:domination}
    With the above notation, $S_a$ dominates $S_b$ if and only if $\theta_1 \le \theta_2$. 
\end{observation}

\begin{proof}
    First, consider the case where $\theta_1 \le \theta_2$. 
    For any orientation $\theta \in [\theta_1, \theta_2]$, it holds that $\sigma_{\theta}(S_b) \subseteq \sigma_{\theta}(S_a)$ (see Figure~\ref{fig:decision_notations}(b), with a slight abuse of notation).
    Therefore, we have $\sigma_{[\theta_1, \theta_2]}(S_b) \subseteq \sigma_{[\theta_1, \theta_2]}(S_a)$.

    Now consider the other case where $\theta_1 > \theta_2$.
    For any orientation $\theta \in [\theta_2, \theta_1]$, it holds that $\sigma_{\theta}(S_a) \subsetneq \sigma_{\theta}(S_b)$ (see Figure~\ref{fig:decision_notations}(a)). 
    Therefore, we have $\sigma_{[\theta_2, \theta_1]}(S_a) \subsetneq \sigma_{[\theta_2, \theta_1]}(S_b)$.
\end{proof}

\subsubsection{Overview of the Algorithm}
\label{sec:decision_overview}

In this section, we give an overview of the data structures maintained throughout the algorithm and explain how the main loop proceeds. 
We begin by recalling a data structure for maintaining the convex hull of a point set under a special sequence of insertions and deletions.

\begin{lemma}[\cite{wang2023dynamic}]
\label{lem:dyn_conv}
    We can dynamically maintain the convex hull of a set $P$ of points in the plane to support each \emph{window-sliding update} (i.e., either insert a point below all points of $P$ or delete the topmost point of $P$) in $O(1)$ amortized time. 
    The convex hull can be explicitly reported in $O(h)$ time, where $h$ is the number of vertices of the convex hull. 
\end{lemma}

We note that in the above data structure, the neighbors of any given vertex on the convex hull can be found in $O(1)$ time. 
Although it is not explicitly mentioned in~\cite{wang2023dynamic}, it can be achieved since each point $p$ (implicitly) maintains two pointers to its left and right neighbors, respectively, if $p$ is a vertex of the current convex hull (see Sections~3.1 and~4 of \cite{wang2023dynamic}).

Observe that as $\sigma_\omega$ moves downward, a point is inserted into $S_a$ below all existing points of $S_a$, or the topmost point of $S_b$ is deleted from $S_b$. 
We refer to each such moment as \emph{an event}. 
Each event corresponds to a window-sliding update of either $S_a$ or $S_b$, and we have at most $2n$ events.  

\paragraph*{Data structures and Main loop}
\label{par:ds_mainloop}
As $\sigma_\omega$ moves downward, we maintain the following using the notation defined above:
\begin{itemize}
    \item the dynamic data structures for $\conv(S_a)$ and $\conv(S_b)$ (from Lemma~\ref{lem:dyn_conv}),
    \item the tangent lines $\ell_1$ and $\ell_2$, with their corresponding points $a_1$, $b_1$, $a_2$, and $b_2$,
    \item the two antipodal pairs of $S_a$ with respect to the orientations $\theta_1$ and $\theta_2$, respectively.\footnote{There may be more than one antipodal pair of $S_a$ with respect to $\theta_1$ (resp., $\theta_2$). 
    Our algorithm works even if one of them is chosen arbitrarily. 
    Thus, for simplicity, we assume throughout this section that there is exactly one antipodal pair with respect to $\theta_1$ (resp., $\theta_2$).}
\end{itemize}

Initially, we start with a position of $\sigma_\omega$ such that $S_a$ contains only $s_1$, the topmost point of $S$. 
If the corresponding set $S_b$ is empty, we stop the algorithm and report YES. Otherwise, we initialize the data structures for $\conv(S_a)$ and $\conv(S_b)$, respectively.
Then we find $\ell_1$ and $\ell_2$, together with $\{a_i, b_i, \theta_i \mid i = 1, 2\}$. 
Note that the two antipodal pairs of $S_a$ with respect to $\theta_1$ and $\theta_2$ are both $(s_1, s_1)$.

In the main loop of our decision algorithm, we do the following as $\sigma_\omega$ moves downward.
If $S_b$ dominates $S_a$, we just update the above data structures.  
Once $S_a$ begins to dominate $S_b$, we include, for every next event, an additional step that tests whether $\width_{[\theta_1, \theta_2]}(S_a) \le \omega$, using the maintained data structures. If this condition holds, we stop the algorithm and report YES; otherwise, we proceed the algorithm. 

If $S_b$ becomes empty, we set $\theta_1 = 0$ and $\theta_2 = \pi$. 
We then test whether $\width_{[\theta_1, \theta_2]}(S_a) = \width(S_a) \le \omega$. 
If this condition holds, we report YES; otherwise, we report NO.

The correctness of this algorithm, assuming that all data structures are maintained properly and the test $\width_{[\theta_1, \theta_2]}(S_a) \le \omega$ is evaluated correctly, is guaranteed by the following lemma.

\begin{lemma}
    \label{lem:decision_correctness}
    For any placement of $\sigma_\omega$ where $S_a$ dominates $S_b$, we have $\width(S_a \cup S_b) \le \omega$ if and only if $\width_{[\theta_1, \theta_2]}(S_a) \le \omega$. 
\end{lemma}

\begin{proof}
    First, observe that $\omega < d(\conv(S_a), \conv(S_b)) $ 
    as $S_a$ and $S_b$ are separated by a horizontal slab of width $\omega$.
    If $\width(S_a \cup S_b) \le \omega$, then $\width(S_a\cup S_b) =width_{[\theta_1, \theta_2]}(S_a)$ by the above observation and Lemma~\ref{lem:dom_properties}(ii).
    Therefore, we have $width_{[\theta_1, \theta_2]}(S_a) \le \omega$. 

    Conversely, if $width_{[\theta_1, \theta_2]}(S_a) \le \omega$, then $width_{[\theta_1, \theta_2]}(S_a \cup S_b) \le \omega$ by Lemma~\ref{lem:dom_properties}(i).
    By definition, $\width(S_a \cup S_b) \le width_{[\theta_1, \theta_2]}(S_a \cup S_b)$.
    Therefore, we have $\width(S_a \cup S_b) \le \omega$.
\end{proof}

\subsubsection{Maintaining Data Structures and Testing the Width Condition}
\label{sec:ds_maintain}

In this section, we describe how to maintain the data structures and test the condition $\width_{[\theta_1, \theta_2]}(S_a) \le \omega$ for each event that occurs as $\sigma_\omega$ moves downward.

We begin by proving Lemma~\ref{lem:width_test_obs}, which motivates the part of the algorithm that tests whether $\width_{[\theta_1, \theta_2]}(S_a) \le \omega$.
Let $\sigma_\omega$ be a horizontal slab of width $\omega$ where $S_a$ dominates $S_b$.
Recall that $\theta_1$ and $\theta_2$ are the orientations of the right and left outer tangent lines, respectively.
By Observation~\ref{obs:domination}, we have $\theta_1 \le \theta_2$.
Now, let $\sigma_\omega'$ be another horizontal slab of width $\omega$, positioned below $\sigma_\omega$ (the two slabs may partially overlap).
We use the notations $S_a'$, $\theta_1'$, and $\theta_2'$ analogously for $\sigma_{\omega}'$.
Then, it always holds that $[\theta_1, \theta_2] \subseteq[\theta_1', \theta_2']$.

\begin{lemma}
    \label{lem:width_test_obs}
    With the above notation, assume that $width_{[\theta_1, \theta_2]}(S_a) > \omega$. 
Then, \\$width_{[\theta_1', \theta_2']} (S_a') \le \omega$ if and only if 
$width_{[\theta_1' , \theta_1]}(S_a') \le \omega$ or
$width_{[\theta_2, \theta_2']}(S_a') \le \omega$.
\end{lemma}

\begin{proof}
    Since $[\theta_1, \theta_2] \subseteq [\theta_1', \theta_2']$,
    we can partition the range $[\theta_1', \theta_2']$ into three subranges: $[\theta_1', \theta_1]$, $[\theta_1, \theta_2]$, and $[\theta_2, \theta_2']$.
    Let $a = width_{[\theta_1', \theta_1]}(S_a')$, 
    $b = width_{[\theta_1, \theta_2]}(S_a')$, and 
    $c = width_{[\theta_2, \theta_2']}(S_a')$.
    By definition, $width_{[\theta_1', \theta_2']}(S_a') = \min \{a, b, c \}$. 

    By the assumption and since $S_a' \supseteq S_a$, we have $b \ge width_{[\theta_1, \theta_2]}(S_a) > \omega$.
    If $width_{[\theta_1', \theta_2']}(S_a') \allowbreak \le \omega$, then $a$ or $c$ should be at most $\omega$.
    Conversely, if $a \le \omega$ or $c \le \omega$, it follows that $width_{[\theta_1', \theta_2']}(S_a') \le \omega$.
\end{proof}


We first note that the data structures (from Lemma~\ref{lem:dyn_conv}) for $\conv(S_a)$ and $\conv(S_b)$ can be maintained in $O(n)$ total time, since each event corresponds to a window-sliding update.

We now describe how to handle each event depending on its type: either an insertion into $S_a$ or a deletion from $S_b$.

We first observe that if the point being inserted or deleted lies strictly between $\ell_1$ and $\ell_2$, then nothing changes except the convex hulls (see Figure~\ref{fig:simple_events}).
Therefore, we assume that the point corresponding to an event does not lie strictly between $\ell_1$ and $\ell_2$.

\begin{figure}[ht]
    \centering
    \includegraphics[width=0.9 \textwidth]{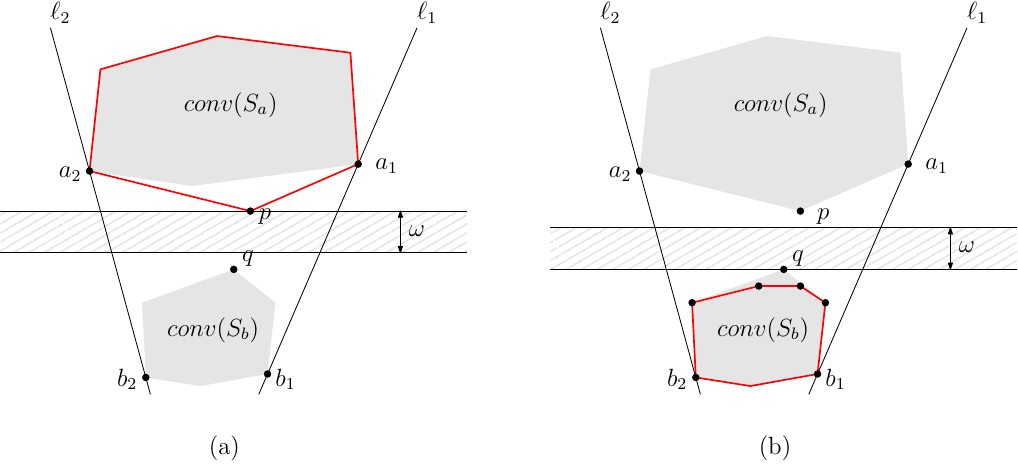}
    \caption{(a) $p$ is inserted into $S_a$. (b) $q$ is deleted from $S_b$.
    The red polygons show the convex hulls immediately after the events. 
    Observe that the tangent lines $\ell_1$ and $\ell_2$ remain unchanged.    
    }
    \label{fig:simple_events}
\end{figure}

\paragraph*{Insertion into $S_a$}
Let $p$ be the point inserted into $S_a$, and assume that $p$ lies to the right of $\ell_1$ (see Figure~\ref{fig:insertions}). The case where $p$ lies to the left of $\ell_2$ can be handled symmetrically.

\textbf{Tangent searching }
By our assumption, $\ell_2$ remains unchanged, while $\ell_1$ changes. 
Let $\ell_1'$ denote the new tangent line, and let $\theta_1'$ denote its orientation. 
Let $S_a' = S_a \cup \{p\}$. 
Analogously, let $a_1' \in S_a'$ and $b_1' \in S_b$ be the points lying on $\ell_1'$. Note that $a_1' = p$. 
We find $b_1'$ by traversing $\conv(S_b)$ clockwise from $b_1$; we refer to this as the \emph{insertion-type tangent searching procedure}. 
This procedure takes $O(1+k)$ time, where $k$ is the number of vertices on $\conv(S_b)$ strictly between $b_1$ and $b_1'$. 
We say that these vertices are \emph{involved} in the procedure.

\textbf{Antipodal pair searching }
Now we find the antipodal pairs of $S_a'$ with respect to $\theta_1'$ and $\theta_2$ (since $\ell_2$ does not change). 
Let $q$ be the point such that $(a_1, q)$ is the antipodal pair of $S_a$ with respect to $\theta_1$. 

For the orientation $\theta_1'$, observe that $a_1' (= p)$ is one point of the antipodal pair of $S_a'$. 
We can find the other point by traversing $\conv(S_a')$ clockwise from $q$. 
Let $q'$ denote the point such that $(a_1', q')$ is the antipodal pair of $S_a'$ with respect to $\theta_1'$. 
We refer to this as the \emph{insertion-type antipodal-pair searching procedure}, which takes $O(1 + k')$ time, where $k'$ is the number of vertices on $\conv(S_a')$ (equivalently, on $\conv(S_a)$) strictly between $q$ and $q'$.
We say that these vertices are \emph{involved} in the procedure.

For the orientation $\theta_2$, observe that $a_2$ is always one point of the antipodal pairs for both $S_a$ and $S_a'$. 
Let $r$ be the point such that $(a_2, r)$ forms an antipodal pair of $S_a$ with respect to $\theta_2$ (see Figure~\ref{fig:insertions}(b)). 
Let $r'$ be the point such that $(a_2, r')$ forms an antipodal pair of $S_a'$ with respect to $\theta_2$. 
Then $r' = r$ if $p$ lies to the left of the line of orientation $\theta_2$ passing through $r$, and $r' = p$ otherwise. 

\textbf{Testing the width condition }
If $S_a'$ does not dominate $S_b$, there is nothing left to do.
This can be easily detected by Observation~\ref{obs:domination}.
Otherwise, we need to test whether $\width_{[\theta_1', \theta_2]}(S_a') \le \omega$.
By Lemma~\ref{lem:width_test_obs}, it suffices to determine whether $\width_{[\theta_1', \theta_1]}(S_a') \le \omega$.
We do this by following the antipodal pairs of $S_a'$ from the one corresponding to $\theta_1'$ to that of $\theta_1$.
Note that $p$ is always one point of the antipodal pairs of $S_a'$ for all orientations between $\theta_1'$ and $\theta_1$, since $p$ lies to the right of $\ell_1$.
Hence, $(p, q)$ is the antipodal pair of $S_a'$ with respect to $\theta_1$, and $(p, q')$ is the one with respect to $\theta_1'$ as explained above.
By traversing $\conv(S_a')$ counterclockwise starting from $q'$, we can visit all antipodal pairs corresponding to the range of orientations $[\theta_1', \theta_1]$ in order.
Following them sequentially allows us to find an orientation $\theta \in [\theta_1', \theta_1]$ that minimizes $\width_\theta(S_a')$.
This procedure takes $O(1 + k')$ time, where, as before, $k'$ is the number of vertices on $\conv(S_a')$ (equivalently, on $\conv(S_a)$) that lie strictly between $q$ and $q'$.
Note that this is bounded by the time required for the insertion-type antipodal-pair searching procedure.

\medbreak
The total time for handling point insertions as $\sigma_\omega$ moves downward depends on the two procedures: the insertion-type tangent searching procedure and the insertion-type antipodal-pair searching procedure. 
We conclude that it takes $O(n)$ time in total by showing that each input point can be involved at most once in each of the procedures.

\begin{lemma}
    \label{lem:insertion_time}
    Each point of $S$ can be involved in the insertion-type tangent searching procedure and the insertion-type antipodal-pair searching procedure at most once, respectively, in the entire algorithm.
\end{lemma}

\begin{proof}
    We first show that each input point is involved in the insertion-type tangent searching procedure at most once.
    For a point $p$ inserted into $S_a$, recall that $b_1 \in S_b$ is the original point of tangency on $\ell_1$, and $b_1' \in S_b$ is the new point of tangency on $\ell_1'$ (see Figure~\ref{fig:insertions}(a)).
    We use the following observation:
    when traversing the points on $\conv(S_b)$ in clockwise order from $b_1$ to $b_1'$, their $y$-coordinates decrease.
    
    Let $s$ be a point on $\conv(S_b)$ that lies strictly between $b_1$ and $b_1'$, that is, $y(b_1) > y(s) > y(b_1')$.
    Assume to the contrary that $s$ is involved again in the tangent searching procedure when another point $t$ is inserted after $p$.
    Let $u_1$ and $u_1'$ denote the original and new points of tangency defined analogously for this insertion.
    By the above observation, we have $y(u_1) > y(s) > y(u_1')$.
    Moreover, since the new point of tangency is found by traversing the convex hull in clockwise order, 
    it holds that $y(b_1) > y(b_1') \ge y(u_1) > y(u_1')$.
    This contradicts that $y(b_1) > y(s) > y(b_1')$.

    In a similar way, we can show that each input point is involved in the insertion-type antipodal-pair searching procedure at most once.
    Recall that $q$ and $q'$ are the points such that $(a_1, q)$ is the antipodal pair of $S_a$ with respect to $\theta_1$ and $(a_1',q')$ is the one of $S_a'$ with respect to $\theta_1'$ (see Figure~\ref{fig:insertions}(a)). 
    We use a similar observation: when traversing the points on $\conv(S_a)$ in clockwise order from $q$ to $q'$, their $y$-coordinates increase.
    
    Let $s$ be a point on $\conv(S_a)$ that lies strictly between $q$ and $q'$, that is, $y(q) < y(s) < y(q')$.
    Assume to the contrary that $s$ is involved again in the antipodal-pair searching procedure when another point $t$ is inserted after $p$.
    Analogously, let $v$ and $v'$ denote the points such that $v$ is a point of the original antipodal pair and $v'$ is a point of the new antipodal pair for this insertion.
    By the above observation, we have $y(v) < y(s) < y(v')$.
    Moreover, since a point of the new antipodal pair is found by traversing the convex hull in clockwise order, 
    it holds that $y(q) < y(q') \le y(v) < y(v')$.
    This contradicts $y(q) < y(s) < y(q')$.
\end{proof}

\begin{figure}[ht]
    \centering
    \includegraphics[width=0.9 \textwidth]{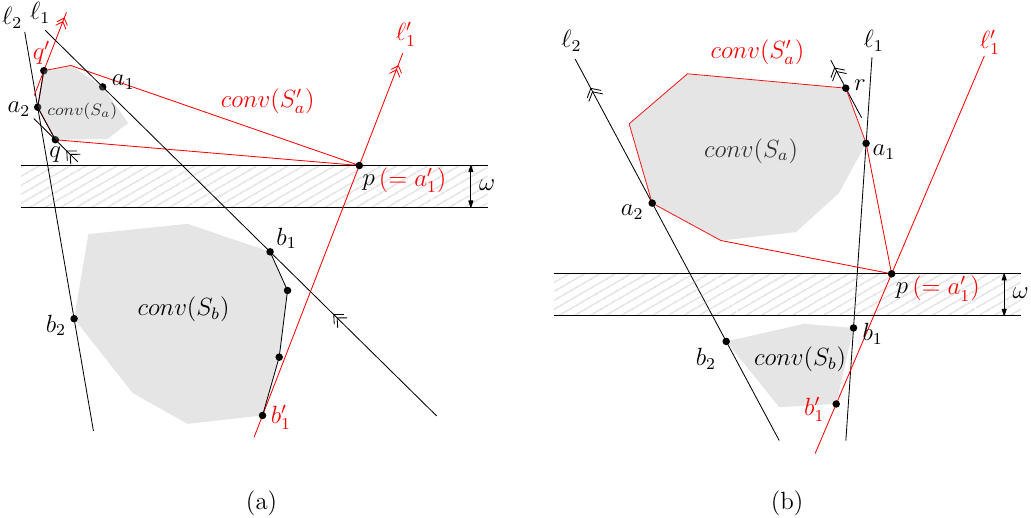}
    \caption{Two cases of insertions. 
    In each figure, the changes after the insertion are shown in red: $\conv(S_a')$, $a_1'$, $b_1'$, and $\ell_1'$. 
    In (a), observe that $p$ becomes a point of the antipodal pair of $S_a'$ with respect to $\theta_2$. 
    In contrast, in (b), the antipodal pair of $S_a'$ with respect to $\theta_2$ remains $(a_2, r)$ even after inserting $p$.
    In (a), two points are involved in the tangent searching procedure, and one point is involved in the antipodal-pair searching procedure.}
    \label{fig:insertions}
\end{figure}

In a symmetric way, we handle the deletion of a point from $S_b$. 

\paragraph*{Deletion from $S_b$}

Let $p$ be the point deleted from $S_b$. 
As we only consider the case where $p$ does not lie strictly between $\ell_1$ and $\ell_2$, $p$ should lie on $\ell_1$ or $\ell_2$. 
We assume that $p$ lies on $\ell_1$, that is $p = b_1$ (see Figure~\ref{fig:deletions}).
The case where $p$ lies on $\ell_2$ can be handled symmetrically.

\textbf{Tangent searching}
By our assumption, $\ell_2$ remains unchanged, while $\ell_1$ changes. 
Let $\ell_1'$ denote the new tangent line, and let $\theta_1'$ denote its orientation. 
Define $S_b' = S_b \setminus \{p\}$. 
Analogously, let $a_1' \in S_a$ and $b_1' \in S_b'$ be the points lying on $\ell_1'$.

We find $a_1'$ and $b_1'$ as follows.  
Let $t$ be the counterclockwise neighbor of $p$ on $\conv(S_b)$.  
Starting from $t$, we traverse $\conv(S_b')$ clockwise until we find a point $v_1$ on $\conv(S_b')$ such that $\overline{a_1 v_1}$ is tangent to $\conv(S_b')$ at $v_1$.
If $\overline{a_1 v_1}$ is also tangent to $\conv(S_a)$ at $a_1$, we are done.  
Otherwise, we set $u = a_1$ and $v = v_1$.  
We then move $u$ to its clockwise neighbor on $\conv(S_a)$, and update $v$ so that $\overline{uv}$ becomes tangent to $\conv(S_b')$ at $v$. 
The point $v$ can be updated by traversing $\conv(S_b')$ clockwise from its previous position.
We repeat this process until $\overline{uv}$ becomes tangent to both $\conv(S_a)$ at $u$ and $\conv(S_b')$ at $v$. 
The final positions of $u$ and $v$ are $a_1'$ and $b_1'$.


We call the above the \emph{deletion-type tangent searching procedure}, which takes $O(1+k_1+k_2)$ time, where $k_1$ is the number of points on $\conv(S_a)$ strictly between $a_1$ and $a_1'$, and $k_2$ is the number of points on $\conv(S_b')$ strictly between $t$ and $b_1'$. 
We say that these $k_1 + k_2$ points are \emph{involved} in the procedure.

    

\textbf{Antipodal pair searching}
Now we find the antipodal pairs of $S_a$ with respect to $\theta_1'$. Observe that the antipodal pairs of $S_a$ with respect to $\theta_2$  remain unchanged as $\ell_2$ does not change.
Let $q$ be the point such that $(a_1, q)$ is the antipodal pair of $S_a$ with respect to $\theta_1$ (see Figure~\ref{fig:deletions}).

For orientation $\theta_1'$, observe that $a_1'$ is one point of the antipodal pair of $S_a$. 
We find the other point by traversing $\conv(S_a)$ clockwise from $q$. 
Let $q'$ denote the new point such that $(a_1', q')$ is the antipodal pair of $S_a$ with respect to $\theta_1'$. 
We refer to this as the \emph{deletion-type antipodal-pair searching procedure}, which takes $O(1 + k')$ time, where $k'$ is the number of vertices on $\conv(S_a)$ strictly between $q$ and $q'$.
We say that these vertices are \emph{involved} in the procedure.

\textbf{Testing the width condition }
If $S_a$ does not dominate $S_b'$, there is nothing left to do. 
This can be easily detected by Observation~\ref{obs:domination}.
Otherwise, we need to test whether $\width_{[\theta_1', \theta_2]}(S_a) \le \omega$.
By Lemma~\ref{lem:width_test_obs}, it suffices to determine whether $\width_{[\theta_1', \theta_1]}(S_a) \le \omega$.
We do this by following the antipodal pairs of $S_a$ from the one corresponding to $\theta_1'$ to that of $\theta_1$.
Recall that $(a_1', q')$ and $(a_1, q)$ are the antipodal pairs of $S_a$ with respect to $\theta_1'$ and $\theta_1$, respectively. 
By traversing $\conv(S_a)$ clockwise starting from $a_1$ and $q$ simultaneously, we can visit all antipodal pairs corresponding to the range of orientations $[\theta_1', \theta_1]$.
Following them sequentially allows us to find an orientation $\theta \in [\theta_1', \theta_1]$ that minimizes $\width_\theta(S_a)$.
This takes $O(1 + k_1 + k')$ time, where, as before, $k_1$ is the number of points on $\conv(S_a)$ strictly between $a_1$ and $a_1'$, 
and $k'$ is the number of vertices on $\conv(S_a)$ strictly between $q$ and $q'$.
Therefore, this is bounded by the total time required by the deletion-type tangent searching procedure and the deletion-type antipodal-pair searching procedure.

\medbreak
The total time for handling point deletions as $\sigma_\omega$ moves downward depends on the two procedures: the deletion-type tangent searching procedure and the deletion-type antipodal-pair searching procedure. 
We conclude that it takes $O(n)$ time in total by showing that each input point can be involved at most twice in each of the procedures.

\begin{figure}[ht]
    \centering
    \includegraphics[width=0.5 \textwidth]{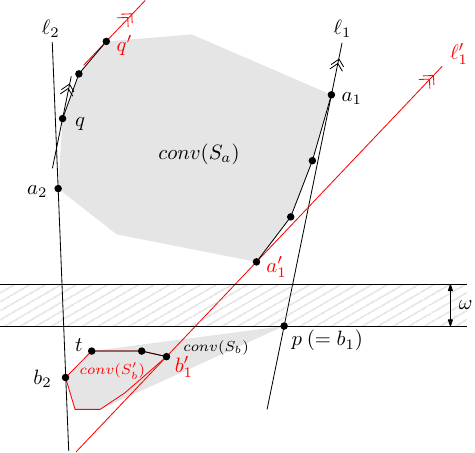}
    \caption{Changes after the deletion of $p$ are shown in red: $\conv(S_b')$, $a_1'$, $b_1'$, and $\ell_1'$. Observe that $(a_1', q')$ becomes the antipodal pair of $S_a$ with respect to $\theta_1'$ after the deletion.
    Observe that three points ($k_1 = 2, k_2 = 1$) are involved in the deletion-type tangent searching procedure and one point is involved in the deletion-type antipodal-pair searching procedure.}
    \label{fig:deletions}
\end{figure}

\begin{lemma}
    \label{lem:deletion_time}
    Each point of $S$ can be involved in the deletion-type tangent searching procedure and the deletion-type antipodal-pair searching procedure at most twice, respectively, in the entire algorithm.
\end{lemma}

\begin{proof}
    For a point $p$ deleted from $S_b$, recall that $k_1$ denotes the number of points on $\conv(S_a)$ strictly between $a_1$ and $a_1'$, and $k_2$ is the number of points on $\conv(S_b')$ strictly between $t$ and $b_1'$, where $t$ is the counterclockwise neighbor of $p$ on $\conv(S_b)$.

    We first show that each input point is involved in the procedure as one of the $k_1$ points at most once. 
    We use the following observation:
    when traversing the points on $\conv(S_a)$ in clockwise order from $a_1$ to $a_1'$, their $y$-coordinates decrease.
    Let $s$ be a point on $\conv(S_a)$ that lies strictly between $a_1$ and $a_1'$, that is, $y(a_1) > y(s) > y(a_1')$.
    Assume to the contrary that $s$ is involved again in the same procedure when another point $t$ is inserted after $p$.
    Let $u_1$ and $u_1'$ denote the original and new points of tangency of $S_a$ (of that moment) defined analogously for this insertion.
    By the above observation, we have $y(u_1) > y(s) > y(u_1')$.
    Moreover, since the new point of tangency is found by traversing the convex hull in clockwise order, 
    it holds that $y(a_1) > y(a_1') \ge y(u_1) > y(u_1')$.
    This contradicts that $y(a_1) > y(s) > y(a_1')$.

    We now show that each input point is involved in the procedure as one of the $k_2$ points at most once. This follows directly from the following observation:
    A point involved in this procedure (as one of the $k_2$ points) becomes a vertex of the convex hull at the moment when $p$ is deleted and remains on the convex hull until it is deleted during the sweeping process.
    Once a point is involved in this procedure, it cannot be involved again, since after being deleted from the convex hull, it cannot appear on it again.

    It now remains to show that each input point is involved in the deletion-type antipodal-pair searching procedure (i.e., as one of the $k'$ points) at most once.
    Recall that $(a_1, q)$ and $(a_1', q')$ are the antipodal pairs of $S_a$ with respect to $\theta_1$ and $\theta_1'$, respectively.
    We use a similar observation with the case for $k_1$: when traversing the points on $\conv(S_a)$ in clockwise order from $q$ to $q'$, their $y$-coordinates increase.
    Let $s$ be a point on $\conv(S_a)$ that lies strictly between $q$ and $q'$, that is, $y(q) < y(s) < y(q')$.
    Assume to the contrary that $s$ is involved again in the antipodal-pair searching procedure when another point $t$ is inserted after $p$.
    Analogously, let $v$ and $v'$ denote the points such that $v$ is a point of the original antipodal pair and $v'$ is a point of the new antipodal pair for this insertion.
    By the above observation, we have $y(v) < y(s) < y(v')$.
    Moreover, since a point of the new antipodal pair is found by traversing the convex hull in clockwise order, 
    it holds that $y(q) < y(q') \le y(v) < y(v')$.
    This contradicts $y(q) < y(s) < y(q')$.    
\end{proof}

Observe that the initialization of all data structures can be done in $O(n)$ time by Lemma~\ref{lem:dyn_conv}.
By Lemma~\ref{lem:insertion_time} and~\ref{lem:deletion_time}, the total time to process all events and to evaluate the condition $\width_{[\theta_1', \theta_1]}(S_a) \le \omega$ as $\sigma_\omega$ moves downward is $O(n)$. Now we conclude the following theorem.

\begin{theorem}
    \label{thm:one_fixed_decision}
    Given a set $S$ of $n$ points in the plane, sorted by their $y$-coordinates, and a real value $\omega > 0$, we can decide in $O(n)$ time whether there exists a pair of slabs of width $\omega$ whose union covers $S$ and one of which is horizontal.
\end{theorem}

\subsection{Optimization Algorithm}
\label{sec:one_fixed_optimization}

We give a brief summary of our approach to obtain Theorem~\ref{thm:one_fixed}. A complete proof is given in Appendix~\ref{sec:one_fixed_appendix}. The optimality of our algorithm is given in Appendix~\ref{sec:lowerbound}.
We first sort the points of $S$ so that $s_1, \ldots, s_n$ are ordered in decreasing order of their $y$-coordinates.
Let $W$ denote the set of all pairwise differences between the $y$-coordinates of points in $S$.
Since $W$ can be represented as a sorted matrix~\cite{frederickson1982complexity}, we can identify two consecutive values $w_0, w_1 \in W$ such that $w_0 < w^* \le w_1$ in $O(n \log n)$ time, by combining our decision algorithm described in Section~\ref{sec:one_fixed_decision} with an efficient selection algorithm for sorted matrices~\cite{frederickson1984generalized}.

Observe that the sequence of changes in the sets $S_a$ and $S_b$ that occur as $\sigma_w$ moves downward is the same for any $w_0 <w<w_1$.
Using this observation, we show that $w^*$ can be computed by tracking ${width}(S_a \cup S_b)$ as $\sigma_w$ moves downward, for any $w_0 < w < w_1$. 
We track ${width}(S_a \cup S_b)$ by following essentially the same procedure as in the decision algorithm. 
The authors of~\cite{ahn2024constrained} also use this idea, but their algorithm requires an additional data structure, which results in an extra logarithmic factor.

\begin{restatable}{theorem}{onefixed}
    \label{thm:one_fixed}
    Given a set $S$ of $n$ points in the plane and an orientation $\theta$, we can find a pair of slabs such that one of them has orientation $\theta$ and their union covers $S$ while minimizing the maximum of their widths in $O(n\log n)$ time.
\end{restatable}

\subsection{Approximation Algorithm}
\label{sec:one_fixed_approximation}

Recall that in Section~\ref{sec:eps_certificate}, we showed that an $\eps$-certificate $Q \subseteq S$ of size $O(1/\eps^2)$ can be computed in $O(n)$ time. 
After obtaining $Q$, we find a pair of slabs $(\sigma_1, \sigma_2)$ such that their union covers $Q$, $\sigma_1$ has orientation $\theta$, and the maximum of their widths is minimized.
Without loss of generality, we assume that $\sigma_1$ and $\sigma_2$ have the same width (otherwise, we expand the smaller one).
Since $Q \subseteq S$, it follows that $w(\sigma_1) = w(\sigma_2) \le w^*$. 
Such a pair can be found in $O(|Q|\log |Q|)$ time by Theorem~\ref{thm:one_fixed}. 

Let $(\sigma_1', \sigma_2')$ be the pair of slabs obtained by taking the $\eps$-expansions of $\sigma_1$ and $\sigma_2$, respectively. 
By the definition of an $\eps$-expansion, $\sigma_1'$ has orientation $\theta$ and $w(\sigma_1')= w(\sigma_2') \le (1+\eps)w^*$. 
Furthermore, by the definition of an $\eps$-certificate, $\sigma_1' \cup \sigma_2'$ covers $S$. 
Therefore, $(\sigma_1', \sigma_2')$ forms a $(1+\eps)$-approximate solution. 

It takes $O(n + |Q|\log |Q|)$ time to obtain $(\sigma_1', \sigma_2')$. 
Let $\mu = \min \{n, 1/\eps^2\}$, then we have $|Q| = O(\mu)$. 
Therefore, the total running time is $O(n + \mu \log \mu)$, which can be expressed as both $O(n \log (1/\eps))$ and $O(n + (1/\eps^2) \log (1/\eps))$. 
We now conclude the following theorem.

\begin{restatable}{theorem}{onefixedapx}
    \label{thm:one_fixed_apx}
    Given a set $S$ of $n$ points in the plane and an orientation $\theta$, we can find a pair of slabs such that their union covers $S$, one of them has orientation $\theta$, and the maximum of their widths is at most $(1+\eps)w^*$ in $O(n+ \mu\log \mu)$ time, where $\mu = \min \{n, 1/\eps^2\}$.  
\end{restatable}

\section{General Two-Line-Center}
\label{sec:general}

Throughout this section, we denote by $\Sigma^* = (\sigma_1^*, \sigma_2^*)$ an optimal pair of slabs for the general two-line-center problem. Let $w^* = \max \{w(\sigma_1^*), w(\sigma_2^*) \}$.


We begin by proving Lemma~\ref{lem:angles}, which will be used to compute a candidate set of orientations for one of the slabs in a $(1+\eps)$-approximate solution.

\begin{lemma}
    \label{lem:angles}
    Let $p$ and $q$ be two points in $\mathbb R^2$, let $\widetilde w > 0$ be  a real number, let $c\ge 1$ be a constant, and let $\eps> 0$ be a real number.
    In $O(1/\eps)$ time, we can compute a set $\Gamma$ of $O(1/\eps)$ orientations, such that the following holds:\\
    For every finite set $P$ of points in $\mathbb R^2$ and for every slab $\sigma$, if 
    
    \mytxt{(1)} $P\subset \sigma$, 
    \mytxt{(2)} $p, q \in \sigma$, 
    \mytxt{(3)} $d(p,q) \ge \diam(P)/4$, and
    \mytxt{(4)} $w(\sigma) \le \widetilde w \le c\cdot w(\sigma)$,
    \\then there exists an orientation $\gamma\in \Gamma$ and a slab $\sigma'$, such that
    
    \mytxt{(1)} $P \subset \sigma'$, 
    \mytxt{(2)} $w(\sigma') \le (1+ \eps) \cdot w(\sigma)$, and
    \mytxt{(3)} $\sigma'$ has the orientation $\gamma$.
\end{lemma}

\begin{proof}
    We remark that this proof is based on the approach in~\cite{agarwalTwo}.
    Without loss of generality, assume that $p$ and $q$ lie on a horizontal line. 
    Let $D = d(p, q)$. 
    Let $\theta$ be the orientation with $0< \theta \le \pi/2$ such that $\sin\theta = \min \{1, \widetilde w/{D} \}$.
    Let $\delta =  c_\delta \eps$ for some constant $0 < c_\delta < 1$. 
    Let $\Gamma = \{\gamma_i = (i- \lceil \frac{1}{\delta}\rceil)\cdot \delta\theta  \mid 0 \le i \le 2\lceil \frac{1}{\delta}\rceil \}$.
    Observe that $\Gamma$ contains uniformly placed orientations in the range $[-\theta, \theta]$, and $|\Gamma| = O(1/\eps)$. (Here we slightly abuse the notation by redefining orientations to lie in $[-\pi/2, \pi/2]$ instead of $[0, \pi]$.)
    In the following, we prove that $\Gamma$ satisfies the conditions in the statement.

    Let $P$ be a set of points in $\mathbb{R}^2$ and let $\sigma$ be a slab, where $P$ and $\sigma$ satisfy the four conditions in the statement.
    Let $\alpha$ be the orientation of $\sigma$, and let $w = w(\sigma)$. 
    For simplicity, we assume that $0\le \alpha < \pi/2$. The other case can be shown in a symmetric way. 
    Since $p, q \in \sigma$, we have $\sin \alpha \le w/D \le \widetilde w/D$ (see the right triangle having $\overline{pq}$ as one of its sides in Figure~\ref{fig:angle_lemma}).
    Observe that if $\widetilde w/ D < 1$, then $\sin \theta = \widetilde w/ D$. If $\widetilde w/ D \ge 1$, then $\sin \theta =1$. In both cases, it holds that $\sin \alpha \le \sin \theta$. Therefore, we have $\alpha \le \theta$.

     Let $\gamma$ be the largest orientation in $\Gamma$ such that $\gamma \le \alpha$. 
     Since $ \alpha \le  \theta$, we have $\alpha - \gamma < \delta \theta$. 
     Let $R$ be the minimum bounding box of $P$ with orientation $\alpha$. See Figure~\ref{fig:angle_lemma}.
     Denote by $o$ and $t$ the lowest and highest vertices of $R$, respectively. 
    Now consider a slab $\sigma'$ of orientation $\gamma$ whose two bounding lines pass through $o$ and $t$. 
    Clearly, $P \subset \sigma'$. 
    We next show that $w(\sigma') \le (1+\eps)w$, which completes the proof.

    Let $b$ be a vertex of $R$ other than $o$ and $t$. 
    Consider the line $\ell$ orthogonal to $\sigma'$ that passes through $b$. 
    The line $\ell$ intersects the two bounding lines of $\sigma'$. 
    Observe that one intersection point lies inside $\sigma$, while the other lies outside. 
    Let $a$ and $c$ denote these two intersection points, with $a$ being the one inside $\sigma$. 
    Then $w(\sigma') = \overline{ab} + \overline{bc} \le w + \overline{bc}$.

     We have $\overline{bc} = \overline{oc}\cdot \tan(\alpha - \gamma) 
    \le \overline {oc} \cdot \tan (\delta \theta)$.
    We consider the following two cases.

     \textbf{Case 1. $\widetilde w / D \le 1/2$. } That is, $\theta \le  \pi/6$.
     In this case, since $\delta < 1$ and $\theta \le \pi/6$, we have 
     $\tan(\delta \theta) 
    \le \delta \tan \theta$.
    Note that $\delta  \tan \theta  = \delta  \sin \theta/ \cos \theta $
    $\le (2/\sqrt 3)  \delta  \sin\theta \le (2/\sqrt 3)  \delta  \widetilde w /D \le (2/\sqrt 3) c \delta  w /D  $.
    Since $\overline{oc} \le {diam}(P) \le 4D$, it follows that $\overline{bc} \le (8/\sqrt 3) c \delta  w $.

    \textbf{Case 2. $\widetilde w / D > 1/2$. }
    We have
    $\tan(\delta \theta) 
    = \sin (\delta \theta) / \cos (\delta \theta) 
    \le 2 \cdot \sin (\delta \theta) \le 2 \delta \theta \le \pi \delta $  
    since $\delta \theta \le \pi/3$ assuming $\delta \le 2/3$.
    Also, $\overline{oc} \le {diam}(P) \le 4D \le 8\widetilde w \le 8 c w$.  
    Combining these bounds gives $\overline{bc} \le 8\pi \delta c w$. 
    
    By choosing $\delta = \min \{ 2/3, \sqrt 3 \eps /(8c) ,  \eps / (8\pi c)\}$, it holds that $\overline{bc} \le \eps w$ in both cases. 
\end{proof}

\begin{figure}[ht]
    \centering
    \includegraphics[width=0.45 \textwidth]{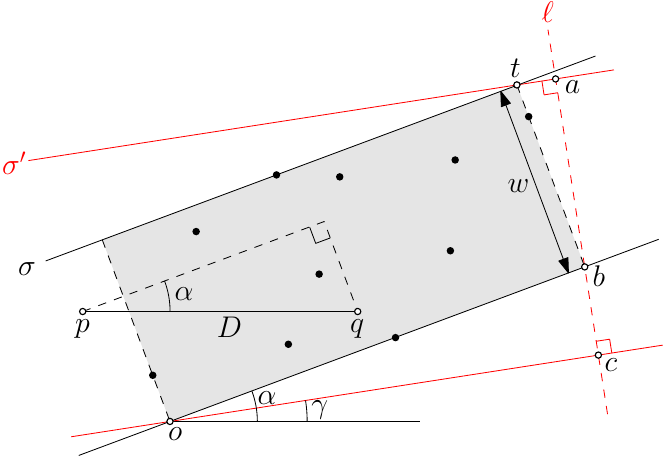}
    \caption{The black points represent the points in $P$ and the gray region represents the minimum bounding box of $P$ with orientation $\alpha$.}
    \label{fig:angle_lemma}
\end{figure}

By Lemma~\ref{lem:anchor_pair}, we can compute a set $F$ of at most $11$ pairs such that one of them is an anchor pair of $\Sigma^*$ in $O(n)$ time.  
For the moment, assume that we have an anchor pair $(p, q)$ of $\Sigma^*$. (We will repeat the following process for every pair in $F$.)
Without loss of generality, assume that $p, q \in \sigma_1^*$ and $d(p, q) \ge {diam}((\sigma_1^* \setminus \sigma_2^*) \cap S)/4$.

By Theorem~\ref{thm:10apx}, we can compute $\widetilde w$ such that $w^* \le \widetilde w \le 10 w^*$ in $O(n)$ time.

Consider a slab ${\sigma}_0^*$ that has the same orientation as $\sigma_1^*$,  has width $w^*$, and contains $\sigma_1^*$.  
In particular, if $w(\sigma_1^*) = w^*$, then ${\sigma}_0^*$ is $\sigma_1^*$.

Observe that $P = S \cap (\sigma_1^* \setminus \sigma_2^*)$ and $\sigma_0^*$ satisfy the four conditions in the statement of Lemma~\ref{lem:angles} for $(p, q)$ and the value $\widetilde w$ computed above, with $c = 10$.  
Therefore, in $O(1/\eps)$ time, we can compute a set $\Gamma$ of $O(1/\eps)$ orientations such that there exists a slab $\sigma'$ with orientation $\gamma \in \Gamma$ that covers $S \cap (\sigma_1^* \setminus \sigma_2^*)$ and satisfies $w(\sigma') \le (1+\eps) w^*$.  
This implies that there exists a pair of slabs whose union covers $S$, one of which has an orientation from $\Gamma$, and whose maximum width is at most $(1+\eps) w^*$.

For every $\gamma \in \Gamma$, we find a pair of slabs whose union covers $S$, one of which has orientation $\gamma$, and whose maximum width is minimized. 
By Theorem~\ref{thm:one_fixed}, this can be done in $O((n/\eps) \log n)$ time in total.  
We repeat this procedure for every pair in $F$, which adds only a constant factor to the running time.
Observe that each of the pairs of slabs we obtain has the property that its union covers $S$.
Among these pairs, we output the one with the smallest maximum width.  
Thus, we conclude the following lemma.

\begin{lemma}  
    \label{lem:general_temp}
    We can find a pair of slabs such that their union covers $S$ and the maximum of their widths is at most $(1+\eps)w^*$ in $O((n/\eps) \log n)$ time. 
\end{lemma}

To further refine the logarithmic term, we compute an $\eps$-certificate $Q$ of $S$ (see Section~\ref{sec:eps_certificate}). 
By Theorem~\ref{thm:computing_eps_certificate}, this can be done in $O(n)$ time and $|Q| = O(1/\eps^2)$. 
We then compute a $(1+\eps)$-approximate pair $(\sigma_1, \sigma_2)$ for the general two-line-center problem on $Q$, which can be done in $O((|Q|/\eps)\log |Q|)$ time by Lemma~\ref{lem:general_temp}.
Without loss of generality, we assume that $\sigma_1$ and $\sigma_2$ have the same width (otherwise, we expand the smaller one).
Let $w = w(\sigma_1) =  w(\sigma_2)$. 
Since $Q \subseteq S$, it is clear that $w \le (1+\eps)w^*$.

Let $(\sigma_1', \sigma_2')$ be the pair of slabs obtained by taking the $\eps$-expansions of $\sigma_1$ and $\sigma_2$, respectively. 
By definition, $\sigma_1' \cup \sigma_2'$ covers $S$, and 
$w(\sigma_1')= w(\sigma_2') \le (1+\eps)^2 w^* \le (1 +3\eps) w^*$.

The total running time is 
$O(n + (\mu/\eps)\log \mu )$, where $ \mu = \min\{n,\, 1/\eps^2\}$, 
which can also be expressed as both $O((n/\eps)\log(1/\eps))$ and $O(n + (1/\eps^3)\log(1/\eps))$. 
By replacing $\eps$ with $\eps/3$, we obtain the following theorem.

\begin{theorem}
    Given a set $S$ of $n$ points in the plane, we can find a pair of slabs such that their union covers $S$ and the maximum of their widths is at most $(1+\eps)w^*$ in $O(n + (\mu/\eps) \log \mu)$ time, where $\mu = \min \{n, 1/\eps^2 \}$.
\end{theorem}

\section{Two Fixed Orientations}
\label{sec:two_fixed}


By rotating the coordinate system, we assume without loss of generality that one of the given orientations is horizontal. Thus, the two orientations can be taken as $0$ and $\theta$, where $0 \leq \theta < \pi$.
Let $(\sigma_h^*, \sigma_t^* )$ be an optimal pair of slabs where $\sigma_h^*$ is horizontal and $\sigma_t^*$ has orientation $\theta$.
Let $w^* = \max \{w(\sigma_h^*), w(\sigma_t^*) \}$.
Recall that $S$ denote the input set of $n$ points in the plane.


\subsection{A \texorpdfstring{$2$}{2}-Approximation Algorithm}
\label{sec:2_apx_two_fixed}

Let $p_t$ and $p_b$ denote the highest and lowest points in $S$, respectively.
Let $p_\ell$ and $p_r$ be the two extreme points in $S$ in the direction orthogonal to $\theta$. See Figure~\ref{fig:parallelogram}. 
Since we assume $w^* > 0$, the size of $\{p_b, p_t, p_\ell, p_r\}$ is at least two.
Then the following holds.

\begin{lemma}
    \label{lem:two_fixed_extreme_pts}
    There exists an optimal pair $(\sigma_h^*, \sigma_t^*)$ of slabs such that two distinct points $p$ and $q$ from $\{p_b, p_t, p_\ell, p_r\}$ satisfy $p \in \sigma_h^*$ and $q \in \sigma_t^*$.
    (Note that $p$ or $q$ may belong to $\sigma_h^* \cap \sigma_t^*$.)
\end{lemma}
\begin{proof}
    Suppose to the contrary. 
    Then for every optimal pair $(\sigma_h^*, \sigma_t^*)$ of slabs, either $\{p_b, p_t, p_\ell, p_r\} \subset \sigma_h^* \setminus \sigma_t^*$ or $\{p_b, p_t, p_\ell, p_r\} \subset \sigma_t^* \setminus \sigma_h^*$.
    By definition, the points $\{p_b, p_t, p_\ell, p_r\}$ determine the minimum bounding parallelogram of $S$; two of its sides are horizontal and the other two sides have orientation~$\theta$ (see~Figure~\ref{fig:parallelogram}). 
    Let $R$ denote the parallelogram. 
    
    If $\{p_b, p_t, p_\ell, p_r\} \subset \sigma_h^* \setminus \sigma_t^*$, then $R \subset \sigma_h^*$, and hence $S \subset \sigma_h^*$.
    In this case, choose any point $p \in \{p_b, p_t, p_\ell, p_r\}$ and consider the line $\ell_p$ of orientation $\theta$ passing through $p$. 
    Then $(\sigma_h^*, \ell_p)$ forms an optimal pair of slabs, with one slab degenerating to a line.
    This contradicts our assumption, since any point in $\{p_b, p_t, p_\ell, p_r\}\setminus\{p\}$ lies in $\sigma_h^*$ while $p \in \ell_p$.
    
    Analogously, if $\{p_b, p_t, p_\ell, p_r\} \subset \sigma_t^* \setminus \sigma_h^*$, consider the horizontal line $\ell_p$ passing through $p \in \{p_b, p_t, p_\ell, p_r\}$.
    Then $(\ell_p, \sigma_t^*)$ forms an optimal pair of slabs, which contradicts our assumption in the same way.
\end{proof}

\begin{figure}[ht]
\centering
\includegraphics[width=0.23 \textwidth]{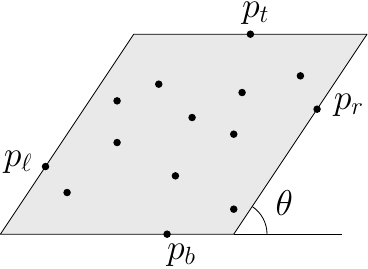}
\caption{The four extreme points $\{ p_b, p_t, p_\ell, p_r \}$ determine the minimum bounding parallelogram. }
\label{fig:parallelogram}
\end{figure}


Using Lemma~\ref{lem:two_fixed_extreme_pts}, we first give a $2$-approximation algorithm.

\begin{lemma}
    \label{lem:two_fixed_const_apx}
    We can compute a $2$-approximation of an optimal pair of slabs in $O(n)$ time. 
\end{lemma}

\begin{proof}
    For the moment, 
    assume that we are given a pair $(p_1, p_2)$ of points in $S$ such that  
    $p_1 \in \sigma_h^*$ and $p_2\in \sigma_t^*$.
    Let $\ell_h$ be the horizontal line through $p_1$, and let $\ell_t$ be the line of orientation $\theta$ through $p_2$. 
    Then we have $\ell_h \subset \sigma_h^*$ and $\ell_t \subset  \sigma_t^*$. 

    Now we determine by which slab each input point should be covered in a greedy way. 
    We initialize two sets $H$ and $T$ as empty sets. 
    For each $p\in S$,  we insert $p$ into $H$ if $d(p, \ell_h) < d(p, \ell_t)$. 
    Otherwise, we insert it into $T$. 
    After applying this process to every input point, we compute the minimum horizontal slab enclosing $H$, denoted by $\sigma_h$, and the minimum slab of orientation $\theta$ enclosing $T$, denoted by $\sigma_t$ (see Figure~\ref{fig:two_fixed_const}).
    Now we show that $(\sigma_h, \sigma_t)$ is a $2$-approximate solution.

    We first show that $ w(\sigma_h) \le 2 w^*$.
    Let $q$ and $q'$ be the lowest and highest points in $H$ (see Figure~\ref{fig:two_fixed_const}). 
    Then $w(\sigma_h) \le d(q, \ell_h) + d(q', \ell_h)$.  (Equality holds if $q$ and $q'$ are on different sides of $\ell_h$.)
    If $q\in \sigma_h^*$, then $d(q, \ell_h) \le w^*$ as $\ell_h \subset \sigma_h^*$. 
    Otherwise, $ d(q, \ell_t) \le w^*$  as $q$ should be covered by $\sigma_t^*$.
    We also have $d(q, \ell_h) < d(q, \ell_t)$ as $q\in H$. Therefore it holds that $d(q, \ell_h) \le w^*$ in any case. 
    Similarly, we can also show that $d(q', \ell_h) \le w^*$. 
    Therefore, $w(\sigma_h) \le d(q, \ell_h) + d(q', \ell_h) \le 2w^*$. We can similarly show that $w(\sigma_t) \le 2w^*$ by choosing the two extreme points of $T$ in the direction orthogonal to $\theta$. 
    Note that once we have $p_1$ and $p_2$, we can find $\sigma_h$ and $\sigma_t$ in $O(n)$ time by finding the extreme points in $H$ and $T$, respectively. 

    By Lemma~\ref{lem:two_fixed_extreme_pts}, 
    we have a constant number of candidates for $(p_1, p_2)$. 
    We apply the above procedure to each candidate pair and select the one that minimizes $\max\{w(\sigma_h), w(\sigma_t)\}$.
    The resulting pair $(\sigma_h, \sigma_t)$ is a $2$-approximate solution.
\end{proof}

    \begin{figure}[ht!]
    \centering
    \includegraphics[width=0.45 \textwidth]{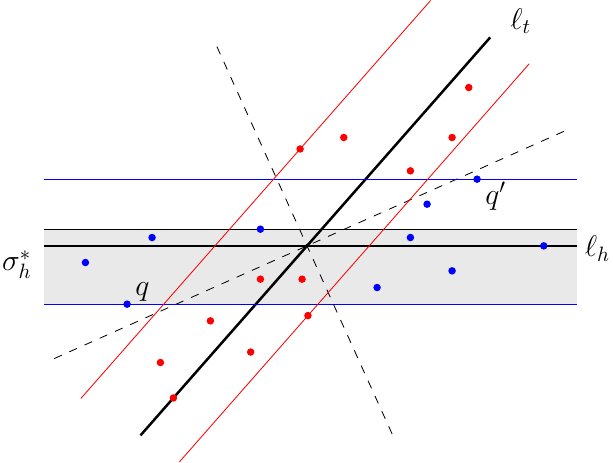}
    \caption{Two lines $\ell_h$ and $\ell_t$ and their bisectors (dashed).
    Blue points belong to $H$, while red points belong to $T$. 
    The blue lines indicate the bounding lines of $\sigma_h$, and the red lines indicate the bounding lines of $\sigma_t$.  
    Points $q$ and $q'$ are the vertical extreme points of $H$. 
    The gray region represents $\sigma_h^*$. 
    Note that $q \in \sigma_h^*$ and $q' \notin \sigma_h^*$.}
    \label{fig:two_fixed_const}
    \end{figure}

\subsection{A \texorpdfstring{$(1+ \eps)$}{(1+epsilon)}-Approximation Algorithm}
\label{sec:eps_apx_two_fixed}

By Lemma~\ref{lem:two_fixed_const_apx}, we can compute $\widetilde w$ such that $w^* \le \widetilde w \le 2w^*$ in $O(n)$ time. 
By Lemma~\ref{lem:two_fixed_extreme_pts}, 
we can compute a set $G$ of at most six pairs of points such that there exists a pair $(p_1, p_2) \in G$ where  
$p_1 \in \sigma_h^*$ and $p_2\in \sigma_t^*$ in $O(n)$ time.
For the moment, assume that we have such a pair $(p_1, p_2)$ of $G$. (We will repeat the following process for every pair in $G$.)

Let $\ell_h$ be the horizontal line through $p_1$, and let $\ell_t$ be the line of orientation $\theta$ through $p_2$. 
Then we have $\ell_h \subset \sigma_h^*$ and $\ell_t \subset  \sigma_t^*$. 
Without loss of generality, we assume that $\ell_h$ is $y=0$.
We also assume that $w(\sigma_h^*) \ge w(\sigma_t^*)$. We can handle the other case by running the same algorithm in the coordinate system rotated by $\theta$.

Let $L = \{\ell_i: y = ({\eps \widetilde w}/{4})\cdot (i- 1 -\lceil {4}/{\eps} \rceil )  
\mid 1 \le i \le 2\lceil 4/\eps \rceil +1 \}$ be the set of horizontal lines.
Note that $\ell_1$ is $y = - ({\eps \widetilde w}/{4})\cdot \lceil {4}/{\eps} \rceil \le -\widetilde w $ 
and $\ell_{2\lceil 4/\eps \rceil +1 }$ is $y = ({\eps \widetilde w}/{4})\cdot \lceil {4}/{\eps} \rceil \ge \widetilde w $. 
The distance between two consecutive lines is  $\eps \widetilde w / 4$ and $|L| =  2\lceil 4/\eps \rceil +1$. 

For two indices $ 1 \le i \le j \le 2\lceil 4/\eps \rceil +1 $, 
let $\sigma(i, j)$ be the horizontal slab bounded by $\ell_i$ and $\ell_j$ (see Figure~\ref{fig:two_fixed_gridlines}). 
If $i=j$, then $\sigma(i, j)$ is $\ell_i$, which covers the points lying on it.

\begin{figure}[ht]
\centering
\includegraphics[width=0.4 \textwidth]{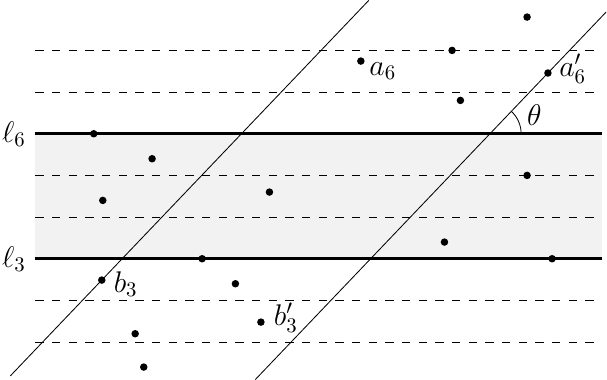}
\caption{The lines in $L$ (dashed), $\ell_3$, and $\ell_6$ are shown. The gray region represents $\sigma(3,6)$.
Points $a_6$ and $a'_6$ (resp., $b_3$ and $b_3'$) are the two extreme points in the direction orthogonal to $\theta$ among the points lying strictly above $\ell_6$ (resp., strictly below $\ell_3$).
Note that the minimum slab of orientation $\theta$ enclosing the points outside $\sigma(3, 6)$ is determined by $a_6'$ and $b_3$. 
}
\label{fig:two_fixed_gridlines}
\end{figure}

\begin{lemma}
    \label{lem:two_fixed_gridlines}
    There is a horizontal slab $\sigma$ such that  
    both bounding lines are in $L$, 
    $\sigma_h^* \subseteq \sigma $, and $w(\sigma) \le (1 + \eps) w^*$.
\end{lemma}
\begin{proof}
    Let $i'$ and $j'$ be the indices such that $\sigma(i', j')$ is the minimal slab that encloses $\sigma_h^*$ among all slabs $\sigma(i,j)$ with $1 \le i \le j \le 2\lceil 4/\eps \rceil + 1$.
    Such indices exist because $d(\ell_1, \ell_h) \ge \widetilde w \ge w^*$ and $d(\ell_{2 \lceil 4/\eps \rceil +1 }, \ell_h)\ge \widetilde w \ge w^*$, and $\ell_h \subset \sigma_h^*$. 
    Then $\sigma_h^* \subseteq \sigma(i', j')$ and $w(\sigma(i', j')) \le w(\sigma_h^*) + \eps \widetilde{w}/2$. 
    Since $w(\sigma_h^*) \le w^*$ and $\widetilde w \le 2w^*$, it follows that $w(\sigma(i', j')) \le (1+\eps) w^*$. 
\end{proof}

Observe that once we find a slab $\sigma$ satisfying $\sigma_h^* \subseteq \sigma$ and 
$w(\sigma) \le (1+\eps)w^*$, the other slab is determined immediately: it is the minimum-width slab 
$\sigma'$ of orientation $\theta$ that encloses all points not covered by $\sigma$.  
Note that $(\sigma, \sigma')$ forms a $(1+\eps)$-approximate solution.

By Lemma~\ref{lem:two_fixed_gridlines},   there are at most $|L|^2 = O(1/\eps^2)$ candidates for such a slab $\sigma$.  
However, we preprocess the input points so that we avoid testing all $O(1/\eps^2)$ candidates.

\paragraph*{Preprocessing}
The lines in $L$ partition the plane into $O(1/\eps)$ regions. 
For each point $p \in S$, we determine the region containing $p$, which can be done in constant time using the floor function.  
Within each region, we find the two extreme points in the direction orthogonal to $\theta$. 
This takes $O(n)$ time in total. 

For each line $\ell_i \in L$, we store the two extreme points in the direction orthogonal to $\theta$ for the points that lie strictly above $\ell_i$, and denote them by $a_i$ and $a'_i$ (see Figure~\ref{fig:two_fixed_gridlines}).
Similarly, for each line $\ell_i \in L$, we store the two extreme points for those lying strictly below $\ell_i$, and denote them by $b_i$ and $b'_i$.
This can be done in $O(n)$ total time by visiting the partitioned regions in order.

\paragraph*{Algorithm} 
For fixed $i$ and $j$, $ 1 \le i \le j \le 2\lceil 4/\eps \rceil +1  $, note that the smallest slab of orientation $\theta$ that encloses the input points strictly outside of $\sigma(i, j)$ is the smallest slab of orientation $\theta$ that encloses $\{a_i, a'_i, b_j, b'_j \}$ (see Figure~\ref{fig:two_fixed_gridlines}). 
That is, $\width_\theta(S \setminus \sigma(i, j)) = \width_\theta (\{a_i, a'_i, b_j, b'_j \})$.
(Recall the definition of $\width_\theta$ in Section~\ref{sec:preliminaries}.)
We can compute $\width_\theta(S \setminus \sigma(i, j))$ in constant time once we know $\{a_i, a'_i, b_j, b'_j \}$. 

For  $  1 \le i, j \le 2\lceil 4/\eps \rceil +1$, 
we define $M[i, j] = d(\ell_i, \ell_j)$ if $i\le j$ and 
$d(\ell_i, \ell_j) \ge w_\theta(S \setminus \sigma(i, j))$. 
Let $M[i, j] = -\infty$ if $i>j$ or $d(\ell_i, \ell_j) < w_\theta(S \setminus \sigma(i, j))$. 
We first show that $M$ is a sorted matrix in the following lemma. 
A smallest non-negative element in a sorted matrix can be found in time that is linear in the number of rows plus the number of columns if each element can be computed in constant time.
Observe that each element of $M$ can be computed in constant time by the preprocessing step.

\begin{lemma}
    $M$ is a sorted matrix. In particular, each row is a non-decreasing sequence and each column is a non-increasing sequence. 
\end{lemma}
\begin{proof}

    Let $t = 2\lceil 4/\eps \rceil + 1$, which is the number of rows (resp., the number of columns) of $M$.

    Fix an index $i$ with $1 \le i \le t$.  
    For every $j$ with $1 \le j < i$, we have $M[i,j] = -\infty$.
    Let $k$ be the smallest index with $i \le k \le t$ such that $d(\ell_i, \ell_k) \ge w_\theta(S \setminus \sigma(i,k))$.  
    If no such $k$ exists, then $M[i,j] = -\infty$ for all $j \ge i$.  
    If such a $k$ exists, then $M[i,j] = -\infty$ for all $j$ with $i \le j < k$, and $M[i,j] = d(\ell_i, \ell_j)$ for all $j$ with $k \le j \le t$.  
    In both cases, the $i$-th row of $M$ is a non-decreasing sequence.
    
    Fix an index $j$ with $1 \le j \le t$.  
    For every $i$ with $j < i$, we have $M[i,j] = -\infty$.
    Let $k'$ be the largest index with $1 \le k' \le j$ such that $d(\ell_{k'}, \ell_j) \ge w_\theta(S \setminus \sigma(k',j))$.  
    If no such $k'$ exists, then $M[i,j] = -\infty$ for all $i \le j$.  
    If such a $k'$ exists, then $M[i,j] = -\infty$ for all $i$ with $k' < i \le j$, and $M[i,j] = d(\ell_i, \ell_j)$ for all $i$ with $1 \le i \le k'$.  
    In both cases, the $j$-th column of $M$ is a non-increasing sequence.
\end{proof}

Let $(i', j')$ be the pair of indices such that $\sigma(i', j')$ is the minimal slab that encloses $\sigma_h^*$ among all slabs $\sigma(i,j)$ with 
$1 \le i \le j \le 2\lceil 4/\eps \rceil + 1$.
By our assumption that $w(\sigma_h^*) \ge w(\sigma_t^*)$, we have 
$d(\ell_{i'}, \ell_{j'}) \ge w(\sigma_h^*) \ge w(\sigma_t^*)  \ge w_\theta(S \setminus \sigma(i', j'))$. 
By Lemma~\ref{lem:two_fixed_gridlines}, 
$d(\ell_{i'}, \ell_{j'}) \le (1+\eps)w^*$. 
Therefore,
$w^* \le M[i', j'] = d(\ell_{i'}, \ell_{j'}) \le (1+\eps)w^*,$
which implies that a smallest non-negative element in $M$ exists and is at most $(1+\eps)w^*$.

We compute a pair $(i'', j'')$ of indices such that $M[i'', j'']$ is a smallest non-negative element in $M$. This can be done in $O(1/\eps)$ time.
This directly yields a pair of slabs: $\sigma(i'', j'')$ and the smallest slab of orientation $\theta$ that encloses $S \setminus \sigma(i'', j'')$.  
Observe that these two slabs together form a $(1+\eps)$-approximate pair of slabs.

Now, we conclude the following theorem.

\begin{theorem}
    \label{thm:two_fixed_apx}
    Given a set $S$ of $n$ points in the plane and two orientations $\theta_1$ and $\theta_2$, 
    we can find a pair of slabs, one with orientation $\theta_1$ and the other with orientation $\theta_2$, such that their union covers $S$ and their maximum width is at most $(1+\eps)w^*$ in $O(n + 1/\eps)$ time.
\end{theorem}

\section{Parallel Two-Line-Center}
\label{sec:parallel}


Throughout this section, let $\Sigma^* = (\sigma_1^*, \sigma_2^* )$ be an optimal pair of slabs for the parallel-two-line-center problem.
Let $w^* = \max \{w(\sigma_1^*), w(\sigma_2^*) \}$.
Recall that $S$ denotes the input set of $n$ points in the plane. 


\textbf{Notation} 
Let $\rho$ be a fixed constant with $1/2\le \rho < 1$. 
Given a pair of interior-disjoint parallel slabs $\Sigma = (\sigma_1, \sigma_2)$, 
let $W(\Sigma)$ denote the largest distance between a boundary of $\sigma_1$ and a boundary of $\sigma_2$.
Let $g(\Sigma)$ denote the smallest distance between a boundary of $\sigma_1$ and a boundary of $\sigma_2$.
See Figure~\ref{fig:gap_ratio}.

\begin{definition}[Gap-ratio]
    Given a pair of interior-disjoint parallel slabs $\Sigma = (\sigma_1, \sigma_2)$, the \emph{gap-ratio} of $\Sigma$ is defined as $g(\Sigma)/ W(\Sigma)$. 
    If $W(\Sigma)= 0$, the gap-ratio is not defined.
\end{definition}

\begin{figure}[ht]
\centering
\includegraphics[width=0.32 \textwidth]{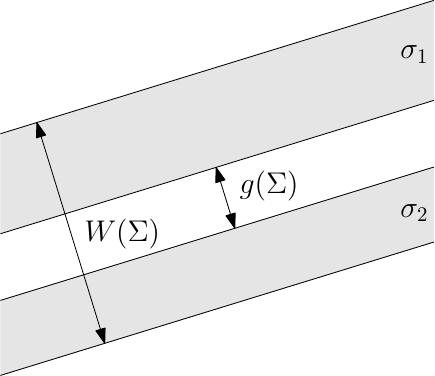}
\caption{For a pair of interior-disjoint parallel slabs $\Sigma = (\sigma_1, \sigma_2)$, we denote by $W(\Sigma)$ and $g(\Sigma)$ the largest and smallest distances, respectively, between a boundary of $\sigma_1$ and a boundary of $\sigma_2$.}
\label{fig:gap_ratio}
\end{figure}

To handle this problem, we consider two cases: whether the gap-ratio of $\Sigma^*$ is at most $\rho$, or not. 
In the following, we present two different algorithms, one for each case. 
In Sections~\ref{sec:small_gap_ratio} and~\ref{sec:large_gap_ratio}, we prove the corresponding theorems.

\begin{theorem}[Small gap-ratio]
    \label{thm:small_gap_ratio}
    Assuming that there exists an optimal pair $\Sigma^* = (\sigma_1^*, \sigma_2^*)$ for the parallel two-line-center problem whose gap ratio is at most $\rho$, 
    we can compute a pair of parallel slabs such that their union covers $S$ and the maximum of their widths is at most $(1+\eps)w^*$ in $O(n/\eps)$ time. 
\end{theorem}

\begin{theorem}[Large gap-ratio]
    \label{thm:large_gap_ratio}
    Assuming that there exists an optimal pair $\Sigma^* = (\sigma_1^*, \sigma_2^*)$ for the parallel two-line-center problem whose gap ratio is at least $\rho$, 
    we can compute a pair of parallel slabs such that their union covers $S$ and the maximum of their widths is at most $(1+\eps)w^*$ in $O(n+ \mu\log \mu)$ time, where $\mu = \min \{n, 1/\eps^2\}$. 
    If the assumption on the gap-ratio is not satisfied, the algorithm may return NO.
\end{theorem}

We first prove  Theorem~\ref{thm:parallel} by combining Theorem~\ref{thm:small_gap_ratio} and Theorem~\ref{thm:large_gap_ratio} with an $\eps$-certificate (see Section~\ref{sec:eps_certificate}).

\begin{theorem}
    \label{thm:parallel}
    Given a set $S$ of $n$ points in the plane, we can find a pair of parallel slabs such that their union covers $S$ and the maximum of their widths is at most $(1+\eps)w^*$~in $O(n + \mu /\eps)$ time, where $\mu = \min \{n, 1/\eps^2\}$.  
\end{theorem}

\begin{proof}
    We first compute an $\eps$-certificate $Q\subseteq S$ of $S$. 
By Theorem~\ref{thm:computing_eps_certificate}, this takes $O(n)$ time, and $|Q| = O(\mu)$ where $\mu = \min \{n, 1/\eps^2\}$. 
Let $w^*(Q)$ denote the maximum width of an optimal solution for the parallel two-line-center problem on $Q$. It is clear that $w^*(Q) \le w^*$ since~$Q\subseteq S$. 

We run the algorithms of Theorem~\ref{thm:small_gap_ratio} and Theorem~\ref{thm:large_gap_ratio} on $Q$. 
Observe that at least one of them returns a pair of parallel slabs whose union covers $Q$ and whose maximum width is at most $(1+\eps)w^*(Q)$.
Among the (at most) two pairs, let $(\sigma_1, \sigma_2)$ be the one whose maximum width is smaller.
Without loss of generality, assume that $\sigma_1$ and $\sigma_2$ have the same width (otherwise, we expand the smaller one). Therefore, $w(\sigma_1) = w(\sigma_2) \le (1+\eps)w^*(Q)$.

Let $(\sigma_1', \sigma_2')$ be the pair of slabs obtained by taking the $\eps$-expansions of $\sigma_1$ and $\sigma_2$, respectively.  
By the definition of an $\eps$-expansion, we have $w(\sigma_1') \le (1+\eps)w(\sigma_1) \le (1+\eps)^2w^*(Q) \le (1+\eps)^2w^*$. 
Therefore, $w(\sigma_1')\le (1+O(\eps)) w^*$, and the same bound holds for $w(\sigma_2')$.
Moreover, by the definition of an $\eps$-certificate, $\sigma_1' \cup \sigma_2'$ covers $S$.  
Therefore, $(\sigma_1', \sigma_2')$ forms a $(1+O(\eps))$-approximate solution.

The total running time is $O(n + \mu/\eps + \mu \log \mu)$ time, which is $O(n + \mu /\eps)$. 
This can be expressed as both $O(n/\eps)$ and $O(n + 1/\eps^3)$.
\end{proof}

\subsection{Small Gap-Ratio (proof of Theorem~\ref{thm:small_gap_ratio})}
\label{sec:small_gap_ratio}
In this section, we assume that there exists an optimal pair $\Sigma^* = (\sigma_1^*, \sigma_2^*)$ for the parallel two-line-center problem whose gap ratio is at most $\rho$. 
Recall that $w^* = \max \{w(\sigma_1^*), w(\sigma_2^*) \}$.
We present an $O(n/\eps)$-time algorithm that outputs a pair of parallel slabs such that their union covers $S$ and the maximum of their widths is at most $(1+\eps)w^*$ under the assumption.

\begin{observation}
    \label{obs:gap_ratio}
    For a pair $\Sigma = (\sigma_1, \sigma_2)$ of parallel slabs  whose gap-ratio is at most $\rho$,
    it holds that $g(\Sigma) \le \frac{2\rho}{1-\rho} \cdot w_{\max}$ where $w_{\max} = \max \{w(\sigma_1), w(\sigma_2)\}$. 

\begin{proof}
    Let $g = g(\Sigma)$ and $W = W(\Sigma)$.
    By definition, we have $W = w(\sigma_1) + w(\sigma_2) + g \le 2w_{\max} + g$.
    Since the gap-ratio of $\Sigma$ is at most $\rho$, it holds that $\rho \ge g/W \ge g/(2w_{\max} + g)$.
    Therefore, $2\rho\cdot w_{\max} \ge (1-\rho)g$, which implies $g \le \frac{2\rho}{1-\rho}\cdot w_{\max}$.   
\end{proof}
    
\end{observation}

\begin{lemma}
\label{lem:small_gap_const}
We can compute $\widetilde w$ such that $w^* \le \widetilde w \le (2 + \frac{2\rho}{1-\rho} ) w^* $ in $O(n)$ time.
\end{lemma}

\begin{proof}
We first compute $w'$ such that $\width(S) \le w' \le 2\cdot \width(S)$ using a $2$-approximation algorithm for the width, 
which can be done in $O(n)$ time~\cite{chan2006faster}. 
Let $\widetilde w = w'/2$. 
We show that $\widetilde w$ satisfies the condition in the statement. Observe that $w^* \le \width(S)/2$ as a slab of $\width(S)$ enclosing $S$ directly gives us a pair of parallel slabs of width $\width(S)/2$ whose union covers $S$. 
Therefore, we have $w^* \le \width(S)/2 \le \widetilde w$. It holds that $\width(S) \le W(\Sigma^*) = g(\Sigma^*) + w(\sigma_1^*) + w(\sigma_2^*) \le g(\Sigma^*) + 2w^*$. Since $\widetilde w \le \width(S)$ and by Observation~\ref{obs:gap_ratio}, we obtain $\widetilde w \le (2 + \frac{2\rho}{1-\rho}) w^*$.
\end{proof}



\begin{lemma}
    \label{lem:parallel_angles}
    Let $p$ and $q$ be two points in $\mathbb R^2$ and let $\widetilde w > 0$ and $\eps> 0$ be real numbers.
    In $O(1/\eps)$ time, we can compute a set $\Gamma$ of $O(1/\eps)$ orientations, such that the following holds:
    
    For every finite set $P$ of points in $\mathbb R^2$ and for every pair of parallel slabs $\Sigma= (\sigma_1, \sigma_2)$, if 
    \begin{enumerate}[(1)]
        \item The gap-ratio of $\Sigma$ is at most $\rho$, 
        \item $P\subset \sigma_1 \cup \sigma_2$, 
        \item $w_{\max} \le \widetilde w \le c\cdot w_{\max}$ where $w_{\max} = \max \{w(\sigma_1), w(\sigma_2) \}$ and $c = 2 + \frac{2\rho}{1-\rho}$,
        \item $p, q \in \sigma_1 \cup \sigma_2$ and $d(p,q) \ge \diam(P)/2$,
    \end{enumerate}
    then there exists an orientation $\gamma \in \Gamma$ and  a pair of parallel slabs $\Sigma' = (\sigma_1', \sigma_2')$, such that
    \begin{enumerate}[(1)]
        \item $P \subset \sigma_1' \cup \sigma_2'$,
        \item $w'_{\max} \le (1+ \eps) \cdot w_{\max}$ where $w'_{\max} = \max \{w(\sigma'_1), w(\sigma'_2) \}$,
        \item both $\sigma'_1$ and $\sigma_2'$ have the orientation $\gamma$.
    \end{enumerate}
\end{lemma}

\begin{proof}
    We remark that the idea of this proof is based on the approach in~\cite{agarwalTwo}.
    Without loss of generality, assume that $p$ and $q$ lie on a horizontal line. 
    Let $D = d(p, q)$. 
    Let $\theta$ be the orientation with $0\le \theta \le \pi/2$ such that $\sin\theta = \min \{1, c \widetilde w/{D} \}$ for $c= 2 + \frac{2\rho}{1-\rho}$.
    Let $\delta =  c_\delta \eps$ for some constant $0 < c_\delta < 1$. 
    Let $\Gamma = \{\gamma_i = (i- \lceil \frac{1}{\delta}\rceil)\cdot \delta\theta  \mid 0 \le i \le 2\lceil \frac{1}{\delta}\rceil \}$.
    Observe that $\Gamma$ contains uniformly placed orientations in the range $[-\theta, \theta]$, and $|\Gamma| = O(1/\eps)$. (Here we slightly abuse the notation by redefining orientations to lie in $[-\pi/2, \pi/2]$ instead of $[0, \pi]$.)
    In the following, we prove that $\Gamma$ satisfies the conditions in the statement.

    Let $P$ be set of points in $\mathbb R^2$ and $\Sigma = (\sigma_1, \sigma_2)$ be a pair of parallel slabs, where $P$ and $\Sigma$ satisfy the four conditions in the statement.
    Let $\alpha$ be the orientation of the slabs in $\Sigma$. 
    For simplicity, assume that $0\le \alpha \le \pi/2$. The other case can be shown in a symmetric way. 
    Observe that $\sin \alpha \le W(\Sigma)/D$, even when $\alpha = \pi/2$ (see the right triangle having $\overline{pq}$ as one of its sides in Figure~\ref{fig:parallel_angle_lemma}).
    We have $W(\Sigma) = g(\Sigma) + w(\sigma_1) + w(\sigma_2) \le g(\Sigma) + 2w_{\max} $. 
    By Observation~\ref{obs:gap_ratio} and as $w_{\max} \le \widetilde w$, $W(\Sigma) \le ( 2 + \frac{2\rho}{1-\rho} ) w_{\max}  \le c \widetilde w$.
    Therefore,
    we have $\sin \alpha \le c \widetilde w /D$.
    Observe that if $ c \widetilde w/ D < 1$, then $\sin \theta =  c \widetilde w/ D$. If $ c \widetilde w/ D \ge 1$, then $\sin \theta =1$. In both cases, it holds that $\sin \alpha \le \sin \theta$. Therefore, we have $\alpha \le \theta$.

     Let $\gamma$ be the largest orientation in $\Gamma$ such that $\gamma \le \alpha$. 
     Since $ \alpha \le  \theta$, we have $\alpha - \gamma < \delta \theta$. 
     Let $R_1$ be the minimum bounding box of $P \cap \sigma_1$ with orientation $\alpha$. See Figure~\ref{fig:parallel_angle_lemma}.
     Denote by $o$ and $t$ the lowest and highest vertices of $R_1$, respectively. 
    Now consider a slab $\sigma_1'$ of orientation $\gamma$ whose two bounding lines pass through $o$ and $t$. 
    Clearly, $P\cap \sigma_1 \subset \sigma'_1$. 
    We next show that $w(\sigma_1') \le (1+\eps)w_{\max}$.

    Let $b$ be a vertex of $R_1$ other than $o$ and $t$. 
    Consider the line $\ell$ orthogonal to $\sigma_1'$ that passes through $b$. 
    The line $\ell$ intersects the two bounding lines of $\sigma_1'$. 
    Observe that one intersection point lies inside $\sigma_1$, while the other lies outside. 
    Let $a$ and $c$ denote these two intersection points, with $a$ being the one inside $\sigma_1$. 
    Then $w(\sigma_1') = \overline{ab} + \overline{bc} \le w(\sigma_1) + \overline{bc} \le w_{\max} + \overline{bc} $.

     We have $\overline{bc} = \overline{oc}\cdot \tan(\alpha - \gamma) 
    \le \overline {oc} \cdot \tan (\delta \theta)$. 
    Observe that $\overline{oc} \le \overline{ob} \le \diam(P)$. 
    As $\diam(P)/2 \le D$, we have $\overline{oc} \le 2D$. 
    We now consider the following two cases.

     \textbf{Case 1. $c \widetilde w / D \le 1/2$. } That is, $\theta \le  \pi/6$.
     In this case, since $\delta < 1$ and $\theta \le \pi/6$, we have 
     $\tan(\delta \theta) 
    \le \delta \tan \theta$.
    Note that $\delta  \tan \theta  = \delta  \sin \theta/ \cos \theta $
    $\le (2/\sqrt 3)  \delta  \sin\theta = (2/\sqrt 3)    c \delta \widetilde w /D \le (2/\sqrt 3) c^2 \delta  w_{\max} /D  $.
    Since $\overline{oc}  \le 2D$, it follows that $\overline{bc} \le (4/\sqrt 3) c^2 \delta  w_{\max}$.

    \textbf{Case 2. $c\widetilde w / D > 1/2$. }
    Since $\delta \theta \le \pi/3$ (because $\theta \le \pi/2$ and we assume $\delta \le 2/3$), 
    we have
    $\tan(\delta \theta) 
    = \sin (\delta \theta) / \cos (\delta \theta) 
    \le 2 \cdot \sin (\delta \theta) \le 2 \delta \theta \le \pi \delta $.
    Also, $\overline{oc} \le 2D \le 4 c \widetilde w \le 4 c^2 w_{\max}$.  
    Combining these bounds gives $\overline{bc} \le 4\pi  c^2 \delta w_{\max}$. 
    
    By choosing $\delta = \min \{ 2/3, \sqrt 3 \eps / (4 c^2) ,  \eps / 4\pi c^2\}$, it holds that $\overline{bc} \le \eps w_{\max}$ in both cases. 
    Therefore, we have $w(\sigma_1') \le (1+\eps)w_{\max}$.

    In a symmetric way, we obtain $\sigma_2'$ of orientation $\gamma$ and it can be shown that $P\cap \sigma_2 \subset \sigma_2'$ and $w(\sigma_2')  \le (1+\eps)w_{\max}$.
\end{proof}

\begin{figure}[ht]
    \centering
    \includegraphics[width=0.57 \textwidth]{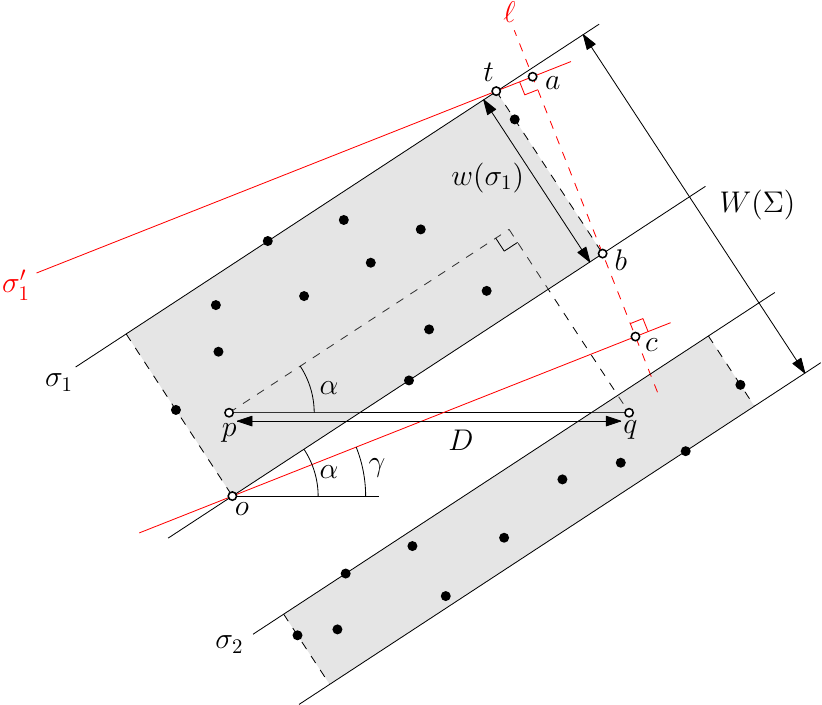}
    \caption{A pair $(\sigma_1, \sigma_2)$ of parallel slabs with orientation $\alpha$ and two points $p, q \in \sigma_1 \cup \sigma_2$ are shown.
    The black points represent the points in $P$.
    The two gray regions represent the minimum bounding boxes of $P\cap \sigma_1$ and $P \cap \sigma_2$ with orientation $\alpha$, respectively. 
    }
    \label{fig:parallel_angle_lemma}
\end{figure}

\begin{observation}
    \label{obs:par_fixed_ori}
    Given an orientation $\theta$, we can find a pair of parallel slabs of orientation $\theta$ such that their union covers $S$ and their maximum width is minimized in $O(n)$ time.
\end{observation}

\begin{proof}
    We can compute the minimum-width slab $\sigma$ of orientation $\theta$ that encloses $S$ in $O(n)$ time by finding two extreme points in the direction orthogonal to $\theta$.
    Let $\ell_1$ and $\ell_2$ be the two bounding lines of $\sigma$, and let $m$ be its center line.
    See Figure~\ref{fig:parallel_fixed_ori}.
    
    Let $m_1$ be the line obtained by translating $m$ toward $\ell_1$ until it touches the first point of $S$.
    Symmetrically, let $m_2$ be the line obtained by translating $m$ toward $\ell_2$ until it touches the first point of $S$.
    Both $m_1$ and $m_2$ can be computed in $O(n)$ time.
    
    Let $\sigma_1$ be the slab having $\ell_1$ and $m_1$ as its bounding lines, and let $\sigma_2$ be the slab having $\ell_2$ and $m_2$ as its bounding lines. 
    Observe that $(\sigma_1, \sigma_2)$ is a pair of slabs of orientation $\theta$ whose union encloses $S$ while minimizing the maximum of their widths. 
\end{proof}

\begin{figure}[ht]
    \centering
    \includegraphics[width=0.4 \textwidth]{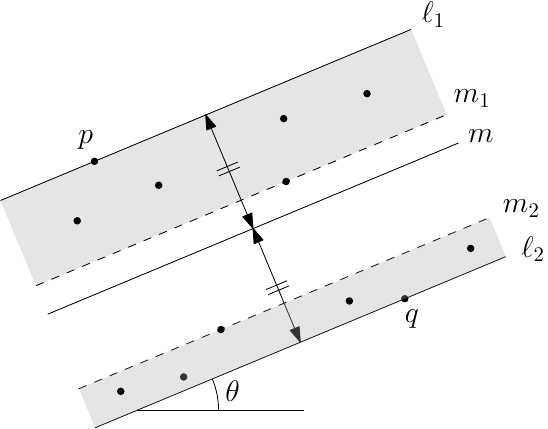}
    \caption{Black points indicate the point in $S$. Observe that $p$ and $q$ are the extreme points in the direction orthogonal to $\theta$. The lines $\ell_1, \ell_2, m_1, m_2$, and $m$ are shown as described above.
    Observe that the pair $(\sigma_1, \sigma_2)$, represented by the two gray regions in the figure, is a pair of slabs of orientation $\theta$ whose union encloses $S$ while minimizing the maximum of their widths. }
    \label{fig:parallel_fixed_ori}
\end{figure}

Now we prove Theorem~\ref{thm:small_gap_ratio} using the above observations and lemmas. 
We first compute a pair of points $p, q$ in $S$ where $d(p, q) \ge \diam(S)/2$, which can be done in $O(n)$ time.
By Lemma~\ref{lem:small_gap_const}, we can compute $\widetilde w$ with  $w^* \le \widetilde w \le (2 + \frac{2\rho}{1-\rho} ) w^* $ in $O(n)$ time.

Observe that $S$ and $\Sigma^*$ satisfy the four conditions in the statement of Lemma~\ref{lem:parallel_angles} for $(p, q)$ and the value $\widetilde w$ computed above.
Therefore, in $O(1/\eps)$ time, we can compute a set $\Gamma$ of $O(1/\eps)$ orientations, such that there exists a pair of parallel slabs $\Sigma' = (\sigma_1', \sigma_2')$  of orientation $\gamma \in \Gamma$ where $S\subset \sigma_1'\cup \sigma_2'$ and $\max \{w(\sigma_1'), w(\sigma_2') \} \le (1+\eps) w^*$.

Therefore, for every orientation $\gamma \in \Gamma$, we compute a pair of parallel slabs of orientation $\gamma$ whose union covers $S$ while minimizing the maximum of their widths. 
Among these $O(1/\eps)$ pairs of slabs, we output the one that minimizes the maximum of their widths.
By Observation~\ref{obs:par_fixed_ori}, it takes $O(n/\eps)$ time in total.

\subsection{Large Gap-Ratio (proof of Theorem~\ref{thm:large_gap_ratio})}
\label{sec:large_gap_ratio}
In this section, we assume that there exists an optimal pair $\Sigma^* = (\sigma_1^*, \sigma_2^*)$ for the parallel two-line-center problem whose gap ratio is at least $\rho$. 
Recall that $w^* = \max \{w(\sigma_1^*), w(\sigma_2^*) \}$.
We present an algorithm that outputs a pair of parallel slabs such that their union covers $S$ and the maximum of their widths is at most $(1+\eps)w^*$ under~the~assumption. 
It runs in $O(n + \mu \log \mu)$ time, where $\mu = \min \{n, 1/\eps^2 \}$, which can be expressed as both $O(n \log (1/\eps))$ and $O(n + (1/\eps^2) \log (1/\eps))$. 

We first recall an algorithm of~\cite{chung2025parallel} that computes a pair of parallel slabs with a gap-ratio at least $\rho$ whose union cover $S$ while minimizing their maximum width.

\begin{theorem}
    [\cite{chung2025parallel}]
    \label{thm:large_gap_prev}
    Given a set $P$ of $n$ points in the plane and a constant $0< \rho < 1$, we can do the following in $O(n\log n)$ time. If there exists a pair of parallel slabs such that (1) their union covers $P$, (2) each slab contains at least one point in $P$, and (3) its gap-ratio is at least $\rho$, then we report one such pair that minimizes the maximum of their widths.
    Otherwise, we report NO.
\end{theorem}

By our assumption that $w^* > 0$ (see Appendix~\ref{sec:w_0} for the case where $w^*=0$), we have $S \cap \sigma_1^* \neq \emptyset$ and $S \cap \sigma_2^* \neq \emptyset$.  
Otherwise, if one of the slabs contains no point of $S$, we can divide the other slab into two halves, obtaining a pair of parallel slabs of width $w^*/2$ whose union still covers $S$, contradicting the optimality.

Let $Q\subseteq S$ be an $\eps$-certificate of $S$ (see Section~\ref{sec:eps_certificate}). Then we have the following.

\begin{lemma}
    \label{lem:Q_nonempty}
    $Q\cap \sigma_1^* \neq \emptyset$ and $Q\cap \sigma_2^* \neq \emptyset$.
\end{lemma}
\begin{proof}
    Suppose $Q\cap \sigma_1^* = \emptyset$. That is, $Q\subseteq S \cap \sigma_2^*$. 
    Let $\Sigma= (\tau_1, \tau_2)$ be the pair of slabs obtained by dividing $\sigma_2^*$ into two halves. 
    See Figure~\ref{fig:large_gap_Q}(a). 
    That is $w(\tau_1) = w(\tau_2) = w(\sigma_2^*)/2$. 
    Observe that $Q\subset \tau_1 \cup \tau_2$. 
    Let $\Sigma'= (\tau_1', \tau_2')$ be the pair of slabs where $\tau_1'$ and  $\tau_2'$ are the $\eps$-expansions of $\tau_1$ and $\tau_2$, respectively.
    Let $D$ be the amount by which each bounding line is shifted by the $\eps$-expansion.
    That is, $D= (1+\eps)w(\tau_1)/2 = (1+\eps)w(\tau_2)/2 \le (1+\eps)w^*/4$.
    See Figure~\ref{fig:large_gap_Q}(b).
    
    Let $p$ a point in $S\cap \sigma_1^*$.
    Without loss of generality, assume that $\tau_1$ is closer to $p$ than $\tau_2$. 
    We now show that $\tau_1'$ does not cover $p$ (which also implies that $\tau_2'$ does not cover $p$), contradicting that $Q$ is an $\eps$-certificate of $S$.
    We have $ d(p, \tau_1)\ge g(\Sigma^*) \ge \rho W(\Sigma^*) \ge \rho w^*$.
    Since $\rho \ge 1/2 >  (1+\eps)/4$, it follows that $ d(p, \tau_1)\ge \rho w^* > (1+\eps)w^*/4$. 
    As shown above, $D \le (1+\eps)w^*/4$.
    Therefore, $\tau_1$ does not cover $p$. 

    Symmetrically, we can show that $Q\cap \sigma_2^* = \emptyset$ leads to a contradiction.
\end{proof}

\begin{figure}[ht]
    \centering
    \includegraphics[width=0.8 \textwidth]{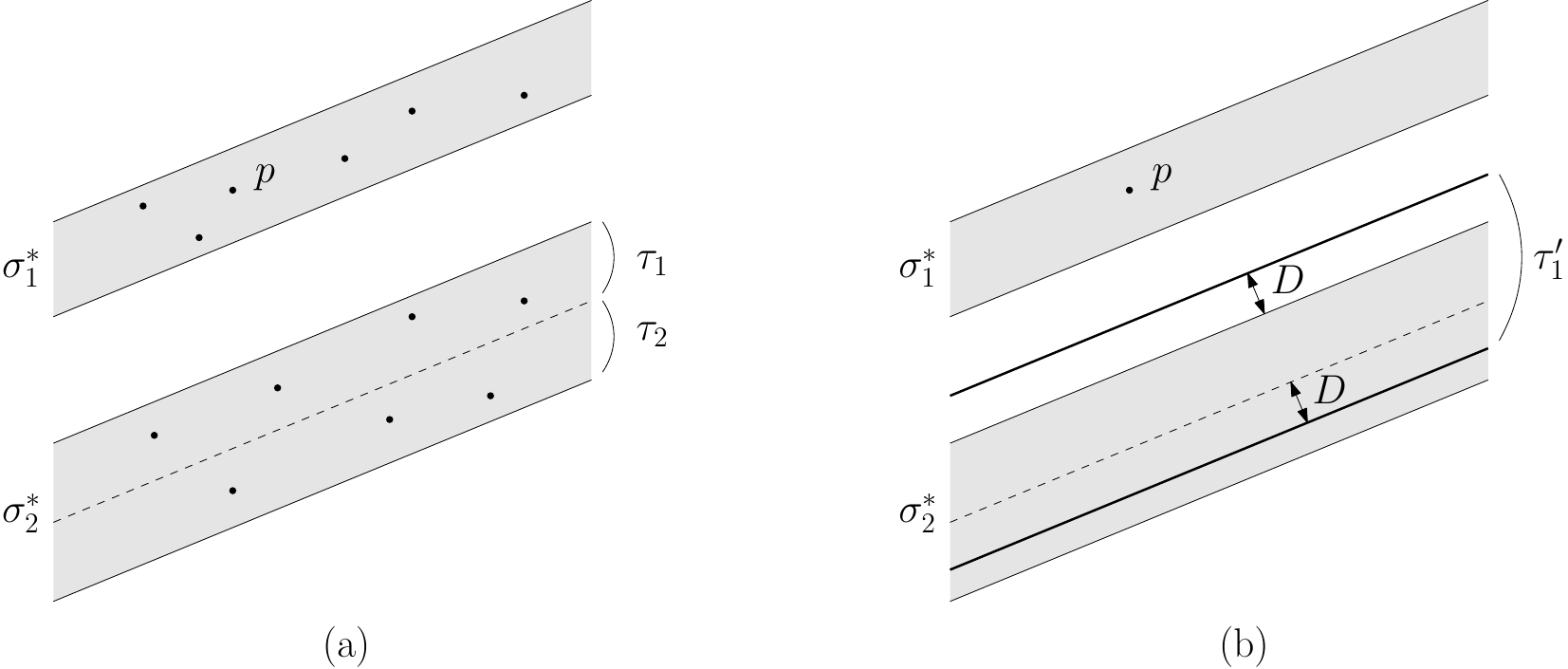}
    \caption{(a) Black points indicate the point in $S$. The dashed line indicates the center line of~$\sigma_2^*$. (b) The $\eps$-expansion $\tau_1'$ of $\tau_1$ is shown.}
    \label{fig:large_gap_Q}
\end{figure}

Now we prove Theorem~\ref{thm:large_gap_ratio}.
By Lemma~\ref{lem:Q_nonempty}, $\Sigma^*=(\sigma_1^*, \sigma_2^*)$ satisfies the three conditions in the statement of Theorem~\ref{thm:large_gap_prev} with respect to $Q$. 
Therefore, we obtain a pair of parallel slabs by running the algorithm in Theorem~\ref{thm:large_gap_prev} on $Q$. 
Let $\Sigma=(\sigma_1, \sigma_2)$ be the pair. (If the assumption on the gap ratio is not satisfied, the algorithm in Theorem~\ref{thm:large_gap_prev} may return NO. In that case, we stop the algorithm and also return NO.) 
Observe that $\max\{w(\sigma_1), w(\sigma_2) \} \le w^*$ since $Q\subseteq S$. 
Let $\Sigma'= (\sigma_1', \sigma_2')$ be the pair of slabs where $\sigma_1'$ and  $\sigma_2'$ are the $\eps$-expansions of $\sigma_1$ and $\sigma_2$, respectively.
By definition, $\sigma_1'\cup \sigma_2'$ covers $S$ and $\max \{w(\sigma_1'), w(\sigma_2') \} \le (1+\eps)w^*$. 
Therefore, $\Sigma'= (\sigma_1', \sigma_2')$ forms a $(1+\eps)$-approximate solution. 

Computing $Q\subseteq S$ of size $O(1/\eps^2)$ can be done in $O(n)$ time by Theorem~\ref{thm:computing_eps_certificate}.
It takes $O(|Q|\log |Q|)$ time to obtain $(\sigma_1', \sigma_2')$ by Theorem~\ref{thm:large_gap_prev}. 
Let $\mu = \min \{n, 1/\eps^2\}$, then we have $|Q| = O(\mu)$. 
Therefore, the total running time is $O(n + \mu \log \mu)$, which can be expressed as both $O(n \log (1/\eps))$ and $O(n + (1/\eps^2) \log (1/\eps))$.

\section{Future Work}

We have presented approximation algorithms for several variants of the two-line-center problem. An open question is whether the dependence on $\eps$ can be further refined for each variant. For the exact algorithms of the problems, no nontrivial lower bounds are currently known, except for the one-fixed-orientation case, for which we provide an optimal algorithm.

Each of the four problems admits natural generalizations with respect to the dimension $d$ and the number of lines $k$ (our problems correspond to $d=2$ and $k=2$). For higher dimensions ($d \ge 3$) or a larger number of lines ($k \ge 3$), very few exact or approximation algorithms are known. 
For the general $k$-line-center problem, the authors in~\cite{agarwal2005approximation} gave a $(1+\eps)$-approximation algorithm whose expected running time is $O(n\log n)$ where the constant depends on $d$, $k$, and $\eps$. 
We leave open the question of whether efficient exact and approximation algorithms can be obtained for these and other variants of the two-line-center problem in higher dimensions or with more than two lines.

Another natural variant of the two-line-center problem is the following. Given a set $S$ of $n$ points in the plane and a fixed orientation difference $\theta$, the goal is to find a pair of slabs whose orientations differ by $\theta$, whose union covers $S$, and that minimizes the maximum of their widths. 
An exact algorithm with running time $O(n^2 \alpha(n)\log n)$ is known for this problem~\cite{ahn2024constrained}, where $\alpha(n)$ denotes the inverse Ackermann function. 
To the best of our knowledge, no approximation algorithm is known for this variant.
It is unclear if our approach based on an anchor pair of an optimal pair of slabs and a constant factor solution (as in Lemma~\ref{lem:angles} and Lemma~\ref{lem:parallel_angles}) can be extended to this variant.
This is because rotating one slab requires rotating the other as well. We addressed this issue in the parallel two-line-center problem by considering two different cases, depending on the gap-ratio of an optimal pair of slabs (see Sections~\ref{sec:small_gap_ratio} and~\ref{sec:large_gap_ratio}).

\bibliography{paper}



\appendix

\section{Proof of Theorem~\ref{thm:10apx}}
\label{sec:linear_10apx}

In this section, we give a $10$-approximation algorithm for the general two-line-center problem that runs in $O(n)$ time. 
As mentioned in the main text, we use the following $O(n)$-time algorithm that incrementally computes a $6$-approximation\footnote{The algorithm in~\cite{chan2006faster} maintains a $(5+\eps)$-approximation of the minimum width for a fixed $\eps>0$. For simplicity, we maintain a $6$-approximation.} of the width of a point set.

\sixapx*

The algorithm from Theorem~\ref{thm:6apx} maintains the following throughout the insertion of points: a set $V$ of points, a real value $w$, and the first inserted point $o$.
Without loss of generality, assume that $o$ is the origin.
It initializes $V = \emptyset$ and $w = 0$.
Whenever a new point $p$ is inserted, $V$ and $w$ are updated according to the procedure described in Algorithm~\ref{alg:insert}.

In~\cite{chan2006faster}, it is shown that the size of $V$ remains $O(1)$ at any time, and that the final value of $w$, denoted by $w_f$, after inserting a set $P$ of points, satisfies $\width(P) \le w_f \le 6 \cdot \width(P)$.

\begin{algorithm}
\caption{Processing a new point $p$ (See the proof of Theorem A.3 in \cite{chan2006faster}. Let $\delta>0$ be a sufficiently small constant.)}
\label{alg:insert}
\begin{algorithmic}[1]
\Function{insert}{$p$}
\State $w = \max \{w, 6\cdot \width(V \cup \{ o, p\} ) \}$ 
\If{for all $v \in V$ , $\|p\| > (1+\delta) \|v\|$ or $\angle pov > \pi/2 + \delta$}
    \State Insert $p$ to $V$
    \State Remove all points $v$ from $V$ such that $\|v\| \le \delta \|p\|$ and $\angle pov \le \pi/2 + \delta$
\EndIf
\EndFunction
\end{algorithmic}
\end{algorithm}

\medbreak
By Lemma~\ref{lem:anchor_pair}, we can find a set $F$ of at most $11$ pairs of points, at least one of which is an anchor pair of $\Sigma^*$ in $O(n)$ time. 
For the moment, assume that we have an anchor pair $(p, q)$ of $\Sigma^*$. (We will repeat the following algorithm for every pair in $F$.)
Without loss of generality, assume that $p, q \in \sigma_1^*$ and $d(p, q) \ge \diam((\sigma_1^*\setminus \sigma_2^* ) \cap S)/4$.

\paragraph*{Our Algorithm}


For $1 \le i \le n$, let $f(i) = w(\sigma(p, q; p_i)) = 2d(\ell_{pq}, p_i)$ and let $g(i)$ denote a $6$-approximate width of $S \setminus \sigma(p, q; p_i)$ obtained by Theorem~\ref{thm:6apx}. The value $g(i)$ depends on the order in which points are inserted, but this does not affect the correctness of our algorithm.
Observe that $f(i)$ is an increasing function, whereas $g(i)$ is not necessarily decreasing (see Figure~\ref{fig:bs_idea}(a)).
Nevertheless, we proceed as in \textsf{\nameref{par:naive}}.

\begin{figure}[ht]
\centering
\includegraphics[width=0.98 \textwidth]{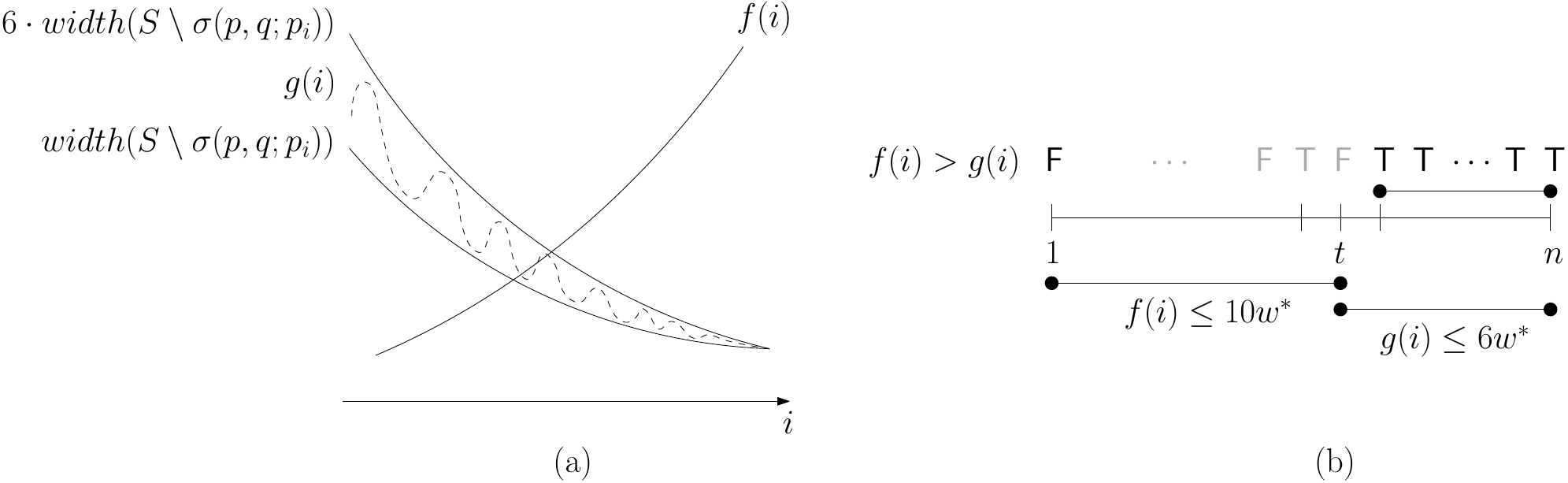}
\caption{
(a) A rough illustration of $f(i)$ and $g(i)$. $f(i)$ is an increasing function, whereas $g(i)$ is not necessarily monotone. 
(b) An illustration of Lemma~\ref{lem:fg}. 
}
\label{fig:bs_idea}
\end{figure}

In the following, we describe our algorithm, which is also summarized in Algorithm~\ref{alg:binary_search}.
We then prove its correctness in Lemma~\ref{lem:bs_correctness} and analyze its running time in Theorem~\ref{thm:10apx}.

\textbf{Data Structures and Invariants.} 
Throughout the algorithm, we maintain the following: a range of indices $R \subseteq [1, n]$, a set $Q \subseteq S$, a value $w$, and a set $V$.
Initially, we set $R = [1, n]$, $Q = S$, $w = 0$, and $V = \emptyset$.
We also set $o = p_n$, which is used in Algorithm~\ref{alg:insert}.
Without loss of generality, assume that $o$ is the origin.

We maintain these variables so that the following invariants hold:
(1)~$Q = \{ p_i \mid i \in R \}$, and
(2)~$w$ and $V$ are the value and the set obtained by inserting the points\footnote{Observe that $p_n$ is always contained in $S\setminus\sigma(p, q; p_e)$ except for when $e = n$. This is why we set $o = p_n$.} of $S \setminus \sigma(p, q; p_e)$ using Algorithm~\ref{alg:insert},
where $e$ is the largest index in $R$.
That is, $\width(S\setminus\sigma(p, q; p_e)) \le w \le 6\cdot \width(S\setminus\sigma(p, q; p_e))$.

\textbf{Main loop.} 
We perform a binary search on the index range $R$.
At each step, let $R = [s, e]$, set $m$ to be the median of $R$, and find the point $p_m$.
We then compare $f(m)$ and $g(m)$.
The value $f(m)$ can be computed directly from the coordinates of $p$, $q$, and $p_m$.
To obtain $g(m)$, we need to insert the points of $S \setminus \sigma(p, q; p_m)$ using Algorithm~\ref{alg:insert}.
However, since we already have $w$ and $V$ obtained by inserting the points in $S \setminus \sigma(p, q; p_e)$, and
$S \setminus \sigma(p, q; p_e) \subseteq S \setminus \sigma(p, q; p_m)$,
we can obtain $g(m)$ by inserting only the additional points $(S\setminus\sigma(p, q; p_m))\setminus (S\setminus\sigma(p, q; p_e))=\{p_i | m < i \le e \}$.

If $f(m) \ge g(m)$, we set $R = [s, m]$.
In this case, since the largest index in $R$ changes, we update $w$ and $V$ to maintain the invariant, using the value and the set obtained when computing $g(m)$ (see line~\ref{line:update_wV} of Algorithm~\ref{alg:binary_search}).
If $f(m) < g(m)$, we set $R = [m, e]$ and $w$ and $V$ remain unchanged.
In both cases, we update $Q$ to reflect the new range $R$.

\textbf{Base case.} When $R$ contains at most three indices, we select an index $r \in R$ that minimizes $\max \{ f(r), g(r) \}$.
Observe that for every $i\in R$, $f(i)$ is computed directly and $g(i)$ can be computed by inserting at most three points using Algorithm~\ref{alg:insert}. 
Finally, we output the pair of slabs, where one slab is $\sigma(p, q; p_r)$ and the other is a slab of width at most $g(r)$ that covers $S \setminus \sigma(p, q; p_r)$. 
We note that the latter slab can be obtained by $o, w$, and $V$ (see the proof of Theorem A.3 in~\cite{chan2006faster}).

\begin{algorithm}
\caption{A linear-time $10$-approximation algorithm}
\label{alg:binary_search}
\begin{algorithmic}[1]
\State $R \gets [1, n]$
\State $Q \gets S$
\State $w \gets 0, \; V \gets \emptyset, \; o \gets p_n$ 
\While{$|R| > 3$ (equivalently, $|Q| > 3$)}
    \State $m\gets \lfloor (s + e)/2 \rfloor$, where $R = [s, e]$
    \State Find the point $p_{m}$ in $Q$
    \State $f \gets 2d(\ell_{pq}, p_{m})$ \label{line:f}
    \State $g, V' \gets $ a $6$-approximate width of $S\setminus \sigma(p,q;p_{m})$ and a set obtained by Theorem~\ref{thm:6apx}\label{line:g}
    \If{$f \ge  g$} \label{line:bs_condition}
        \State $R \gets [s, m]$ \label{line:if_R}
        \State $(w, V)\gets (g, V')$ \label{line:update_wV}
        \State $Q \gets \{p_{s}, p_{s +1}, \ldots, p_{m} \}$ \label{line:if_Q}
    \Else
        \State $R \gets [m, e]$ \label{line:else_R}
        \State $Q \gets \{p_{m}, p_{m +1}, \ldots, p_{e} \}$ \label{line:else_Q}
        
    \EndIf
\EndWhile

\State Find an index $r\in R$ which minimizes $\max\{f(r), g(r)\}$\label{line:final_index} \label{line:final_R}
\State Report the pair of slabs corresponding to $p_r$: one is $\sigma(p, q; p_r)$, and the other is a slab of width at most $g(r)$ that covers $S \setminus \sigma(p, q; p_r)$.

\end{algorithmic}
\end{algorithm}

To prove the correctness in Lemma~\ref{lem:bs_correctness}, we first show Lemma~\ref{lem:true_false},~\ref{lem:fg}, and Corollary~\ref{cor:fgt},~\ref{cor:exists_t}.

\begin{restatable}{lem*}{truefalse}
    \label{lem:true_false}
    For $R = [s, e]$ at any time, it holds that $f(s)< g(s)$ and $f(e) \ge g(e)$.
\end{restatable}

\begin{proof}
    Let $I(i)$ be the boolean value indicating whether $f(i) \ge g(i)$.
    We will prove that $I(s) = \false$ and $I(e) = \true$ for $R = [s, e]$ throughout the algorithm.
    We prove this by induction on the iteration of the while loop.
    Initially, $R = [1, n]$.
    Observe that $f(1) = 2d(\ell_{pq}, p_1) = 0$.
    Since $g(1) > 0$, we have $f(1) < g(1)$, so $I(1) = \false$.
    On the other hand, $g(n) = 0$ because $S \setminus \sigma(p, q; p_n) = \emptyset$, which implies $I(n) = \true$.
    
    Now, let $R = [s_i, e_i]$ be the range of indices at the beginning of the $i$-th iteration of the while loop, and let $m_i = \lfloor (s_i + e_i)/2 \rfloor$.
    Assume as the induction hypothesis that $I(s_i) = \false$ and $I(e_i) = \true$.
    Let $s_{i+1}$ and $e_{i+1}$ be the indices defining the range after this iteration.
    
    If $I(m_i) = \true$, then $[s_{i+1}, e_{i+1}] = [s_i, m_i]$ (see line~\ref{line:if_R}).
    Otherwise, $[s_{i+1}, e_{i+1}] = [m_i, e_i]$ (see line~\ref{line:else_R}).
    In both cases, it holds that $I(s_{i+1}) = \false$ and $I(e_{i+1}) = \true$.
\end{proof}

Let $t$ be the largest index such that $f(t) = 2d(\ell_{pq}, p_t) \le 10w^*$.
As $d(\ell_{pq}, p_1) = 0$, $t$ exists.

\begin{restatable}{lem*}{fg}
    \label{lem:fg}
    The following statements hold. See Figure~\ref{fig:bs_idea}(b).
    \begin{enumerate}[(i)]
        \item For every $i$ with $1\le i \le t$, we have $f(i) \le 10w^*$. 
        \item For every $i$ with $t\le i \le n$, 
    we have $g(i) \le 6w^*$. 
        \item For every $i$ with $t < i \le n$, we have $f(i) > g(i)$. 
    \end{enumerate}
\end{restatable}

\begin{proof}
    \textbf{\textsf{(i)}} This follows directly from the definition of $t$ and the fact that $f(i)$ is an increasing function.

    \textbf{\textsf{(ii)}}  Let $k$ be the smallest index such that $\sigma(p, q; p_k)$ covers the points in $S\cap (\sigma_1^* \setminus \sigma_2^*)$. 
    By Lemma~\ref{lem:10apx_existence}, we have $2d(\ell_{pq}, p_k) \le 10w^*$. 
    Therefore, $k\le t$ by the definition of $t$. 
    Observe that $k \le i$ for every $i$ with $t\le i\le n$. 

    By the definition of $g(i)$, we have $g(i) \le 6\cdot \width(S \setminus \sigma(p, q; p_i))$.
    Since $k \le i$, we also have $\width(S\setminus \sigma(p,q;p_{i})) \le \width(S\setminus \sigma(p,q;p_{k})) $.
    Therefore, $g(i) \le 6\cdot \width(S \setminus \sigma(p, q; p_k))$.
    Finally, $\width(S \setminus \sigma(p, q; p_k)) \le w^*$, otherwise $w(\sigma_2^*) \le w^*$ would be contradicted because $\sigma(p, q; p_k)$ covers all points in $S \cap (\sigma_1^* \setminus \sigma_2^*)$.
    Therefore, $g(i) \le 6 w^*$ for every $i$ with $t \le i \le n$.
    
    \textbf{\textsf{(iii)}}  By the definition of $t$, we have  $10w^* < 2d(\ell_{pq}, p_i) = f(i)$ for every $i$ with $t < i \le n$. 
    From {\textsf{(ii)}}, we know that $g(i) \le 6 w^*$ in the same range.
    Thus, for all $i$ with $t < i \le n$, we have $g(i) \le 6w^* \le 10w^* < f(i)$. 
\end{proof}

By combining Lemma~\ref{lem:fg}(i) and (ii), we obtain the following corollary.

\begin{corollary}
    \label{cor:fgt}
    $\max \{f(t) , g(t)\} \le 10w^*$.
\end{corollary}

Let $R_f$ be the final range $R$ obtained after the while loop. 

\begin{restatable}{cor*}{existst}
    \label{cor:exists_t}
    There exists $j\in R_f$ such that $j \le t$. 
\end{restatable}

\begin{proof}
Suppose, to the contrary, that all indices in $R_f$ are at least $t+1$. 
By Lemma~\ref{lem:fg}(iii), we have $f(i) > g(i)$ for every $i \in R_f$.
In particular, if $R_f = [s, e]$, then $f(s) > g(s)$, which contradicts Lemma~\ref{lem:true_false}.
\end{proof}

\begin{restatable}{lem*}{bscorrectness}
    \label{lem:bs_correctness}
    For the final index  $r\in R_f$ obtained in base case, $\max \{ f(r), g(r) \} \le 10 w^*$. 
\end{restatable}

\begin{proof}
    If $t \in R_f$, then $\max \{ f(r), g(r) \} \le \max \{ f(t), g(t) \} \le 10 w^*$ by Corollary~\ref{cor:fgt}.  

    If $t \notin R_f$, then Corollary~\ref{cor:exists_t} implies that $R_f \subseteq [1, t-1]$.  
    By Lemma~\ref{lem:true_false}, there exists an index $b \in R_f$ such that $f(b) \ge g(b)$.  
    Hence, $\max \{ f(b), g(b) \} = f(b) \le f(t) \le \max \{ f(t), g(t) \} \le 10 w^*$.  
    Therefore, we have $\max \{ f(r), g(r) \} \le \max \{ f(b), g(b) \} \le 10 w^*$.
\end{proof}

Now we show the following theorem.

\tenapx*

\begin{proof}
    By Lemma~\ref{lem:bs_correctness}, our algorithm computes a pair of slabs whose widths are at most $10w^*$ and whose union covers $S$.
    We now analyze its running time.

    Initializing the variables $R$, $Q$, $w$, $V$, and $o$ clearly can be done in $O(n)$ time.

    Next, we show that each iteration of the while loop, with the current index range $R = [s, e]$, takes $O(|R|)$ time.
    For the median $m$ of $R$, we find $p_m$ using a linear-time selection algorithm in $Q$, which can be done in $O(|Q|)=O(|R|)$ time~\cite{Blum1973Time_selection}.

    The value $f(m)$ is obtained directly.
    To compute $g(m)$, we apply Algorithm~\ref{alg:insert} to the points $(S\setminus\sigma(p, q; p_m))\setminus (S\setminus\sigma(p, q; p_e))=\{p_i | m < i \le e \}$
    updating the current $w$ and $V$.
    By Theorem~\ref{thm:6apx}, this takes $O(e-m)=O(|R|)$ time.

    Since the size of $R$ decreases by a factor of two in each iteration, the total time over all iterations of the while loop is $O(n + n/2 + n/4 + \cdots) = O(n)$.

    After the while loop ends, we find an index $r\in R$ which minimizes $\max\{f(r), g(r)\}$.
    Observe that for every $i\in R$, $f(i)$ is obtained directly and $g(i)$ can be computed by inserting at most three points using Algorithm~\ref{alg:insert}. 
    Therefore, it takes $O(1)$ time for finding $r$. 
    We obtain one slab $\sigma(p, q;p_r)$ and it remains to compute the other slab, which has width at most $g(r)$ and covers $S \setminus \sigma(p,q;p_r)$.
    In the proof of Theorem~A.3 in~\cite{chan2006faster}, it is shown that this can be done by  computing the minimum-width slab enclosing $V\cup \{ o\}$ and expanding it by a factor determined by $\delta$ and $w$.
    Since $|V|=O(1)$ at all times, this step takes $O(|V|\log |V|)=O(1)$ time.

    We resolve the assumption that we have an anchor pair $(p, q)$ of $\Sigma^*$. We repeat the entire algorithm for every pair in $F$, using it as $(p, q)$, which adds only a constant factor to the running time.
    Observe that for each pair in $F$, we obtain a pair of slabs whose union covers $S$. 
    Among these pairs, we output the one with the smallest maximum width.  
\end{proof}

\section{Proof of Theorem~\ref{thm:one_fixed}}
\label{sec:one_fixed_appendix}

We first sort the points of $S$ so that $s_1, \ldots, s_n$ are ordered in decreasing order of their $y$-coordinates.
Let $W$ denote the set of all pairwise differences between the $y$-coordinates of points in $S$.
Since $W$ can be represented as a sorted matrix~\cite{frederickson1982complexity}, we can identify two consecutive values $w_0, w_1 \in W$ such that $w_0 < w^* \le w_1$ in $O(n \log n)$ time, by combining our decision algorithm described in Section~\ref{sec:one_fixed_decision} with an efficient selection algorithm for sorted matrices~\cite{frederickson1984generalized}.

\textbf{Notation} We use the same notation $S_a$, $S_b$, and $\{\ell_i, \theta_i, a_i, b_i \mid i = 1, 2\}$ as in Section~\ref{sec:one_fixed_decision}.
For any $w > 0$, we denote by $\sigma_w$ a horizontal slab of width $w$.

\textbf{Assumption}
Throughout this section, we assume that there exists a position of $\sigma_{w^*}$ such that $\width(S \setminus \sigma_{w^*}) \le w^*$  (i.e., $\width(S_a\cup S_b) \le w^*$) and $S_a$ dominates $S_b$ in that position.
The other case can be handled symmetrically by reflecting $S$ across a horizontal line.

We first introduce two lemmas from \cite{ahn2024constrained}.
Observe that the sequence of changes in the sets $S_a$ and $S_b$ that occur as $\sigma_w$ moves downward is the same for any $w_0 < w < w_1$.
Let $X$ denote the set of all pairs $(A, B)$ such that $A \subseteq S$ and $B \subseteq S$ appear as $S_a$ and $S_b$, respectively, during the sweep of the plane by $\sigma_w$ for $w_0 < w < w_1$.

\begin{lemma}[Lemma 9 of \cite{ahn2024constrained}]
    \label{lem:constrained_1}
    $w^* = \min \{w_1,  \min_{(S_a, S_b) \in X} \{\width(S_a\cup S_b) \}\}$.
\end{lemma}

Lemma~\ref{lem:constrained_1} implies that by tracking $\width(S_a \cup S_b)$ as $\sigma_w$ moves downward (for any $w_0 < w < w_1$), we can determine $w^*$ by comparing its minimum value with $w_1$.
Lemma~\ref{lem:constrained_2} and Corollary~\ref{cor:one_fixed_idea} further refine this observation, implying that $w^*$ can be obtained more directly from the procedure of the decision algorithm.

\begin{lemma}[Lemma 11 of \cite{ahn2024constrained}]
    \label{lem:constrained_2}
    For $(S_a,  S_b) \in X$ where $S_a$ dominates $S_b$, it holds that $\width(S_a \cup S_b) < w_1$
    if and only if $\width_{[\theta_1, \theta_2]}(S_a) < w_1$. 
    Moreover, if $\width_{[\theta_1, \theta_2]}(S_a) < w_1$, 
    then $\width(S_a \cup S_b)  =\width_{[\theta_1, \theta_2]}(S_a)$.
\end{lemma}

By combining Lemma~\ref{lem:constrained_1} and~\ref{lem:constrained_2}, we can obtain the following corollary.  

\begin{corollary}
    \label{cor:one_fixed_idea}
    $w^* = \min\{ w_1,  \min \{\width_{[\theta_1, \theta_2]}(S_a) \mid (S_a, S_b) \in X, \; S_a \textrm{ dominates } S_b\} \}$. 
\end{corollary}

By Corollary~\ref{cor:one_fixed_idea}, 
once we find a position of $\sigma_w$ where $S_a$ dominates $S_b$ and $\width_{[\theta_1, \theta_2]}(S_a)$ is minimized,
we can compute $w^*$ by comparing this value with $w_1$.
We now show that we can find the position during the procedures in the decision algorithm presented in Section~\ref{sec:one_fixed_decision}.

We first fix $w_0 < w < w_1$. 
Recall that in the decision algorithm, we test whether $\width_{[\theta_1, \theta_2]}(S_a) < w$ while $\sigma_w$ moves downward, maintaining the following data structures
(see~\emph{\nameref{par:ds_mainloop}} in Section~\ref{sec:one_fixed_decision}).

\begin{itemize}
    \item the dynamic data structures for $\conv(S_a)$ and $\conv(S_b)$ (from Lemma~\ref{lem:dyn_conv})
    \item the tangent lines $\ell_1$ and $\ell_2$, with their corresponding points $a_1$, $b_1$, $a_2$, and $b_2$
    \item the two antipodal pairs of $S_a$ with respect to the orientations $\theta_1$ and $\theta_2$, respectively
\end{itemize}

As $\sigma_w$ moves downward, we will maintain an additional value $w_{\min}$, which represents the minimum value of $\width_{[\theta_1, \theta_2]}(S_a)$ observed so far. 
Initially, we set $w_{\min} = \infty$.
Observe that $S_b$ dominates $S_a$ from the beginning until a certain moment, after which $S_a$ dominates $S_b$ until the end. 
Until $S_a$ begins to dominate $S_b$ (which can be easily detected using Observation~\ref{obs:domination}), we just follow the decision algorithm without updating $w_{\min}$. 
At the moment when $S_a$ starts to dominate $S_b$, we set  $w_{\min} = \width_{[\theta_1, \theta_2]}(S_a)$. 
This can be done in $O(n)$ time by following the antipodal pairs in the data structure that we maintain for $\conv(S_a)$.

In the following, we show how $w_{\min}$ can be updated at each event during the execution of the decision algorithm.
At the end of the procedure, the minimum between $w_1$ and the final value of $w_{\min}$ is equal to $w^*$ by Corollary~\ref{cor:one_fixed_idea}.

We note that the authors of~\cite{ahn2024constrained} also used the observation from Corollary~\ref{cor:one_fixed_idea} to perform a sweep of the slab $\sigma_w$.
However, their approach requires an additional data structure to compute $\width_{[\theta_1, \theta_2]}(S_a)$ at each event, resulting in $O(\log^3 n)$ time per event.

\paragraph*{Insertion into $S_a$}
Let $p$ be the point inserted into $S_a$, and let $S_a' = S_a \cup \{p\}$.
We define $\{\ell_i', \theta_i', a_i', b_i' \mid i=1, 2 \}$ analogously for $S_a'$ and $S_b$. 
That is $\ell_1'$ and $\ell_2'$ are the right and the left outer common tangents of $S_a'$ and $S_b$, respectively. 
Let $w_{\min}'$ denote the updated value after inserting $p$.
Therefore, $w_{\min}' = \min \{w_{\min}, \width_{[\theta_1', \theta_2']}(S_a')\}$.

Suppose that $p$ lies on or between $\ell_1$ and $\ell_2$. 
Then both $\ell_1$ and $\ell_2$ remain unchanged, and consequently so do $\theta_1$ and $\theta_2$.
Since every slab $\sigma_\theta(S_a)$ for $\theta \in [\theta_1, \theta_2]$ contains $p$, it follows that $\width_{[\theta_1, \theta_2]}(S_a) = \width_{[\theta_1, \theta_2]}(S_a')$. 
As $w_{\min} \le \width_{[\theta_1, \theta_2]}(S_a)$, $w_{\min}' = w_{\min}$. 

Now assume that $p$ lies to the right of $\ell_1$ (the case where $p$ lies to the left of $\ell_2$ can be handled symmetrically).
In this case, $\ell_2' = \ell_2$, while $\ell_1' \neq \ell_1$.
Observe that $\theta_2' = \theta_2$ and $\theta_1' < \theta_1$.
Now we show the following. 

\begin{restatable}{clm*}{optins}
    \label{claim:opt_ins}
    $w_{\min}' = \min \{w_{\min}, \width_{[\theta_1', \theta_1]}(S_a') \}$
\end{restatable}

\begin{claimproof}
    Let $a = \width_{[\theta_1', \theta_1]}(S_a')$, 
    $b = \width_{[\theta_1, \theta_2]}(S_a')$, and
    $c = \width_{[\theta_1', \theta_2]}(S_a')$. 
    By definition, $c = \min \{ a, b\}$. 
    As $w'_{\min} = \min \{w_{\min}, c \}$, 
    we have $w'_{\min} = \min \{w_{\min}, a, b \}$. 
    Observe that $w_{\min} \le  \width_{[\theta_1, \theta_2]}(S_a) \le b$ by the definition of $w_{\min}$ and $S_a \subset S_a'$. 
    Therefore, we have $w'_{\min} = \min \{w_{\min}, a \} =  \min \{w_{\min}, \width_{[\theta_1', \theta_1]}(S_a')\}$.
\end{claimproof}

The above claim allows us to consider only the range $[\theta_1', \theta_1]$. 
Recall that, in \emph{Testing the width condition} step of the decision algorithm for insertions, the antipodal pairs of $S_a'$ corresponding to the orientation range $[\theta_1', \theta_1]$ are visited in order.
Therefore, $\width_{[\theta_1', \theta_1]}(S_a')$ can be obtained directly during this procedure.

\paragraph*{Deletion from $S_b$}
Let $p$ be the point deleted from $S_b$, and let $S_b' = S_b \setminus \{p\}$.
We define $\{\ell_i', \theta_i', a_i', b_i' \mid i=1, 2 \}$ analogously for $S_a$ and $S_b'$. 
That is $\ell_1'$ and $\ell_2'$ are the right and the left outer common tangents of $S_a$ and $S_b'$, respectively. 
Let $w_{\min}'$ denote the updated value after deleting $p$.
Therefore, $w_{\min}' = \min \{w_{\min}, \width_{[\theta_1', \theta_2']}(S_a)\}$.

Suppose that $p$ lies strictly between $\ell_1$ and $\ell_2$. 
Then both $\ell_1$ and $\ell_2$ remain unchanged, and consequently so do $\theta_1$ and $\theta_2$.
Therefore, $w_{\min}' = w_{\min}$ by definition.

Now assume that $p$ lies on or to the right of $\ell_1$ (the case where $p$ lies on or to the left of $\ell_2$ can be handled symmetrically).
In this case, $\ell_2' = \ell_2$, while $\ell_1' \neq \ell_1$.
Observe that $\theta_2' = \theta_2$ and $\theta_1' < \theta_1$.
Therefore, $w_{\min}' = \min \{w_{\min}, \width_{[\theta_1', \theta_2]}(S_a)\} =  \min \{w_{\min}, \width_{[\theta_1', \theta_1]}(S_a), \width_{[\theta_1, \theta_2]}(S_a)\}$.
Note that $w_{\min} \le \width_{[\theta_1, \theta_2]}(S_a) $. 
Therefore, $w_{\min}' = \min \{w_{\min}, \width_{[\theta_1', \theta_1]}(S_a) \}$.

The above observation allows us to consider only the range $[\theta_1', \theta_1]$. 
Recall that, in \emph{Testing the width condition} step of the decision algorithm for deletions, the antipodal pairs of $S_a$ corresponding to the orientation range $[\theta_1', \theta_1]$ are visited in order.
Therefore, $\width_{[\theta_1', \theta_1]}(S_a)$ can be obtained directly during this procedure.

Now we can conclude the following theorem. 

\onefixed*

\section{Lower Bound for the One-Fixed-Orientation Problem}
\label{sec:lowerbound}

\begin{restatable}{lem*}{lowerbound}
    The one-fixed-orientation two-line-center problem requires $\Omega (n\log n)$ time in the algebraic computation tree model. 
\end{restatable}

\begin{proof}
    Given a set $P$ of $n$ points in the plane, it is known that computing $\width(P)$ requires $\Omega(n \log n)$ time~\cite{Houle_width, lee1986geometric}.
    We give a reduction from the width problem to the one-fixed-orientation two-line-center problem.
    We show that we can construct a set $Q$ of $n+2$ points in $O(n)$ time such that the following holds:
    Let $(\sigma_h^*, \sigma_t^*)$ be a pair of slabs whose union covers $Q$, where $\sigma_h^*$ is horizontal, and the maximum of their widths is minimized.
    Let $w^*(Q) = \max\{ w(\sigma_h^*), w(\sigma_t^*)\}$.
    Then $w^*(Q) = \width(P)$.

    We compute the largest difference between the $x$-coordinates of the points in $P$, and denote it by $L_1$. Similarly, we compute the largest difference between the $y$-coordinates of the points in $P$, and denote it by $L_2$.
    Let~$L = \max\{L_1, L_2\}$.
    
    Consider a square of side length $L$ that covers $P$, and denote it by $R$.
    Let $q_1, q_2 \in \mathbb{R}^2$ be two points on a horizontal line such that $d(q_1, q_2) = 7L$, $d(R, \ell_{q_1q_2}) = L$, and $R$ and $\overline{q_1q_2}$ are vertically center-aligned (see Figure~\ref{fig:lowerbound}).
    Let $Q = P \cup \{q_1, q_2\}$.

    It is clear that $w^*(Q) \le \width(P)$, as we can construct a pair of slabs of width $\width(P)$ where one covers $P$ and the other covers $\{q_1, q_2\}$. 
    Thus, it remains to show that $w^*(Q) \ge \width(P)$. 
    Assume to the contrary that $w^*(Q) < \width(P)$. 
    Since $q_1$ and $q_2$ lie on a horizontal line, we have two cases: 
    (1) $q_1, q_2 \in \sigma_h^*$, or 
    (2) $q_1, q_2 \notin \sigma_h^*$.

    Consider the case where $q_1, q_2 \in \sigma_h^*$.
    By the assumption and the definitions, we have $w(\sigma_h^*) \le w^*(Q) < width (P) \le L$ and 
    $ L \le d(p, \ell_{q_1q_2})$ for every $p\in P$. 
    Since $q_1, q_2 \in \sigma_h^*$ and $w(\sigma_h^*) < d(p, \ell_{q_1q_2})$ for every $p\in P$, 
    it follows that $P\cap \sigma_h^* = \emptyset$.  
    Therefore, $\sigma_t^*$ should cover $P$ which implies that $w^*(Q) \ge w(\sigma_t^*) \ge \width(P)$. This contradicts our assumption that $w^*(Q) < \width(P)$.

    Now, consider the case where $q_1, q_2 \notin \sigma_h^*$. 
    By the assumption and the definitions, $w(\sigma_h^*) \le w^*(Q) < \width(P) \le L_2$. 
    Therefore, there exists a point $r \in P$ not contained in $\sigma_h^*$. 
    This implies that $q_1, q_2, r \in \sigma_t^*$. 
    It follows that $w(\sigma_t^*) \ge \width(\{ q_1, q_2, r \})$. 
    Since $\overline{q_1q_2}$ is the longest side in the triangle $\triangle q_1q_2r$, we have $\width(\{ q_1, q_2, r \}) = d(r, \ell_{q_1q_2})$.
    As $d(r, \ell_{q_1q_2}) \ge L$, it follows that $w(\sigma_t^*) \ge L$. 
    Consequently, $w^*(Q) \ge w(\sigma_t^*) \ge L \ge \width(P)$, which contradicts our assumption.
\end{proof}

\begin{figure}[ht]
        \centering
        \includegraphics[width=0.57\textwidth]{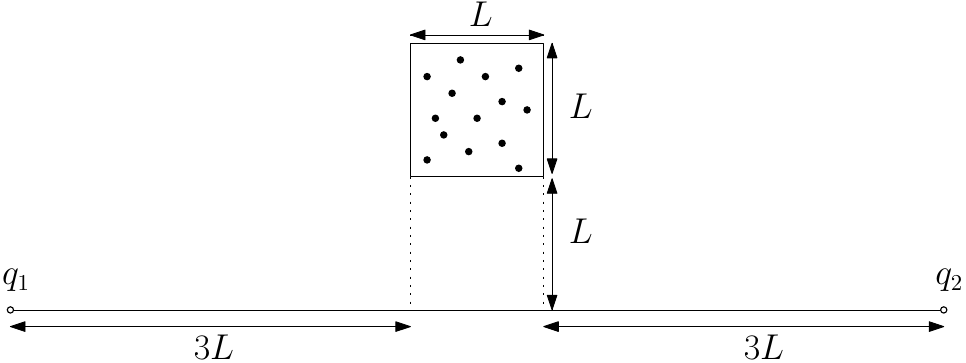}
        \caption{The square of side length $L$ covers the points. }
        \label{fig:lowerbound}
\end{figure}

\section{Ruling Out \texorpdfstring{$w^*=0$}{w*=0} in Linear Time}
\label{sec:w_0}

\paragraph*{General}
Choose two arbitrary points $p$ and $q$ in $S$, and find a point $r \in S$ that does not lie on $\ell_{pq}$. 
If no such point exists, then all points of $S$ lie on $\ell_{pq}$, which implies $w^* = 0$.
Assume such a point $r$ exists. 
If $w^* = 0$, then at least one slab (of width $0$) in an optimal solution must be $\ell_{pq}$, $\ell_{pr}$, or $\ell_{qr}$.
We test all three cases. 
For example, to test whether there exists an optimal pair of slabs where one of them is $\ell_{pq}$, we check whether all points of $S$ not lying on $\ell_{pq}$ lie on a common line.
This can be done by choosing two arbitrary points $u, v \in S\setminus \ell_{pq}$ and checking if all points in $S\setminus \ell_{pq}$ lie on $\ell_{uv}$. 
This entire procedure takes linear time.

\paragraph*{One fixed orientation}
By rotating the coordinate system, we assume without loss of generality that the given orientation $\theta$ is $0$.

Choose two arbitrary points $p$ and $q$ in $S$, and find a point $r \in S$ that does not lie on $\ell_{pq}$. 
If no such point exists, then all points of $S$ lie on $\ell_{pq}$, which implies $w^* = 0$.
Assume such a point $r$ exists. 
If $w^* = 0$, then at least one slab (of width $0$) in an optimal solution must be $\ell_{pq}$, $\ell_{pr}$, or $\ell_{qr}$.
We test all three cases. 
We give an example for testing whether there exists an optimal pair of slabs where one of them is $\ell_{pq}$.
If $\ell_{pq}$ is horizontal, we check whether all points of $S$ not lying on $\ell_{pq}$ lie on a common line.
This can be done by choosing two arbitrary points $u, v \in S\setminus \ell_{pq}$ and checking if all points in $S\setminus \ell_{pq}$ lie on $\ell_{uv}$. 

If $\ell_{pq}$ is not horizontal, we check whether all points of $S$ not lying on $\ell_{pq}$ lie on a horizontal line.
This can be done by choosing an arbitrary point $u \in S\setminus \ell_{pq}$ and checking if all points in $S\setminus \ell_{pq}$ lie on the horizontal line passing through $u$. 
This entire procedure takes linear time.

\paragraph*{Two fixed orientations}
By rotating the coordinate system, we assume without loss of generality that one of the given orientations is horizontal. Thus, the two orientations can be taken as $0$ and $\theta$, where $0\le \theta < \pi$. Let $(\sigma_h^*, \sigma_t^* )$ be an optimal pair of slabs where $\sigma_h^*$ is horizontal and $\sigma_t^*$ has orientation $\theta$.

Choose an arbitrary point $p \in S$, and then find another point $q \in S$ with a different $y$-coordinate from $p$. If no such point exists, then all points of $S$ lie on a horizontal line, which implies $w^*=0$.
If such $q$ exists, then for $w^* = 0$ it must be the case that $p$ or $q$ (or possibly both) lies on $\sigma_t^*$.
Let $\ell_p$ be the line of orientation $\theta$ passing through $p$, and define $\ell_q$ analogously for $q$.
If $w^* = 0$, then $\sigma_t^*$ must be either $\ell_p$ or $\ell_q$.
We test both cases by checking whether all points of $S$ not lying on $\ell_p$ (resp., $\ell_q$) lie on a common horizontal line. 
The procedure above takes linear time in total.

\paragraph*{Parallel}
Choose two arbitrary points $p$ and $q$ in $S$, and find a point $r \in S$ that does not lie on $\ell_{pq}$. 
If no such point exists, then all points of $S$ lie on $\ell_{pq}$, which implies $w^* = 0$.
Assume such a point $r$ exists. 
If $w^* = 0$, then at least one slab (of width $0$) in an optimal solution must be $\ell_{pq}$, $\ell_{pr}$, or $\ell_{qr}$.
We test all three cases. 
For example, to test whether there exists an optimal pair of slabs where one of them is $\ell_{pq}$, we check whether all points of $S$ not lying on $\ell_{pq}$ lie on the line parallel to $\ell_{pq}$ that passes through $r$. 
This entire procedure takes linear time.

\end{document}